\documentclass[preprint,12pt,nopreprintline]{elsarticle}

\usepackage{amsmath}
\usepackage{amssymb}
\usepackage{amsthm}
\newtheorem{theorem}{Theorem}[section]
\newtheorem{corollary}{Corollary}[section]
\newtheorem{definition}{Definition}[section]
\usepackage{booktabs} 
\usepackage{siunitx}
\usepackage{graphicx}
\usepackage{subcaption} 
\usepackage{array}
\usepackage{multirow} 
\usepackage{lineno}
\usepackage{hyperref}

\begin{document}

\begin{frontmatter}

\title{Spatiotemporal decoupled physics-informed Stone-Weierstrass neural operator for long-time prediction of time-dependent parametric PDEs}

\cortext[cor1]{Corresponding author}
\cortext[cor2]{Co-corresponding author}

\author{Shan Ding}

\author{Yongfu Tian\corref{cor2}}
\ead{tianyf@mail.tsinghua.edu.cn}

\author{Lang Qin}
\author{Hongxiang Ma}
\author{Guofeng Su}

\author{Rui Yang\corref{cor1}}
\ead{ryang@tsinghua.edu.cn}

\affiliation{organization={School of Safety Science},
            addressline={Tsinghua University}, 
            postcode={100084}, 
            state={Beijing},
            country={China}}

\begin{abstract}
Driven by rapid advances in artificial intelligence and modern GPU computing capabilities, deep learning methods based on the optimization paradigm have provided new pathways to solve spatiotemporal physical problems, whose mathematical core lies in solving partial differential equations (PDEs). As an emerging class of function-space learning methods, neural operators (NOs) have exhibited great potential in efficient PDE solving. However, existing mainstream neural operator frameworks suffer from critical bottlenecks when modeling time-dependent PDEs over long time horizons, including accuracy degradation, insufficient stability, high training costs, and excessive memory consumption, which severely limit their practical deployment. To address these challenges in long-time prediction with neural operators, we propose a novel spatiotemporally decoupled physics-informed neural operator architecture, termed the physics-informed Stone-Weierstrass neural operator (PI-SWNO). The design is theoretically grounded in the decoupling paradigm combining time-invariant spatial basis functions with time-varying evolution coefficients, as well as the Stone-Weierstrass approximation theorem. By encoding spatial and temporal information via two separate subnetworks, the framework structurally mitigates the accumulation of errors over extended time intervals. Furthermore, we introduce a time-marching batch-wise sampling strategy to resolve the memory bottleneck of full-range modeling over extended time spans, ensuring continuity and convergence of full-time-domain solutions.
\end{abstract}



\end{frontmatter}


\section{Introduction}
\label{sec1}

Spatiotemporal physical problems are fundamental topics in science and engineering, most commonly formulated mathematically as partial differential equations (PDEs). Many real-world spatiotemporal physical problems lack analytical solutions. Traditional numerical solvers, including the finite difference method (FDM)~\cite{VARGAS2022}, the finite element method (FEM)~\cite{KERGRENE2016}, and the finite volume method (FVM)~\cite{BUCHMULLER2016}, are capable of generating high-precision approximate solutions. However, they only support single-configuration simulations per run and suffer from high computational costs, making them unable to meet industrial requirements for rapid multi-configuration simulation and real-time inference.

The classical numerical iteration paradigms rely on CPU-centric computing and have grown prevalent with advances in CPU technology. Meanwhile, driven by advances in artificial intelligence (AI) and modern GPU hardware, deep learning under the optimization paradigm has emerged as an alternative approach for solving spatiotemporal physical problems~\cite{RAISSI2019, li2023transformer, MUCKE2021}. Neural operators (NOs) learn end-to-end mappings between infinite-dimensional function spaces. Once trained, they deliver millisecond-scale inference across diverse configurations, positioning them as a valuable and promising direction in the field of AI for Science. As GPU-centric computing becomes the hardware backbone, neural operators are set to become a core enabling technology under this computational paradigm.

Current neural operator research falls into two major paradigms: purely data-driven and physics-informed. Data-driven neural operators, represented by DeepONet~\cite{23Lu2021} and the Fourier neural operator (FNO)~\cite{li2021fourier}, learn mappings from input functions of PDEs---including initial conditions, boundary conditions, source terms, and physical parameters---to spatiotemporal solution functions using large datasets of paired input-output samples. They yield fast and high-precision PDE solutions when data are abundant. Yet in engineering practice, high-fidelity spatiotemporal labeled data often take days or weeks to generate via traditional numerical methods, sharply restricting the engineering applicability of purely data-driven approaches.

Physics-informed neural operators are built on the theoretical framework of physics-informed neural networks (PINNs)~\cite{RAISSI2019}. They embed the governing equations, initial conditions, and boundary conditions as soft constraints into the loss function, enabling end-to-end learning of PDE solutions without extensive labeled data. This makes them adaptable to data-scarce engineering scenarios and establishes them as a core research direction in neural operator learning. Representative works include physics-informed neural operator (PINO)~\cite{li2023physics} and physics-informed DeepONet (PI-DeepONet)~\cite{2021wangpideeponet}.

Nevertheless, both data-driven and physics-informed neural operators suffer from significant accuracy degradation in long-time prediction tasks. Existing neural operator approaches for long-time PDE solutions can be divided into two main paradigms: auto-regressive (AR) and full-range (FR)~\cite{7Mandl2025}. The AR paradigm centers on training a single-step neural operator that maps the solution at time \(t\) to the solution at time \(t+\Delta t\). Full-time-domain evolution fields are obtained by repeatedly invoking this model via a rolling time window. Its error pattern is characterized by the propagation and accumulation of single-step prediction errors across iterative time steps. To address the inherent bottleneck of iterative error accumulation caused by causal dependency, the core challenge of this paradigm for solving long-time PDEs is to improve the accuracy of the single-step neural operator. Some studies have adopted a multi-step neural operator predictor, yet such methods remain fundamentally auto-regressive.

In contrast, the FR modeling paradigm directly approximates the end-to-end mapping from input conditions to the complete spatiotemporal physical solution through a neural operator model with a fixed parameter space. The solution error is mainly attributed to both the function-space expressiveness of the neural operator and fitting errors induced by non-convex optimization. However, as proven by the non-decreasing theorem of fitting error for fixed-parameter neural operators in Section~\ref{wucha}, the global minimum fitting error of a fixed-parameter neural operator for the solution of time-dependent PDEs is inherently non-decreasing as the time interval expands. Therefore, the core challenge of the FR paradigm for solving long-time PDEs is to effectively control the growth rate of the model’s approximation error with expanding time intervals. There are generally two strategies: one is to improve neural network architectures for stronger function-space expressiveness under a fixed parameter budget; the other is to mitigate fitting errors from non-convex optimization via advanced optimization algorithms, physics-informed regularization, and adaptive sampling strategies.

This paper focuses on the fundamental problem of the FR modeling paradigm: controlling the growth rate of model approximation error with the expansion of time interval. As one of the most representative physics-constrained architectures in operator learning, PI-DeepONet strictly follows the universal approximation theorem on Banach spaces. It offers strong multi-configuration generalization, complex input adaptability, and solvability in unlabeled data settings, making it the most widely adopted physics-informed neural operator framework in engineering applications. We therefore employ PI-DeepONet as our base architecture to explore methods for enhancing the long-time solution capability of fixed-parameter neural operators. The classic PI-DeepONet structure and its various mainstream variants generally adopt spatiotemporally coupled encoding in their trunk networks, where spatial and temporal coordinates are concatenated and fed into the same network to generate spatiotemporally coupled basis functions. This results in a full coupling between the spatial distribution characteristics and the temporal evolution characteristics of the solution.

In contrast, unsteady PDEs widely used in engineering (including heat conduction equations, convection-diffusion equations, wave equations, etc.) exhibit an inherent spatiotemporally decoupled dynamical property: the spatial modes of the system are time-invariant, while the time-domain evolution process is only reflected in the dynamic changes of the weight coefficients corresponding to each order of spatial modes over time. Typical examples include the natural modes of structural vibration, the spatial characteristic modes of heat conduction problems, and the orthogonal basis of flow fields. Based on this physical property of unsteady systems, classical solution methods --- from method of separation of variables level to the Galerkin method~\cite{Arnold2001} and the mode superposition method~\cite{VILLADSEN1995} at the numerical discretization level --- all follow the core decoupling paradigm of time-invariant spatial basis functions combined with time-varying evolution coefficients. This paradigm guarantees the accuracy of the solution while offering clear physical interpretability.

Motivated by this spatiotemporally decoupled physical property of unsteady systems and the universal approximation guarantees of the Stone-Weierstrass theorem, we propose the physics-informed Stone-Weierstrass neural operator (PI-SWNO). The main contributions of this work are as follows.

\begin{itemize}
	\item We establish the non-decreasing theorem of fitting error for fixed-parameter neural operators, rigorously proving that the global minimum fitting error of a fixed-parameter neural operator is non-decreasing as the time interval expands. This provides a key theoretical foundation and justification for developing neural operator methods toward long-time spatiotemporal modeling.
	\item Leveraging the intrinsic physical property of most unsteady PDEs --- time-invariant spatial modes combined with time-varying evolution coefficients --- and the universal approximation guarantee of the Stone-Weierstrass theorem, we propose a fully spatiotemporally decoupled PI-SWNO architecture. Built on the PI-DeepONet framework, we decompose its coupled spatiotemporal trunk network into independent spatial and temporal subnetworks. This enables independent encoding of input function features, spatial distribution patterns, and temporal evolution laws. By structurally breaking the rapid error propagation mechanism induced by spatiotemporally coupled encoding, this design effectively suppresses error accumulation in long-time solutions while also substantially boosting fitting accuracy and generalization for periodic and multi-scale unsteady PDEs.
	\item To resolve the engineering bottleneck of excessive GPU memory usage and impractical end-to-end training caused by dense full-time-domain sampling in long-time full-range modeling, we propose a time-marching batch-wise sampling training strategy. This strategy establishes a two-stage iterative paradigm of the sequential local optimization and global consistency calibration. It drastically reduces per-step training memory overhead while ensuring the continuity, consistency, and convergence of full-time-domain solutions, thus providing a practical engineering solution for stable end-to-end training of long-time unsteady PDEs.
\end{itemize}

The remainder of this paper is organized as follows. Section~\ref{sec2} reviews recent advances in neural operators and their variants for solving time-dependent PDEs. Section~\ref{sec3} introduces the proposed physics-informed Stone-Weierstrass neural operator and the time-marching batch-wise sampling strategy, along with the theoretical foundations. Section~\ref{sec4} verifies the effectiveness of the proposed method through seven benchmark problems and compares its performance with the conventional PI-DeepONet. Section~\ref{sec5} concludes the full paper and outlines the main conclusions.

\section{Related Work}
\label{sec2}

Data-driven neural operators, represented by DeepONet and FNO, form the fundamental architectures in operator learning. Existing studies have carried out a series of works around improving their expressive ability and scenario adaptability. He et al.~\cite{11He2024} designed multiple variants of DeepONet branch networks based on Fourier neural networks (FNN), gated recurrent units (GRU), and long short-term memory (LSTM) networks, which improved the prediction accuracy of transient heat conduction problems. Jin et al.~\cite{13Jin2022} proposed multiple-input operator networks (MIONet), extending the learning ability of operator models for multi-input operators on the product of Banach spaces and completing the theoretical proof of the corresponding universal approximation theorem. Diab et al.~\cite{4Diab2025} developed the temporal neural operator (TNO) for climate modeling and weather forecasting, using a U-Net architecture to extract data features for temperature forecasts over Europe and the globe in winter and summer, and learning mappings between temperature fields over fixed time windows. Li et al.~\cite{li2020graphkernel} constructed the graph neural operator (GNO) based on graph kernel networks, transforming integral operators (e.g., Nyström approximation) into message passing on graphs. This preserves continuous function modeling capability while achieving strong cross-geometry generalization. Raoni\'c et al.~\cite{Bogdan2023convolutional} proposed continuous-discrete equivalence as the fundamental design principle for operator learning and designed the convolutional neural operator (CNO), validating its robustness in zero-shot generalization tasks (e.g., unseen resolution, domain, discretization) on PDE benchmark datasets.

Due to the strong dependence of purely data-driven neural operators on labeled data, physics-informed neural operators have emerged as the mainstream research direction. A series of studies have incorporated physical constraints into the DeepONet architecture. Mandl et al.~\cite{7Mandl2025} modified DeepONet into a forced dual-output architecture and incorporated physical constraints into the loss function, allowing online evaluation of prediction quality through residual monitoring. Karumuri et al.~\cite{16KARUMURI2026} proposed the physics-informed latent neural operator, which consists of two coupled deep neural networks trained in an end-to-end manner: an implicit deep neural network for learning the low-dimensional representation of the solution and a reconstruction deep neural network for mapping this implicit representation back to the physical space. Combined with physical constraints, it significantly improves model generalization in data-scarce scenarios. Wang et al.~\cite{21wang2025} presented a pre-training framework based on latent neural operators, which is pre-trained on a mixed dataset covering various time-dependent PDEs. They employed physics-cross-attention (PhCA) as a universal transformation module and extracted representations of different physical systems in a shared latent space. The pre-trained framework is fine-tuned on a single PDE dataset to achieve general solution of diverse time-dependent PDEs. Koric et al.~\cite{41KORIC2023} conducted a comparative study between the classic data-driven DeepONet and PI-DeepONet for the two-dimensional steady-state heat conduction, verifying that the introduction of physical constraints can improve the solution accuracy in small-sample scenarios. Cho et al.~\cite{34cho2025} proposed physics-informed deep inverse operator networks (PI-DIONs), which realized the reconstruction of source terms and physical fields based on a small amount of measurement data, and proved that physics-informed neural operators can still maintain strong generalization ability in the full function space with very few training samples. Ding et al.~\cite{ding2026} proposed a hierarchical network structure to improve the DeepONet framework, which effectively enhanced the representation ability and convergence speed of the model.

Aiming at the accuracy degradation problem of neural operators in long-time PDE prediction tasks, the AR paradigm is the most widely used approach in existing research. This paradigm trains a neural operator model for single-time-step solution mapping and obtains evolved solutions from the full-time domain via rolling iteration. Most studies mainly focus on optimizing the suppression of iterative error accumulation: Lei et al.~\cite{6Lei2025} developed a one-step model for nonlinear wave equations based on FNO, using windowed inputs and conservation law regularization to reduce exponentially growing cumulative error to linear growth. Cho et al.~\cite{19CHO2026} proposed GraphDeepONet, an auto-regressive model based on graph neural networks, using GNN as the branch network of DeepONet to encode temporal information, achieving robust operator prediction in irregular grids. Other studies have integrated neural operators with traditional numerical analysis methods to break through the iterative bottleneck of AR frameworks: Nayak et al.~\cite{9Nayak2025} used DeepONet to learn the time derivative of the solution function \(\partial u / \partial t\) instead of the solution \(u\), and combined the Runge-Kutta Method to complete numerical time integration, achieving more stable time marching. In addition, state space models (SSMs)~\cite{15Dao2024} have also been introduced into operator learning for dynamic systems. Hu et al.~\cite{5Hu2025} established the theoretical connection between SSMs and various attention mechanisms, comparing the performance of Mamba~\cite{14Gu2024} with RNN~\cite{chung2014}, Transformer~\cite{zhou2021transformer} and other frameworks, and conducting numerical experiments on various dynamic systems and equation types for verification. Buitrago et al.~\cite{10Buitrago2025} proposed the memory neural operator (MemNO), which combines recent SSMs and FNO to model memory. MemNO significantly outperforms baseline models without memory when PDEs are provided at low resolution or contain observation noise during training and testing.

The FR paradigm uses a fixed-parameter neural operator to directly fit the end-to-end mapping from input conditions to complete spatiotemporal solutions, fundamentally avoiding iterative error accumulation in the AR paradigm. Existing optimization efforts fall into two categories: first, enhancing the function space expressive ability of the model via innovative network architectures. Micha{\l}owska et al.~\cite{8Michałowska2024} performed systematic combinatorial experiments on different types of neural operators, recurrent architectures, and training modes, comprehensively verifying performance across configurations. Hu et al.~\cite{12Hu2025} combined the spatial neural operator DeepONet with the temporal sequence model Mamba to propose the DeepOMamba model, validating its performance, training speed, and scalability on six benchmark PDE problems. Second, reducing non-convex fitting errors via improved optimization strategies, including adaptive sampling, physics-based regularization, data decomposition, and dimensionality reduction. Wang et al.~\cite{18WANG2025} combined the PI-DeepONet with the FEM through domain decomposition. The Schwarz alternating method is adopted to couple the subdomains of the two methods, with the pre-trained PI-DeepONet solving complex nonlinear regions and the traditional FEM handling the remaining domains. Chen et al.~\cite{20CHEN2025} proposed the tensor decomposition based neural operator TDMD-DeepONet and the dynamic mode decomposition based neural operator ETDMD-DeepONet, which offer stronger stability and accuracy than the original DeepONet. Chen et al.~\cite{17CHEN2025} proposed a data decomposition method enforcing local dependency. This method strictly limits the information usage range of the neural operator by decomposing the input data into small windows based on the local dependency region, thus ensuring local dependency while maintaining the expressive power of the model. Relevant studies have verified the stability potential of the FR modeling paradigm in long-time prediction, but existing work still has not solved the core problem that the fitting error of fixed-parameter neural operators increases non-decreasingly with the time interval.

In summary, existing research has thoroughly explored multi-configuration generalization, unlabeled data solvability, and long-time prediction stability of neural operators. However, two critical gaps remain: first, no universal method exists within the FR paradigm to effectively control the growth rate of approximation error for fixed-parameter neural operators with expanding time intervals. Second, mainstream physics-informed neural operator architectures rely on spatiotemporally coupled encoding, concatenating spatial and temporal coordinates as input to a single network, ignoring the inherent decoupling physical property of ``time-invariant spatial basis functions + time-varying evolution coefficients'' of unsteady PDE systems, which further aggravates the error divergence in long-time prediction. To address these issues, we propose the spatiotemporally decoupled PI-SWNO, designed to fundamentally improve the long-time solution capability of fixed-parameter neural operators. 

\section{Methodology}
\label{sec3}

\subsection{Problem definition}

Consider an unsteady system defined on a bounded closed domain $\Omega \subseteq \mathbb{R}^n$:
\begin{equation} \label{s0}
\begin{cases}
    \frac{\partial u}{\partial t}=\mathcal{N}(x,t,u,\nabla u, \nabla ^2 u, \cdots,f),  \quad  &(x,t)\in \Omega \times (0, T] \\
	\mathcal{B}(x,t,u, \nabla u)=g(x,t),  \quad  &(x,t) \in \partial \Omega \times (0,T] \\
	\mathcal{I}(x,t,u)=h(x,t),  \quad  &(x,t) \in \Omega \times \{0\}
\end{cases}
\end{equation}
where $u(x,t)$ is the solution function; $\mathcal{N}$ is the linear/nonlinear partial differential operator of the dynamic system; \(f( x,t)\) is the input function, which can be observations, source terms, external forces, equation coefficients, etc.; $\mathcal{B}$ and $\mathcal{I}$ denote the boundary condition operator and initial condition operator, respectively; \(g(x,t)\) is the input boundary condition function, and \(h(x,t)\) is the input initial condition function; $x\in \mathbb{R}^n$ is the spatial coordinate, and $t$ is the time coordinate; $\partial \Omega$ denotes the boundary of $\Omega$, and \(T\) represents the final time of the time domain.

The mapping from the input functions \(f,g,h\) to the corresponding solution function \(u\) is \(\mathcal{G}:(f,g,h) \mapsto u\), the well-posedness of which is guaranteed by the solution uniqueness of the solution of the unsteady system~\eqref{s0}. It is worth noting that in many cases, we only focus on the mapping from a single input function to \(u\), such as \(\mathcal{G}:f \mapsto u\). The neural operator method approximates the mapping \(\mathcal{G}:(f,g,h) \mapsto u\) with a neural network \(G_\Theta\), where the \(\Theta\) is the neural network parameter space, and the learning objective is 
\[
\min _{\Theta}\left\|\mathcal{G}_{\Theta}-\mathcal{G}\right\|^2.
\] 
The convergence of this approximation is weakly guaranteed by the universal approximation theorem for operators~\cite{1995Universal,23Lu2021}.

\subsection{Physics-informed DeepONet}

DeepONet is one of the most widely applied neural operator architectures in operator learning. The classic DeepONet architecture~\cite{23Lu2021} consists of separate branch networks and a trunk network (see Fig.~\ref{fig:DeepONet_framework}), and its mathematical expression is:
\begin{equation}
\mathcal{G}_\theta(f)(x, t) = \sum_{i=1}^n b_i(f) \cdot \phi_i(x, t),
\end{equation}
where \(\mathcal{G}_\theta\) is the DeepONet neural operator parameterized by \(\theta\), aiming to approximate the input-output solution mapping of time-dependent PDEs; \(f\) is the input function of the operator, usually the initial conditions, boundary conditions, source terms, etc. of the PDE; \(n\) is the number of modes of the basis functions; \(b_i(f)\) is the coefficient of the $i$-th mode, generated by the branch network encoding the input function \(f\); \(\phi_i(x, t)\) is the $i$-th spatiotemporally coupled basis function; \(LT\) denotes the linear transformation layer.

\begin{figure}[htbp]
	\centering
	\includegraphics[width=0.95\linewidth]{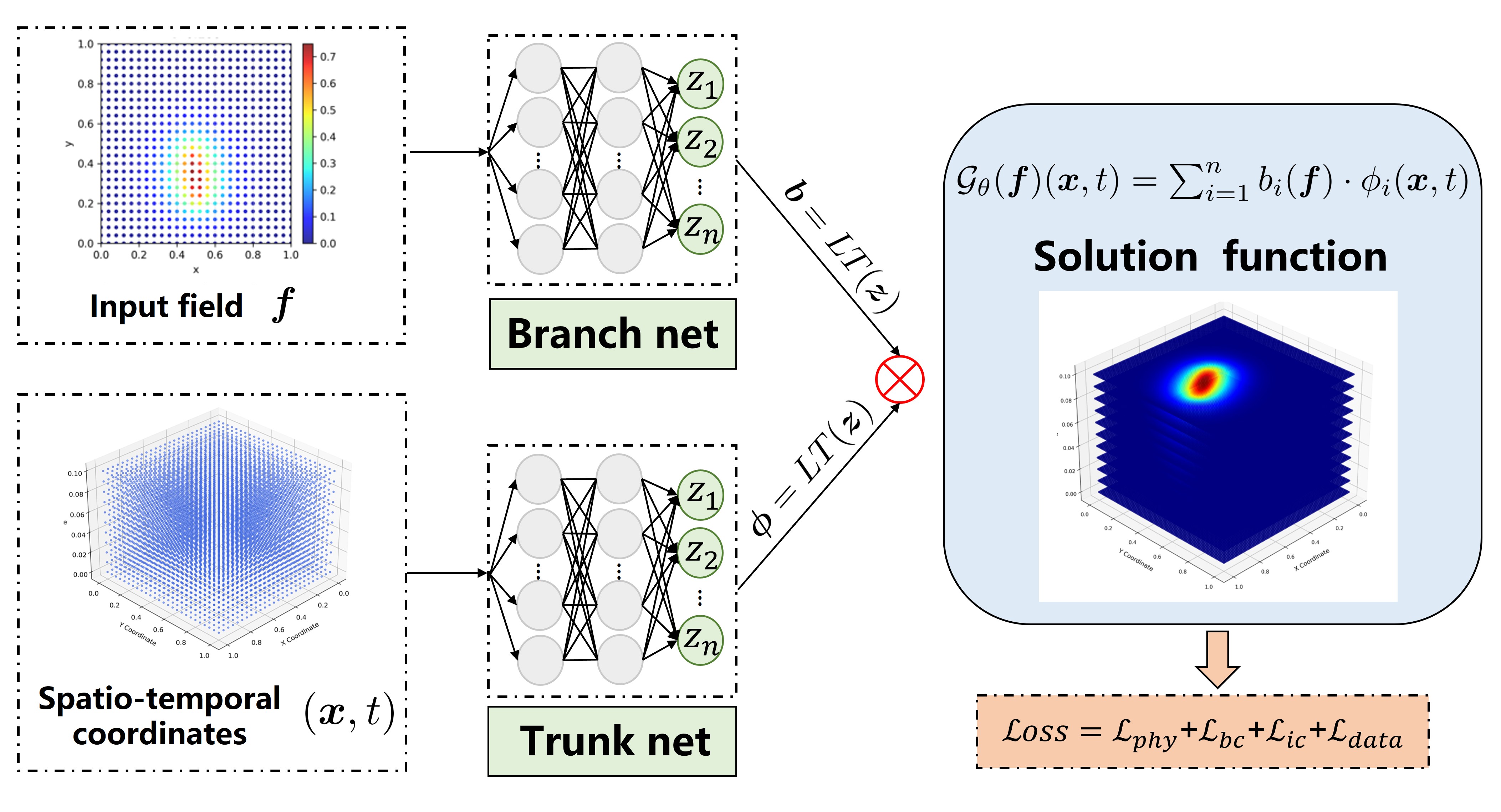}
	\caption{The schematic diagram of PI-DeepONet framework details the structural modules and forward calculation process, including the dual-input encoding via parallel branch and trunk networks, linear transformation of subnetwork outputs, inner product-based solution function construction, and multi-component loss calculation for physics-constrained training.}
	\label{fig:DeepONet_framework}
\end{figure}

PI-DeepONet extends the branch-trunk architecture of classical DeepONet by adopting the physics-constrained training paradigm of PINNs. It incorporates residuals of the target PDE and initial/boundary conditions into the loss function, enabling end-to-end learning of PDE solution mappings without extensive labeled data and suiting data-scarce engineering solving scenarios.

In this framework, time \(t\) and space \(x\) are jointly encoded by the trunk network with shared parameters, fully coupling the spatial distribution and temporal evolution of the solution. This is inconsistent with the classical numerical discretization paradigm of ``time-invariant spatial basis + time-varying evolution coefficients'' in the Galerkin method~\cite{Arnold2001}.

\subsection{Non-decreasing theorem of fitting error with expanding time intervals}
\label{wucha}

A widespread empirical observation is that fixed-parameter neural operators suffer from significant error accumulation in long-time spatiotemporal evolution problems. To formalize this phenomenon, this section establishes the non-decreasing fitting error theorem for fixed-parameter neural operators, providing foundational theoretical support for neural operator modeling of long-time-sequence evolution problems.

We first introduce basic definitions and notation. Let the spatial domain \(\Omega \subset \mathbb{R}^d\) be a Lebesgue measurable set, and denote \(L^2(\Omega)\) as the square-integrable function space on \(\Omega\), whose norm is defined as:
    \begin{equation}
         \|f\|_{L^2(\Omega)}^2 = \int_{\Omega} |f(x)|^2 \mathrm{d}x, \quad \forall f \in L^2(\Omega)
    \end{equation}
Denote the time domain as \(\mathcal{T} = [0, +\infty)\), and \(L^2_{\text{loc}}(\Omega  \times \mathcal{T})\) as the locally square-integrable spatiotemporal function space, i.e., for any finite time \(T>0\), we have \(u \in L^2(\Omega  \times [0,T] )\). Let the parameter space \(\Theta \subset \mathbb{R}^P\), where \(P\) is a fixed positive integer (corresponding to a neural network with fixed structure and parameters). For any parameter \(\theta \in \Theta\), the neural operator \(\mathcal{G}_\theta: \mathcal{T} \to L^2(\Omega)\) is a mapping from the time domain to the spatial square-integrable function space, i.e., for any time \(\tau \in \mathcal{T}\), \(\mathcal{G}_\theta(\cdot, \tau) \in L^2(\Omega)\). For any time \(t \in \mathcal{T}\), the fitting error functional of parameter \(\theta\) over the time interval \([0,t]\) is defined as:
    \begin{equation}
        J_t(\theta) = \int_{0}^{t} \left\| \mathcal{G}_\theta(\cdot, \tau) - u(\cdot, \tau) \right\|_{L^2(\Omega)}^2 \mathrm{d}\tau,
    \end{equation}
where \(u \in L^2_{\text{loc}}(\Omega \times  \mathcal{T})\) is the true state solution of the spatiotemporal physical system. The global minimum fitting error at time \(t\) is defined as the infimum of the error functional over the parameter space:
\begin{equation}
    L_t = \inf_{\theta \in \Theta} J_t(\theta)
\end{equation}

\begin{theorem}[Non-decreasing Theorem of Fitting Error for Fixed-Parameter Neural Operators]
Let the following two premises hold:
\begin{enumerate}
    \item The true solution of the spatiotemporal system satisfies \(u \in L^2_{\text{loc}}(\Omega \times  \mathcal{T} )\);
    \item For any fixed \(\theta \in \Theta\), the mapping \(\tau \mapsto \left\| \mathcal{G}_\theta(\cdot, \tau) - u(\cdot, \tau) \right\|_{L^2(\Omega)}^2\) is a Lebesgue measurable function, and for any finite \(t>0\), the integral \(J_t(\theta)\) is finite.
\end{enumerate}

Then the global minimum fitting error \(L_t\) is a non-decreasing function with respect to time \(t\), i.e., for any times \(t_1, t_2\) satisfying \(0 \leq t_1 < t_2 < +\infty\), it always holds that:

\begin{equation}
    L_{t_1} \leq L_{t_2}
\end{equation}
\end{theorem}

The proof is given in~\ref{wucha_proof}. The theorem shows that the non-decreasing property of the fitting error of fixed-parameter neural operators is independent of the network structure, optimization algorithm, and parameter space characteristics. Its essence stems from the increase in fitting constraints brought by the expansion of the time interval: a longer time interval means that the neural operator needs to fit the system evolution states at more times simultaneously, and its global minimum fitting error cannot be lower than the minimum error of a shorter time interval. Based on this theorem, the core challenge in long-time spatiotemporal evolution modeling with fixed-parameter neural operators is to control the growth rate of error with expanding time intervals.

Notably, \(L_{t_1} \geq L_{t_2}\) may be observed in experiments occasionally because the optimization algorithm cannot guarantee reaching the theoretical global infimum \(L_{t}\) in non-convex optimization problems.

\subsection{Spatiotemporal decoupled architecture: physics-informed Stone-Weierstrass neural operator}

Section~\ref{wucha} points out that the fitting error of any fixed-parameter neural operator will grow with the time interval. A large number of numerical experiments and related studies have confirmed that neural operators with different architectures exhibit significantly different error growth laws in long-time fitting tasks: some structures can achieve linear or even sublinear growth of errors, while the fitting errors of some structures show exponential growth with the expansion of the time interval (see Fig.~\ref{fig:heat1d-fig3} in Section~\ref{case1}). Therefore, achieving slow error growth through architectural design is key to improving long-time modeling performance of neural operators.

Mainstream neural operator architectures universally concatenate spatial and temporal coordinates as input to a single encoding network, approximating PDE solutions by generating spatiotemporally coupled basis functions \(\phi_i(x,t)\). This design completely binds the spatial distribution characteristics and the temporal evolution characteristics of the solution. 
Furthermore, the spatiotemporally coupled encoding method also makes it difficult for the model to
accurately fit periodic and multi-timescale unsteady PDEs. To overcome these limitations, we integrate the classical mathematical principle of the method of separation of variables for unsteady PDEs with the Stone-Weierstrass Approximation Theorem, and propose the physics-informed Stone-Weierstrass neural operator (PI-SWNO), a spatiotemporally decoupled architecture. This architecture structurally resolves the rapid error propagation caused by spatiotemporal coupling, enabling highly accurate, interpretable, and generalizable approximations of unsteady PDE solution mappings.

\begin{figure}[htbp]
	\centering
	\includegraphics[width=0.95\linewidth]{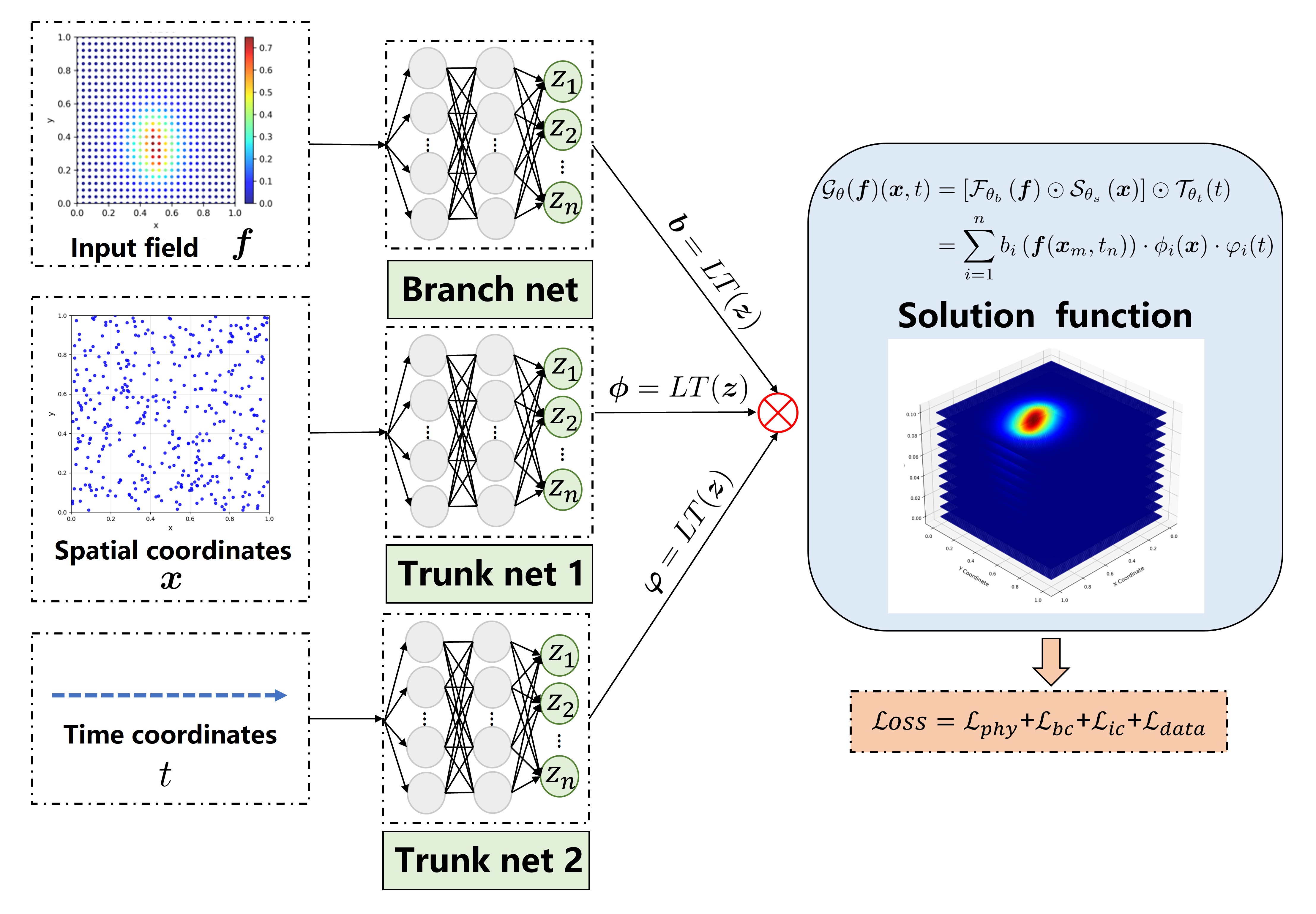}
	\caption{The schematic diagram of PI-SWNO framework details its core structural modules and forward calculation pipeline, including independent encoding of the input function field, spatial coordinates and time coordinates through three parallel subnetworks (branch network, spatial trunk network and temporal trunk network), solution function generation via the fused multiplication of basis coefficients, spatial basis functions and temporal basis functions, and multi-term loss calculation integrating physical equation, boundary condition, initial condition and data fidelity terms for training.}
	\label{fig:PI-SWNO_framework}
\end{figure}

Built on the canonical DeepONet, PI-SWNO fully decouples the original spatiotemporally coupled trunk network into two parameter-independent, functionally decoupled subnetworks: a spatial branch and a temporal branch.
We refer to the original branch network of DeepONet as the input branch.
This enables independent encoding of input function features, spatial distribution patterns, and temporal evolution laws. Its core mathematical formulation is given by:

\begin{equation} \label{sw_net_correct}
\begin{aligned}
\mathcal{G}_\theta(f)(x, t) 
&= \left[ \mathcal{F}_{\theta_b}\left(f\right) \odot \mathcal{S}_{\theta_s}\left(x \right)  \right] \odot \mathcal{T}_{\theta_t}(t) \\
&= \sum_{i=1}^n b_i\left(f(x_m, t_n)\right) \cdot \phi_i(x) \cdot \varphi_i(t)
\end{aligned}
\end{equation}
where $\theta = \{\theta_b, \theta_s, \theta_t\}$ is the set of learnable parameters of the full network, and $\odot$ denotes the Hadamard product (element-wise multiplication) of vectors. The terms are defined as follows:

\begin{itemize}
    \item The input branch network $\mathcal{F}_{\theta_b}$ performs global feature encoding on the input function $f$. Its input is the observed values $\{f(x_m, t_n)\}$ of the input function $f$ at spatiotemporal discrete sampling points $\{(x_m, t_n)\}_{m=1,n=1}^{M,N}$ and its output is mode coefficient vector \(\left[ b_1(f), b_2(f), \dots, b_n(f) \right]^T\).
    \item The spatial branch network ${\mathcal{S}}_{\theta_s}$ encodes spatial coordinates. Its input is the spatial coordinate $x \in \Omega$, and its output is time-invariant spatial basis function vector \(\left[ \phi_1(x), \phi_2(x), \dots, \phi_n(x) \right]^T \).
    \item The temporal branch network $\mathcal{T}_{\theta_t}$ encodes temporal coordinates. Its input is the time coordinate $t \in [0,T]$, and its output is temporal basis function vector \(\left[ \varphi_1(t), \varphi_2(t), \dots, \varphi_n(t) \right]^T\).
\end{itemize}
 
The core operation of Eq.~\eqref{sw_net_correct} proceeds in two steps: first, the mode coefficients output by the input branch network are subjected to the Hadamard product with the spatial basis functions output by the spatial branch network to obtain the input-related spatial mode encoding; then, it is subjected to the Hadamard product with the temporal basis functions output by the temporal branch network, and finally the approximate value of the solution function at any spatiotemporal coordinate $(x,t)$ is obtained by summation.

The approximation capability and structural rationality of PI-SWNO for the mapping \(\mathcal{G}\) are guaranteed by three major theorems: the Stone-Weierstrass approximation theorem~\cite{Rudin1987Real}, the universal approximation theorem for operators~\cite{23Lu2021}, and the universal approximation theorem for neural networks~\cite{Cybenko1989Approximation}. The theoretical proof of the universal approximation property of the PI-SWNO architecture is given in~\ref{appendix:piswno_approximation}. Compared with canonical DeepONet, the spatiotemporally decoupled architecture of PI-SWNO offers the following advantages:

\begin{itemize}
    \item It structurally mitigates error accumulation and divergence in long-time prediction. Spatiotemporally coupled basis functions of canonical DeepONet exhibit the cross-amplification mechanism of spatiotemporal errors, leading to high-slope or even exponential error growth. The time-invariant spatial basis functions of PI-SWNO strictly limit the temporal error within the temporal domain, preventing spatial domain-wide propagation and exponential accumulation across time steps, thereby suppressing long-time error divergence. Most engineering unsteady PDEs possess time-invariant intrinsic spatial eigenmodes; temporal evolution only changes the weights of these modes. Spatiotemporally coupled basis functions force the inherent spatial modes of the system to continuously distort with time, introducing systematic fitting bias. For short-time problems, this bias can be masked by the fitting ability of the network, but for long-time problems, the non-physical mode distortion will continue to accumulate, leading to solution divergence. Decoupling the fitting tasks of spatial distribution and temporal evolution greatly reduces the difficulty of fitting high-dimensional coupled features of long-time sequences, and the network capacity can be efficiently allocated to the core physical laws of the solution, achieving a higher upper bound of fitting accuracy and stronger robustness in long-time generalization under identical or reduced training budgets.
    \item It possesses an inherent structure tailored to periodic unsteady PDEs, enhancing fitting accuracy and generalization. For unsteady processes with periodic or quasi-periodic behavior, the spatial branch network of PI-SWNO learns fixed spatial distribution modes, and the temporal branch network can explicitly capture temporal evolution features such as frequency, amplitude, and phase, significantly reducing the network's learning complexity.
    \item It exhibits flexible adaptability to multi-scale spatiotemporal properties with greater architectural design flexibility. For multi-scale unsteady PDE systems, spatial and temporal features often have markedly distinct scaling characteristics (e.g., multi-scale spatial geometry, fast/slow temporal evolution modes). The fully decoupled spatial and temporal trunk networks of PI-SWNO enable independent network design tailored to each dimension: for instance, multi-scale convolutions for the spatial trunk to capture spatial multi-scale features, and temporal networks or periodic activation functions for the temporal trunk to capture temporal dependencies and periodicity.
\end{itemize}

\subsection{Time-marching batch-wise sampling strategy}

As shown in Eq.~\eqref{sw_net_correct} and Fig.~\ref{fig:PI-SWNO_framework}, training neural operators for time-dependent PDEs demands sufficient random sampling of coordinate points in the spatiotemporal computational domain. Sampling density is a key factor determining solution accuracy and training stability. For long-time PDE solving, fixed sampling density leads to a linear increase in the number of sampling points with the length of the time domain, drastically increasing GPU memory usage and the computational overhead of gradient computation during training — often exceeding the memory capacity of available GPUs. To tackle this problem, we propose a time-marching batch-wise sampling strategy featuring phased local optimization and global consistency calibration (see Fig.~\ref{fig:sample-strategy}). On one hand, localized sampling via temporal blocking drastically reduces the number of sampling points per training step, alleviating memory pressure and computational burden. On the other hand, a global consistency calibration step over the full time domain corrects global biases introduced by blockwise local optimization, ensuring the continuity, consistency, and convergence of full-time-domain solutions.

Let the total time domain to be solved be \([0,T]\). Given a subdomain time step \(\Delta T\), the full time domain is evenly divided into \(K = T/\Delta T\) continuous non-overlapping sub-time domains:

\begin{equation}
\mathcal{T}_k = [(k-1)\Delta T, k\Delta T], \quad k=1,2,\dots,K
\end{equation}

Denote the spatial computational domain as \(\Omega \subset \mathbb{R}^d\), and \(\mathrm{Vol}(\cdot)\) as the Lebesgue measure; define \(\rho_{\text{local}}\) and \(\rho_{\text{global}}\) as the local sampling density and global sampling density, respectively. The execution process of a single training epoch is divided into the following two core stages:

1. Local batch optimization stage: Sequential local optimization over temporal blocks achieves high-precision approximation of PDE solutions within each subdomain while controlling per-step training memory usage. For the $k$-th spatiotemporal sub-domain \(\Omega \times \mathcal{T}_k\), random sampling at the local sampling density \(\rho_{\text{local}}\) yields the local sampling point set:
\begin{equation}
S_k = \{(x_{k,i}, t_{k,i}) \mid x_{k,i} \in \Omega, t_{k,i} \in \mathcal{T}_k\}
\end{equation}
Then, one gradient descent update is performed sequentially in each sub-time domain in the time order from \(k=1\) to \(k=K\).

2. Global consistency calibration stage: The global deviation introduced by block local optimization is corrected through global sampling in the full time domain to ensure the consistency of the full-time-domain solution. Random sampling is performed in the entire spatiotemporal domain \(\Omega \times [0,T]\) according to the global sampling density \(\rho_{\text{global}}\) to generate the global sampling point set:
\begin{equation}
S_{\text{global}} = \{(x_g, t_g) \mid x_g \in \Omega, t_g \in [0,T]\}
\end{equation}
Loss gradients are computed on the global sampling set, and one gradient descent update is applied.

Completing both stages constitutes one full epoch of iteration. The iterative process of local batch optimization and global consistency calibration is repeated until the global objective function converges

\begin{figure}[htbp]
	\centering
	\includegraphics[width=0.98\linewidth]{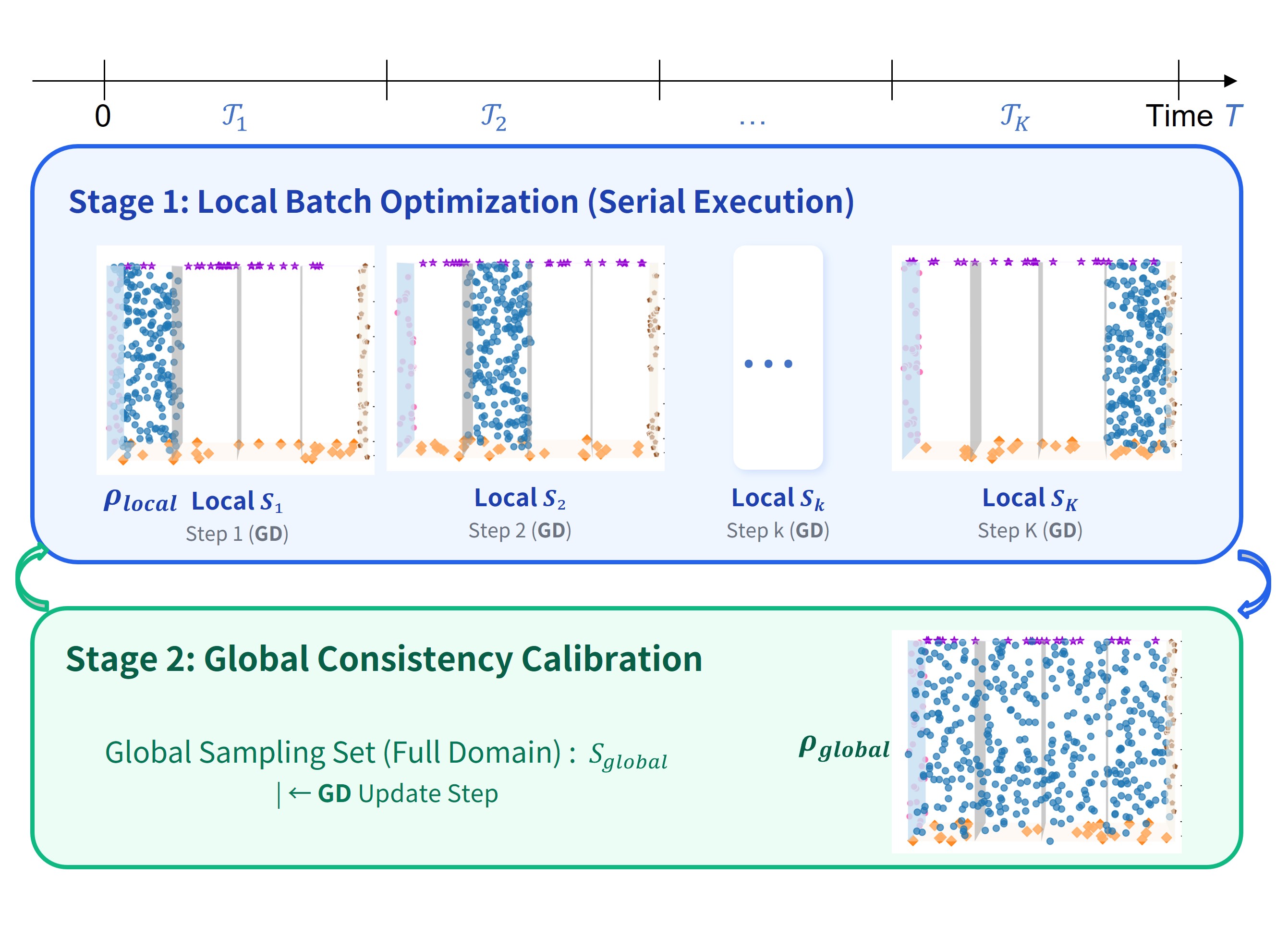}
	\caption{The schematic diagram of time-marching batch-wise sampling strategy illustrates the two-stage iterative training workflow, including the discretization of the full solution time domain into sequential temporal blocks, serial gradient descent updates for local batch optimization on each time-segmented sampling set, full-domain gradient descent update for global consistency calibration, and the closed-loop iteration mechanism that forms a complete training epoch.}
	\label{fig:sample-strategy}
\end{figure}

\section{Computational Experiments}
\label{sec4}

In this section, we employ four representative time-dependent PDEs as benchmark problems: the heat conduction (HC) equations, wave equations, Korteweg-de Vries (KdV) equations, and Burgers equations. A total of 7 numerical cases are designed to validate the effectiveness, stability, and solution accuracy of the proposed PI-SWNO model, with direct performance comparison against the vanilla PI-DeepONet. The details of all test cases are summarized in Table~\ref{tab:CASEoverview}.

\begin{table}[htbp]
\centering
\small
\setlength{\tabcolsep}{2pt}
\caption{Overview of the seven spatiotemporal problems tested in this paper.}
\label{tab:CASEoverview}
\begin{tabular}{ccccc}
\toprule
 Case & Input type &Input function & Input size  &  Trajectories \\
\midrule
1D HC eq.& Source field& Bimodal function  & 26 & 1600 \\
2D HC eq.& Source field & Gaussian function & 676 & 1225\\
1D Wave eq.& Initial condition & Periodic GRF& 101  &  1200 \\
2D Wave eq.& Initial condition & Periodic GRF& 676  &  1200 \\
1D KdV eq. & Initial condition &Periodic GRF&  513 & 1200 \\
1D Burgers eq. &  Initial condition &Periodic GRF& 101 & 1200\\
2D Burgers eq.& Initial condition &Periodic GRF& 676 & 1200\\
\bottomrule
\end{tabular}
\end{table}

The first 4 cases focus on linear time-dependent PDEs, covering parabolic heat conduction equations and hyperbolic wave equations, to test the model’s ability to model diffusion and propagation-dominated linear physical laws. The remaining 3 cases move to strongly nonlinear, strongly spatiotemporally coupled unsteady PDEs, including the KdV equation (nonlinear dispersive waves) and Burgers equations (advection-diffusion shock problems). For each experiment, different input functions are used for training and testing samples to investigate the generalization ability of the neural operators across unseen configurations. The code and dataset for all cases will be made publicly available upon publication.

All numerical experiments are run on a computing cluster equipped with 8 NVIDIA L40S GPUs, each with 46,068 MiB of video memory. The software stack includes NVIDIA driver version 560.35.03 paired with CUDA Toolkit 12.6. In all cases, the Adam optimizer~\cite{ABUEIDDA2025} is employed, with a tanh activation function, and Glorot normal initialization for model parameters. The initial learning rate is set to 1e-3, with a multi-step learning rate decay strategy that reduces the learning rate by a factor of 0.5 at 2000 and 5000 iterations.

To avoid error distortion caused by an excessively small overall $L_2$ norm of the solution—especially for problems with sharp local amplitude variations such as shocks and discontinuities—we adopt the amplitude-normalized relative $L_2$ error (ANRL2E) as the primary performance evaluation metric, defined as:

\begin{align}
\text {ANRL2E} &= \frac{\left\|u_{\text{true}} - u_{\text{pred}}\right\|_2}{\max \left(\left|u_{\text{true}}\right|\right)}\\
&=  \frac{\sqrt{\sum_{j=1}^{n_t} \sum_{k=1}^{n_{x}}\left(u\left(f, x^{(k)} ,t^{(j)},\right)-\hat{u}\left(f, x^{(k)}, t^{(j)}\right)\right)^2}}{ \max\left|u\left(f,x , t\right)\right|}
\end{align}

\subsection{1-D Heat conduction equations}\label{case1}

We first study a one-dimensional heat conduction system:

\begin{equation}\label{eq:case1}
\begin{cases}
&\frac{\partial u}{\partial t} = D \frac{\partial^2 u}{\partial x^2} + f(x), \quad  (x,t) \in [0,1]\times(0,10]\\
&u(0,t)+\frac{\partial u}{\partial x}(0,t) =0, \quad t \in (0,10]  \\
&u(1,t)+\frac{\partial u}{\partial x}(1,t) =0, \quad t \in (0,10] \\
&u(x,0)=0, \quad x \in (0,1) \\
&f(x) = A_1 \exp\left(-\frac{(x - c_1)^2}{2 \sigma_1^2}\right) + A_2 \exp\left(-\frac{(x - c_2)^2}{2 \sigma_2^2}\right)
\end{cases}
\end{equation}
where \(u(x,t)\) denotes the temperature field to be solved, the thermal diffusion coefficient \(D=0.001\), the boundary condition is Robin-type convective heat transfer, the initial condition is set to 0, and the heat source is a bimodal function with 6 tunable parameters: amplitudes \(A_1\), \(A_2\), center positions \(c_1\), \(c_2\), and standard deviations \(\sigma_1\), \(\sigma_2\) controlling the width of the two peaks. A total of 1600 trajectories are generated via parameter sweeping, with detailed parameter settings provided in Table~\ref{tab:paramstab}. We randomly select 80\% of the trajectories for training and the remaining 20\% for testing. The objective is to learn the operator \(\mathcal{G}:f(x) \mapsto u(x,t)\) that maps the heat source function to the exact solution of the equation. The loss function is formulated as:

\begin{equation}\label{eq:loss1}
\begin{aligned}
	\mathcal{L}(\theta)=&\frac{\alpha_f}{BN_f}\sum_{b=1}^B \sum_{i=1}^{N_f}\left\|\frac{\partial \mathcal{G}_\theta\left(f_{b},x_i,t_i\right)}{\partial t}-D\frac{\partial^2\mathcal{G}_\theta\left(f_{b},x_i,t_i\right)}{\partial x^2}-f(x_i)\right\|^2 \\
	+&\frac{\alpha_b}{B N_b} \sum_{b=1}^B \sum_{j=1}^{N_b}\left\|\mathcal{G}_\theta\left(f_{b},x_j,t_j\right)+\frac{\partial \mathcal{G}_\theta\left(f_{b},x_j,t_j\right)}{\partial x}\right\|^2 \\
	+&\frac{\alpha_i}{B N_i} \sum_{b=1}^B  \sum_{p=1}^{N_i}\left\|\mathcal{G}_\theta\left(f_{b},x_p,0\right)\right\|^2  \\
	+&\frac{\alpha_d}{B N_d} \sum_{b=1}^B  \sum_{k=1}^{N_d}\left\|\mathcal{G}_\theta\left(f_{b},x_k,t_k\right)-u\left(x_k,t_k\right)\right\|^2 \\
\end{aligned}
\end{equation}
where \(\mathcal{G}_{\Theta}\) denotes the neural operator network with trainable parameters \(\theta\), \(B\) is the batch size, $N_f$ and $N_b$ are the number of training points randomly sampled from the domain $\Omega$ and boundary $\partial\Omega$, respectively, $N_i$ is the number of spatial sampling points at the initial time, $N_d$ is the number of observation data points, and $\alpha_f$, $\alpha_b$, $\alpha_i$, $\alpha_d$ are the weight coefficients for the PDE residual, boundary condition residual, initial condition residual, and observation data residual, respectively. The weights for each loss term in this case are set to: $\alpha_f=1$, $\alpha_b=5$, $\alpha_i=5$, and $\alpha_d=5$.

We randomly sample 1200 points inside the spatiotemporal domain \(\Omega\), 50 points at the initial time, 200 points on the boundary, and use the temperature values at 50 uniformly sampled points at the final time as observation data. For the PI-DeepONet framework, both the branch network and trunk network are set to 3 fully connected layers, with a hidden dimension of 32 for the branch network, 64 for the trunk network, and 16 basis functions and corresponding coefficients. For the PI-SWNO framework, the input branch, temporal branch, and spatial branch are all set to 3 fully connected layers with a hidden dimension of 32, and 16 basis functions and corresponding coefficients. Both models are trained for 20,000 iterations with a batch size of 128.

We select 3 samples from the test set to visualize the heat source function, ground-truth temperature field, predicted temperature fields from the two models, and the spatiotemporal distribution of absolute errors, as shown in Fig.~\ref{fig:heat1d-fig1}. In terms of overall prediction performance, both PI-DeepONet and PI-SWNO can reproduce the spatiotemporal evolution of the ground-truth temperature field with high accuracy, including the position of the temperature peak, diffusion trend, and gradient distribution driven by the heat source, which is consistent with the exact solution. This validates the effectiveness of both physics-informed neural operator models for solving the heat conduction equation. Notably, PI-SWNO yields a lower absolute error amplitude and a more uniform spatiotemporal error distribution, demonstrating superior gradient capture capability and local prediction accuracy.

\begin{figure}[htbp]
	\centering
	\includegraphics[width=0.98\linewidth]{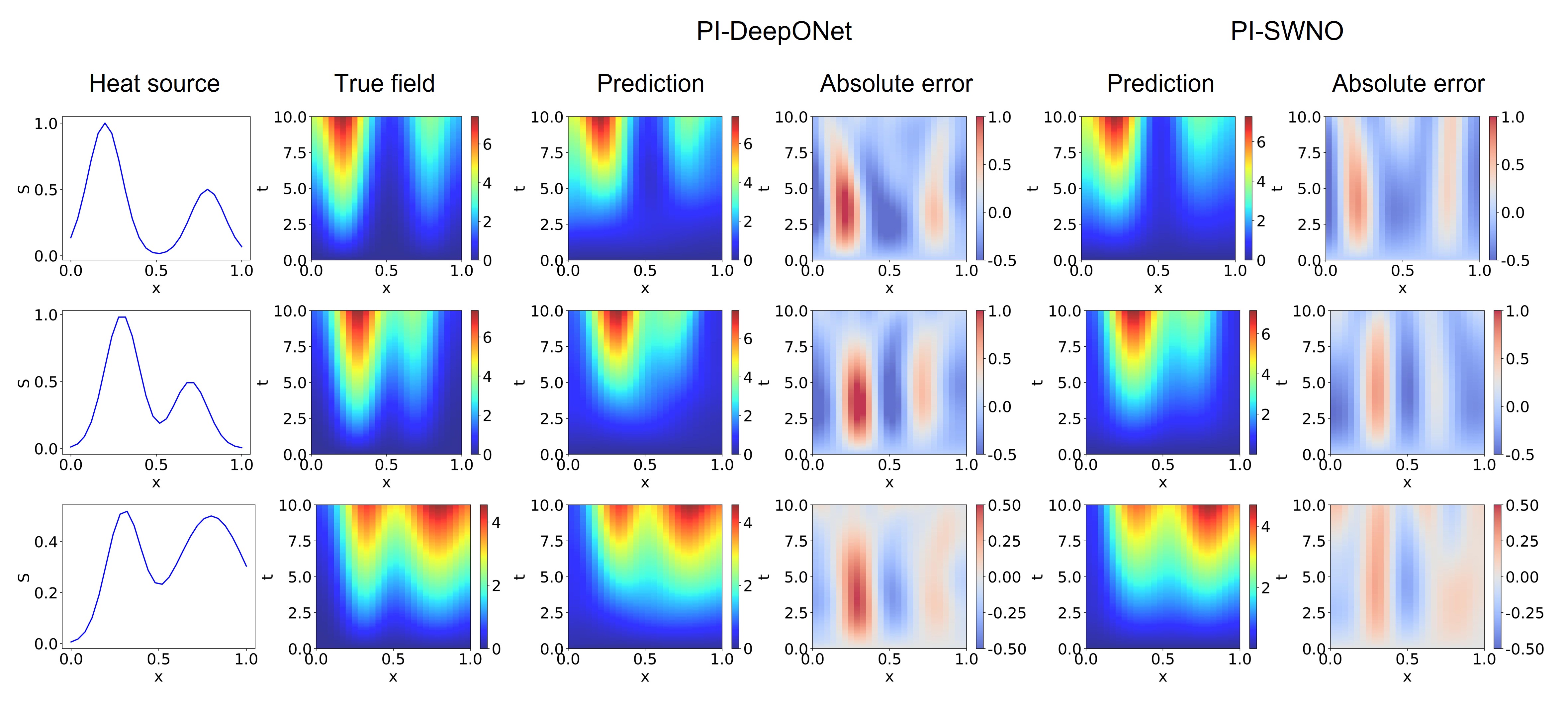}
	\caption{ 1D heat conduction equation: The first column shows the randomly sampled source term function; the second column shows the high-fidelity numerical solution; the third and fourth columns show the predicted solution and its absolute error from PI-DeepONet; the fifth and sixth columns show the predicted solution and its absolute error from PI-SWNO.}
	\label{fig:heat1d-fig1}
\end{figure}

To investigate the evolution of prediction accuracy and stability of the two models across different time domain spans, we gradually extend the time horizon and predict the spatiotemporal evolution of the temperature field.

Fig.~\ref{fig:heat1d-fig3} shows the boxplot distribution of ANRL2E for the two models on the 1D heat conduction equation case, as the time domain span increases from $[0, 10]$ to $[0, 100]$. The boxplot reflects the quartile distribution of errors, the red star marks the mean error, and the whiskers represent the extreme value range of errors. To quantify the growth trend of error with time domain length, we perform least-squares linear regression on the mean error of the two models at different time domain lengths, using the slope of the fitted line to characterize the rate of error growth over time. The error growth rate is 0.085 for PI-SWNO and 1.255 for PI-DeepONet, meaning the error growth rate of PI-DeepONet is 14 times that of PI-SWNO. In the short time horizon scenario (shorter than $[0, 50]$), the two models achieve comparable error levels, with mean errors below 5\% and no significant difference in boxplot distribution or dispersion. In the medium-to-long time horizon scenario (longer than $[0, 50]$), the error of PI-DeepONet increases sharply: the mean error rises rapidly from about 7\% at $[0, 50]$ to about 12\% at $[0, 100]$, with a significant expansion of the interquartile range. This indicates a rapid decline in the model’s ability to capture long-time spatiotemporal evolution features and a marked deterioration in prediction stability. In contrast, PI-SWNO exhibits stronger stability in long time horizon scenarios: even when the time domain span is extended to $[0, 100]$, its mean error remains stable in the range of 4\% to 5\%, with no obvious expansion of the interquartile range, and the upper limit of the error extreme value is consistently below 7\%, with no error explosion phenomenon as the time domain span increases. 

\begin{figure}[htbp]
	\centering
	\includegraphics[width=0.85\linewidth]{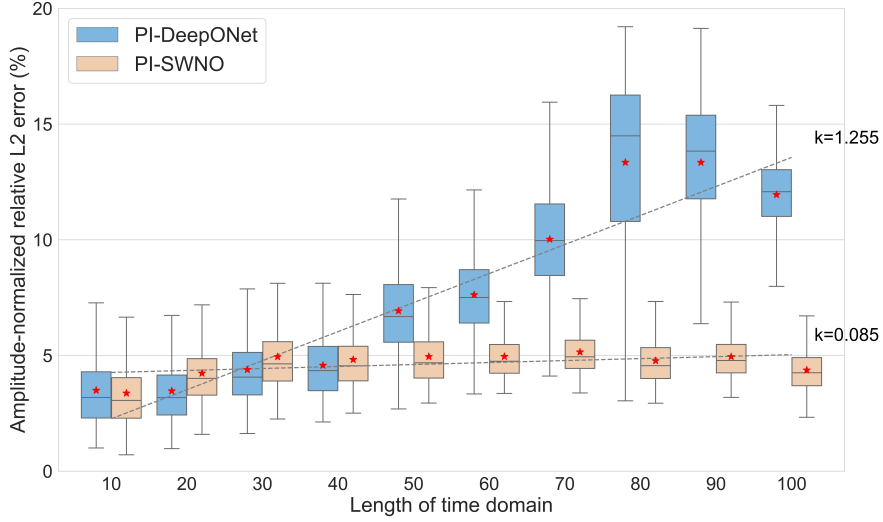}
	\caption{Comparison of long-time ANRL2E growth trends between the baseline PI-DeepONet and the proposed spatiotemporally decoupled PI-SWNO across expanding temporal domain spans on the 1D heat conduction equation.}
	\label{fig:heat1d-fig3}
\end{figure}

Fig.~\ref{fig:heat1d-fig2} shows the statistical distribution of the mean squared error (MSE) and ANRL2E on the test set for PI-DeepONet and PI-SWNO in the short time horizon $[0, 10]$ and long time horizon $[0, 100]$, respectively. The results show that PI-SWNO outperforms PI-DeepONet comprehensively in terms of mean error, extreme value range, and dispersion, with higher overall prediction accuracy, smaller error fluctuation, stronger robustness, better generalization across different heat source samples, and a lower upper error bound for extreme outlier samples.

These results fully demonstrate that PI-SWNO has superior long-term spatiotemporal dependency modeling capability compared to PI-DeepONet, can effectively suppress error accumulation and growth in long-time prediction, and is more suitable for PDE solving tasks with large time spans.

\begin{figure}[htbp]
  \centering
  \includegraphics[width=0.98\linewidth]{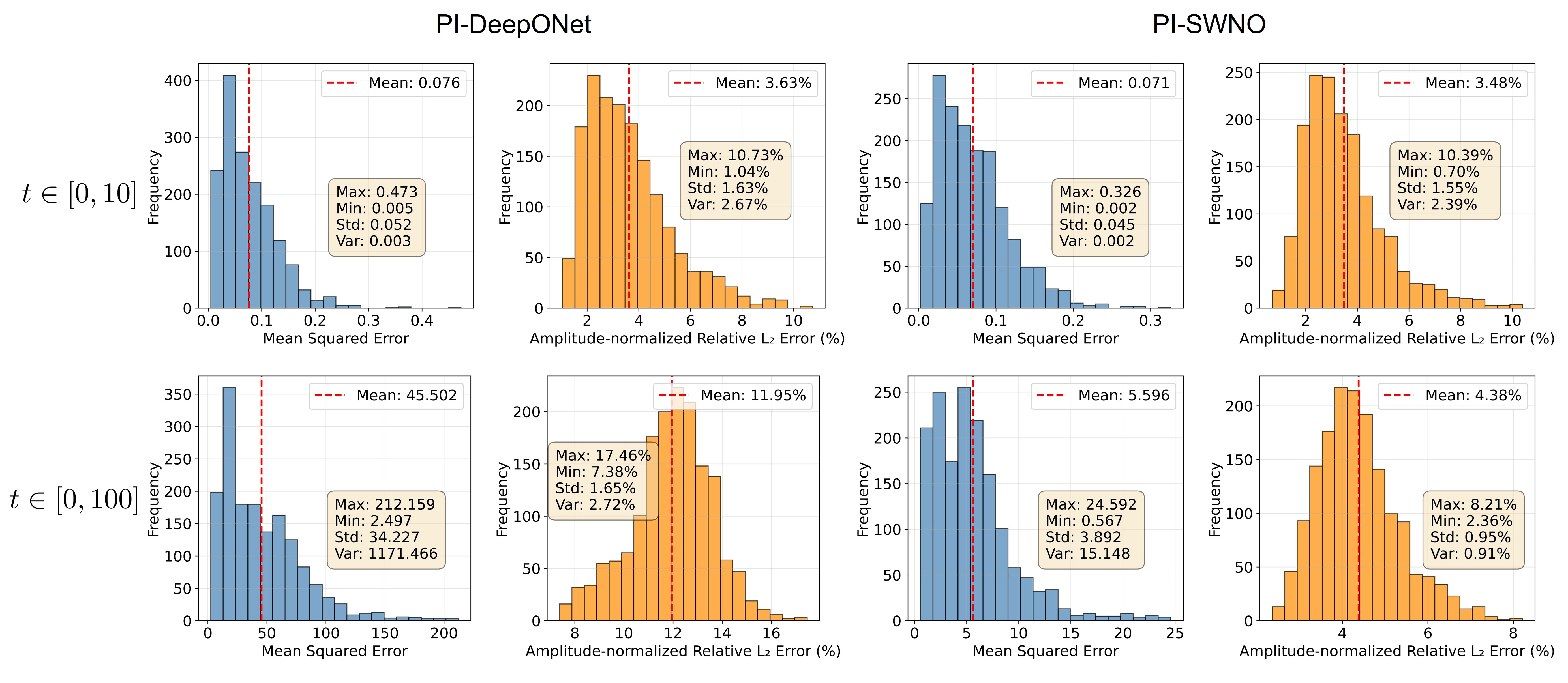}
  \caption{Statistical characteristics of mean squared error (MSE) and ANRL2E for PI-DeepONet and PI-SWNO on the 1D heat conduction equation dataset in the short time horizon $[0, 10]$ and long time horizon $[0, 100]$, respectively.}
  \label{fig:heat1d-fig2}
\end{figure}

\subsection{2-D Heat conduction equations}\label{case2}

To assess the model’s generalization to higher spatial dimensions, we next examine a two-dimensional heat conduction problem with Robin boundary conditions and a Gaussian heat source distribution:
\begin{equation}
\small
\begin{cases}
&\frac{\partial u}{\partial t} = D \left( \frac{\partial^2 u}{\partial x^2} + \frac{\partial^2 u}{\partial y^2} \right) + f(x,y), \quad  (x,y,t) \in [0,1]\times[0,1]\times(0,50]\\
&u(0,y,t)+\frac{\partial u}{\partial x}(0,y,t) =0, \quad u(1,y,t)+\frac{\partial u}{\partial x}(1,y,t) =0, \quad  (y,t) \in [0,1]\times(0,50]  \\
&u(x,0,t)+\frac{\partial u}{\partial y}(x,0,t) =0, \quad u(x,1,t)+\frac{\partial u}{\partial y}(x,1,t) =0, \quad (x,t) \in [0,1]\times(0,50] \\
&u(x,y,0)=0, \quad (x,y) \in (0,1)\times(0,1)\\
&f(x,y) = A \exp\left( -\left( \frac{(x - x_c)^2}{2 \sigma_x^2} + \frac{(y - y_c)^2}{2 \sigma_y^2} \right) \right)
\end{cases}
\end{equation}
where the meaning of each parameter is consistent with that in Eq.~\eqref{eq:case1}, the thermal diffusion coefficient \(D=0.001\), the boundary condition is Robin-type convective heat transfer, the initial condition is set to 0, and the heat source is a 2D Gaussian function with 5 undetermined parameters: amplitude \(A\), heat source center positions \(x_c\), \(y_c\), and standard deviations \(\sigma_x\), \(\sigma_y\) in the two spatial directions. A total of 1225 trajectories are generated via parameter sweeping, with detailed parameter settings provided in Table~\ref{tab:paramstab}. We randomly select 80\% of the trajectories for training and the remaining 20\% for testing. The objective is to learn the operator \(\mathcal{G}:f(x,y) \mapsto u(x,y,t)\) that maps the heat source function to the exact solution of the equation, with the loss function referring to Eq.~\eqref{eq:loss1}. The weights for each loss term in this case are set to: $\alpha_f=1$, $\alpha_b=5$, $\alpha_i=5$, and $\alpha_d=5$.

\begin{table}[htbp]
	\small
	\centering
	\caption{Parameter settings for the bimodal heat source in the 1D heat conduction case and the Gaussian heat source in the 2D heat conduction case.}
	\label{tab:paramstab}
	\begin{tabular}{cccc}
		\toprule
		\multicolumn{2}{c}{Bimodal source} & \multicolumn{2}{c}{Gaussian source} \\
		\cmidrule(lr){1-2} \cmidrule(lr){3-4} 
		Params & Value & Params & Value \\
		$A_1$ & 0.5, 1 & $A$ & 5 \\
		$A_2$ & 0.5, 1 & $x_c$ & 0.2,0.3,0.4,0.5,0.6,0.7,0.8  \\
		$\sigma_1$ & 0.05, 0.1, 0.15, 0.2 & $y_c$ &0.2,0.3,0.4,0.5,0.6,0.7,0.8 \\
		$\sigma_2$ & 0.05, 0.1, 0.15, 0.2  & $\sigma_x$ & 0.05,0.0625,0.075,0.0875,0.1 \\
		$c_1$ & 0.1,0.2,0.3,0.4,0.5 & $\sigma_y$ & 0.05,0.0625,0.075,0.0875,0.1 \\
		$c_2$ & 0.5,0.6,0.7,0.8,0.9 & - & - \\
		\bottomrule
	\end{tabular}
\end{table}

We randomly sample 10,000 points inside the spatiotemporal domain \(\Omega\), 1000 points at the initial time, 4000 points on the boundary, and use the temperature values at 1000 uniformly sampled points at the final time as observation data. For the PI-DeepONet framework, both the branch network and trunk network are set to 6 fully connected layers, with a hidden dimension of 100 for the branch network, 200 for the trunk network, and 64 basis functions and corresponding coefficients. For the PI-SWNO framework, the input branch, temporal branch, and spatial branch are all set to 6 fully connected layers with a hidden dimension of 100, and 64 basis functions and corresponding coefficients. Both models are trained for 50,000 iterations with a batch size of 8.

As shown in Fig.~\ref{fig:heat2d-fig1}, both PI-DeepONet and PI-SWNO can accurately reproduce the spatiotemporal evolution of the temperature field, precisely capturing the spatial position of the temperature peak, diffusion trend, and amplitude growth characteristics. This confirms that both physics-informed neural operators can effectively learn the intrinsic physical constraints of the heat conduction equation. From the comparison of error fields, PI-DeepONet exhibits obvious error accumulation in the heat source core region and diffusion front with large temperature gradients, while PI-SWNO yields a more widely distributed absolute error, rather than being concentrated only in the extremely local region of the temperature peak.

\begin{figure}[htbp]
	\centering
	\includegraphics[width=0.98\linewidth]{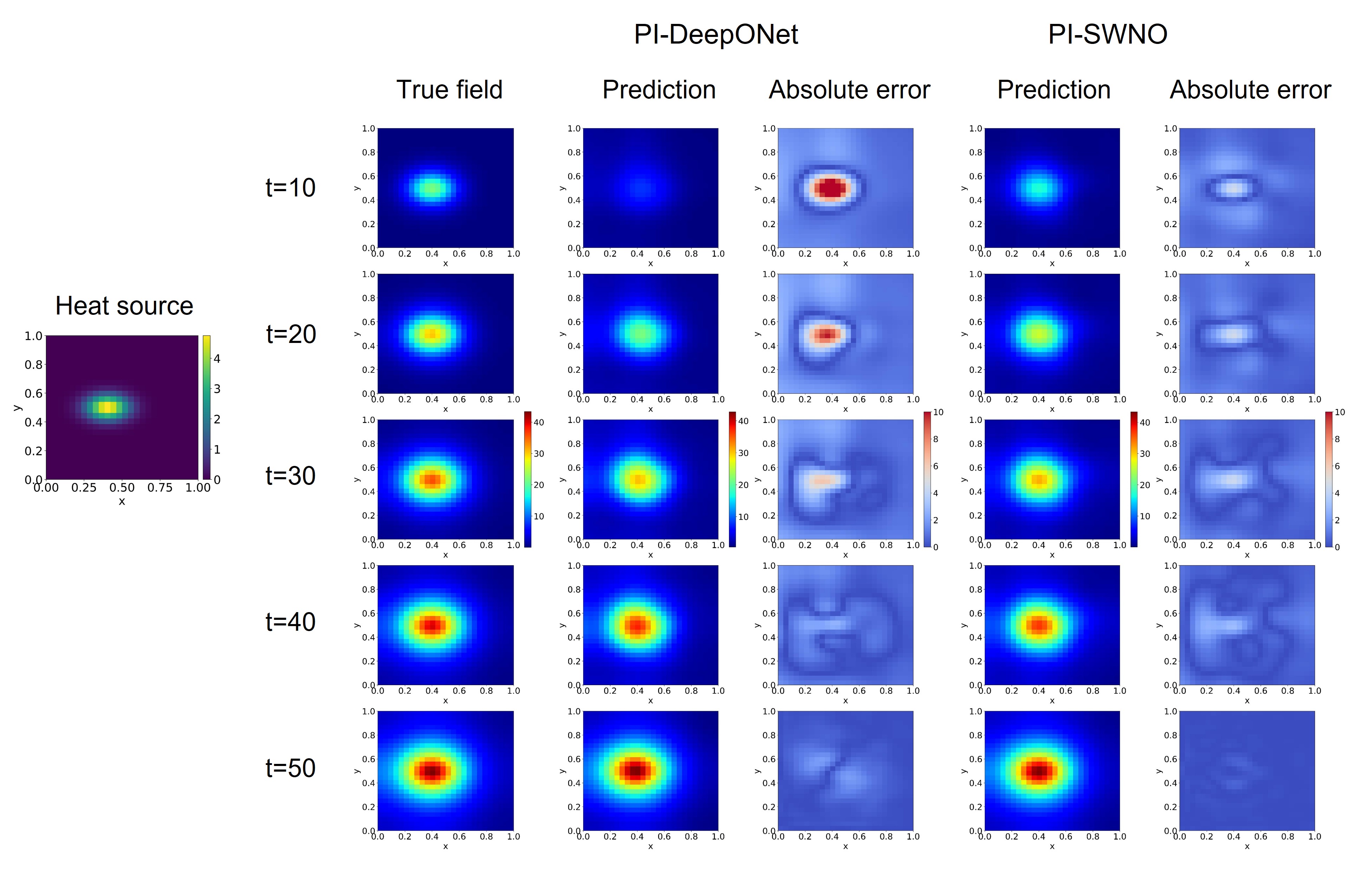}
	\caption{2D heat conduction equation: The first plot shows the randomly sampled source term function; the second column shows 2D slices of the high-fidelity numerical solution at $t = 10, 20, 30, 40, 50$; the third and fourth columns show the predicted solution and its absolute error from PI-DeepONet; the fifth and sixth columns show the predicted solution and its absolute error from PI-SWNO.}
	\label{fig:heat2d-fig1}
\end{figure}

As shown in Fig.~\ref{fig:heat2d-fig3}, the prediction errors of both models show an overall increasing trend with the extension of the prediction time horizon, which is consistent with the non-decreasing fitting error theorem for fixed-parameter neural operators presented in Section~\ref{wucha}, with local decreases attributed to optimization errors. The slopes of the error growth rate curves fitted by least-squares linear regression are 1.277 for PI-DeepONet and 0.238 for PI-SWNO, a 5.4-fold difference. In the short time horizon of 10$\sim$40, PI-DeepONet and PI-SWNO achieve comparable accuracy. When the length of time domain is longer than 50, the error of PI-DeepONet significantly exceeds that of PI-SWNO and grows rapidly, accompanied by an increase in error dispersion. At $[0, 100]$, the minimum error of PI-DeepONet already exceeds 11\%, the maximum error approaches 16\%, and the boxplot distribution is wide. This indicates that the model’s prediction robustness is greatly reduced over long time horizons, with significant performance differences across samples, presenting a typical error accumulation and divergence problem. In contrast, PI-SWNO maintains a slow error growth rate in the long time horizon of 50$\sim$100: the average relative error is still controlled within 3\%, and the maximum extreme error does not exceed 5\%. The error dispersion remains at an extremely low level across the full time domain, with a narrow interquartile range. 

Fig.~\ref{fig:heat2d-fig2} provides a more detailed view of the statistical characteristics of MSE and ANRL2E for the two models on the 2D heat conduction equation dataset in the short time horizon $[0, 10]$ and long time horizon $[0, 100]$. The above results demonstrate that the PI-SWNO model can structurally suppress error accumulation, maintain good physical consistency throughout long-time evolution, and achieve low-rate error growth.

\begin{figure}[htbp]
	\centering
	\includegraphics[width=0.85\linewidth]{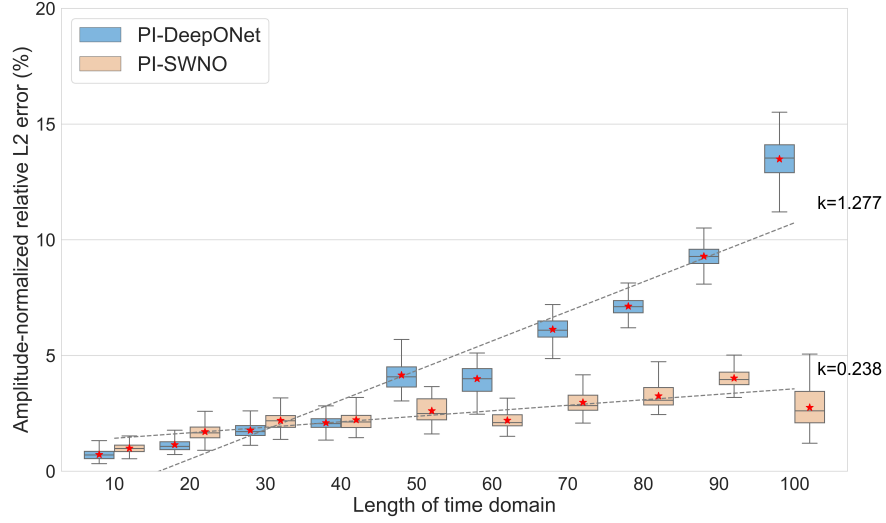}
	\caption{Comparison of long-time ANRL2E growth trends between the baseline PI-DeepONet and the proposed spatiotemporally decoupled PI-SWNO across expanding temporal domain spans on the 2D heat conduction equation.}
	\label{fig:heat2d-fig3}
\end{figure}

\begin{figure}[htbp]
  \centering
  \includegraphics[width=0.98\linewidth]{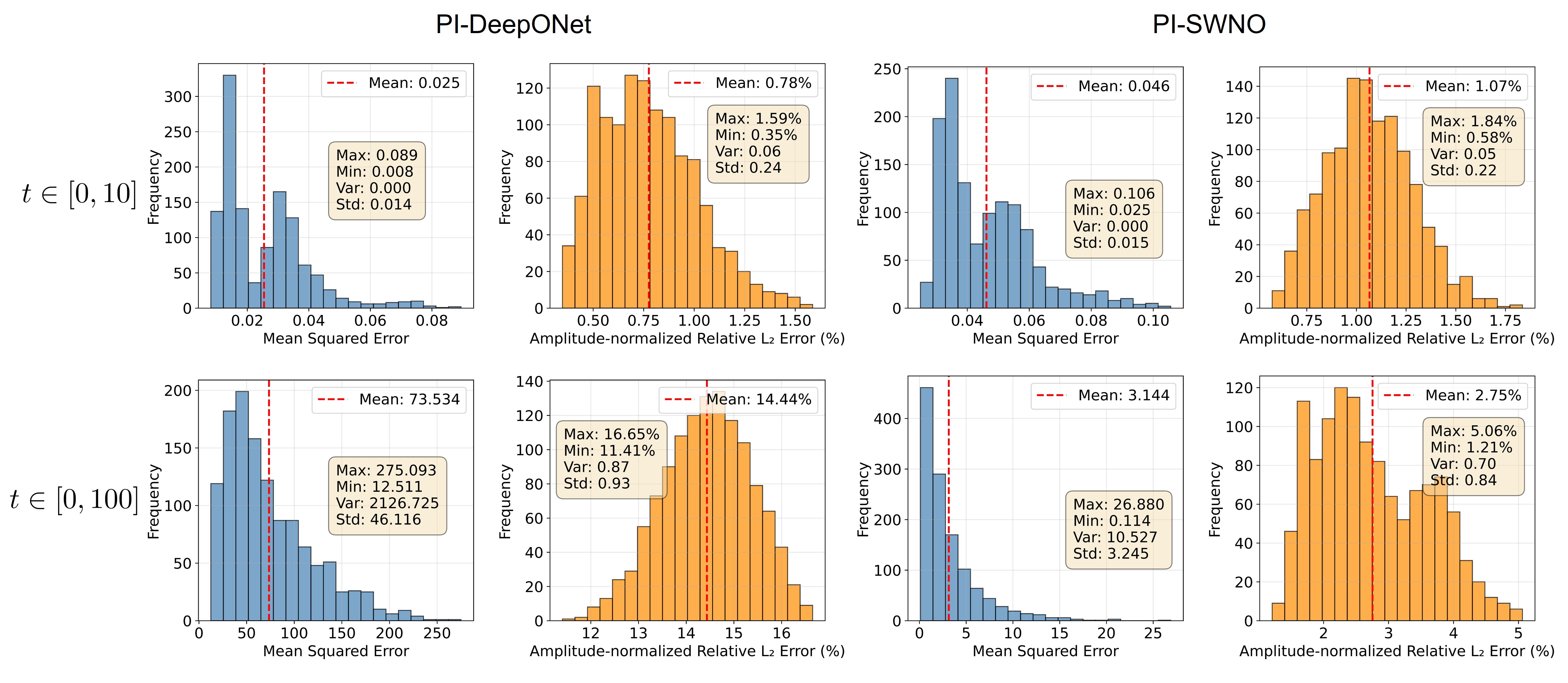}
  \caption{Statistical characteristics of MSE and ANRL2E for PI-DeepONet and PI-SWNO on the 2D heat conduction equation dataset in the short time horizon $[0, 10]$ and long time horizon $[0, 100]$, respectively.}
  \label{fig:heat2d-fig2}
\end{figure}

\subsection{1-D wave equations}\label{case3}

Next, we investigate a one-dimensional wave problem, a canonical linear hyperbolic PDE that describes the propagation of string vibrations, acoustic waves, and perturbation waves:

\begin{equation}\label{eq:case3}
\begin{cases}
&\frac{\partial^2 u}{\partial t^2}=c^2 \frac{\partial^2 u}{\partial x^2}, \quad (x,t) \in [0,1] \times (0,6]\\
&u(0,t)=u(1,t), \quad t \in (0,6]  \\
& \frac{\partial u}{\partial x}(0,t)=\frac{\partial u}{\partial x}(1,t), \quad t \in (0,6] \\
&u(x,0)=f(x), \quad x \in (0,1) 
\end{cases}
\end{equation}
where \(u(x,t)\) denotes the displacement of the wave field to be solved, its second-order time derivative represents acceleration proportional to the spatial curvature, \(c=0.1\) is the wave speed determined by the elastic modulus and density of the medium, and the boundary conditions are non-reflective periodic boundary conditions to ensure no energy loss or boundary reflection when the wave circulates within the computational domain. The initial displacement field is generated using a Whittle-Matérn Gaussian random field (GRF)~\cite{6Lei2025}, constructed as follows: based on the finite difference Laplacian operator \(\Delta\) with periodic boundary conditions, we construct the symmetric positive definite operator $-\Delta + \sigma^2 I$($\sigma=7.0$ is the length scale parameter and $I$ is the identity matrix), reconstruct the covariance matrix via eigenvalue decomposition, and finally sample from a multivariate normal distribution with zero mean and a covariance matrix equal to $\sigma^4$times the reconstructed matrix. This method naturally satisfies the continuity of periodic boundary conditions, strictly guarantees the positive definiteness of the covariance matrix and sampling stability, and precisely controls the smoothness of the random field to match the regularity requirements of the wave equation. The initial velocity field is uniformly set to 0. A total of 1200 trajectories are generated, with 80\% randomly selected for training and the remaining 20\% for testing. The objective is to learn the operator $\mathcal{G}: u_0(x) \mapsto u(x,t)$ that maps the initial displacement field to the exact solution of the equation. The loss function is formulated as:

\begin{equation}\label{eq:loss3}
	\small
\begin{aligned}
	\mathcal{L}(\theta)=&\frac{\alpha_f}{BN_f}\sum_{b=1}^B \sum_{i=1}^{N_f}\left\|\frac{\partial^2 \mathcal{G}_\theta\left(f_{b},x_i,t_i\right)}{\partial t^2}-c^2\frac{\partial^2\mathcal{G}_\theta\left(f_{b},x_i,t_i\right)}{\partial x^2}\right\|^2 \\
	+&\frac{\alpha_b}{B N_b} \sum_{b=1}^B \sum_{j=1}^{N_b}\left\|\mathcal{G}_\theta\left(f_{b},0,t_j\right)-\mathcal{G}_\theta\left(f_{b},1,t_j\right)\right\|^2+\left\|\frac{\partial \mathcal{G}_\theta\left(f_{b},0,t_j\right)}{\partial x}-\frac{\partial \mathcal{G}_\theta\left(f_{b},1,t_j\right)}{\partial x} \right\|^2\\
	+&\frac{\alpha_i}{B N_i} \sum_{b=1}^B  \sum_{p=1}^{N_i}\left\|\mathcal{G}_\theta\left(f_{b},x_p,0\right)-f_b(x_p)\right\|^2  \\
	+&\frac{\alpha_d}{B N_d} \sum_{b=1}^B  \sum_{k=1}^{N_d}\left\|\mathcal{G}_\theta\left(f_{b},x_k,t_k\right)-u\left(x_k,t_k\right)\right\|^2 \\
\end{aligned}    
\end{equation}
where the meaning of each parameter is consistent with Eq.~\eqref{eq:loss1}. The weights for each loss term in this case are set to: $\alpha_f=200$, $\alpha_b=50$, $\alpha_i=500$, and $\alpha_d=50$. We randomly sample 5000 points inside the spatiotemporal domain, uniformly select displacement data at 101 points at the initial time as input, randomly sample 200 points on the boundary \(x=0\) and 200 corresponding points on the boundary \(x=1\) to satisfy the periodic condition, and use the displacement at 101 uniformly sampled points at the final time as observation data. For the PI-DeepONet framework, both the branch network and trunk network are set to 4 fully connected layers with a hidden dimension of 64, and 32 basis functions and corresponding coefficients. For the PI-SWNO framework, the input branch, temporal branch, and spatial branch are all set to 4 fully connected layers, with a hidden dimension of 64 for the input branch, 32 for the temporal and spatial branches, and 32 basis functions and corresponding coefficients. Both models are trained for 20,000 iterations with a batch size of 32.

Fig.~\ref{fig:wave-fig1} shows the ground truth and prediction results for 3 test samples, where the overall absolute error amplitude of PI-SWNO is lower than that of PI-DeepONet, with similar error distribution patterns for the two models.

\begin{figure}[htbp]
	\centering
	\includegraphics[width=0.98\linewidth]{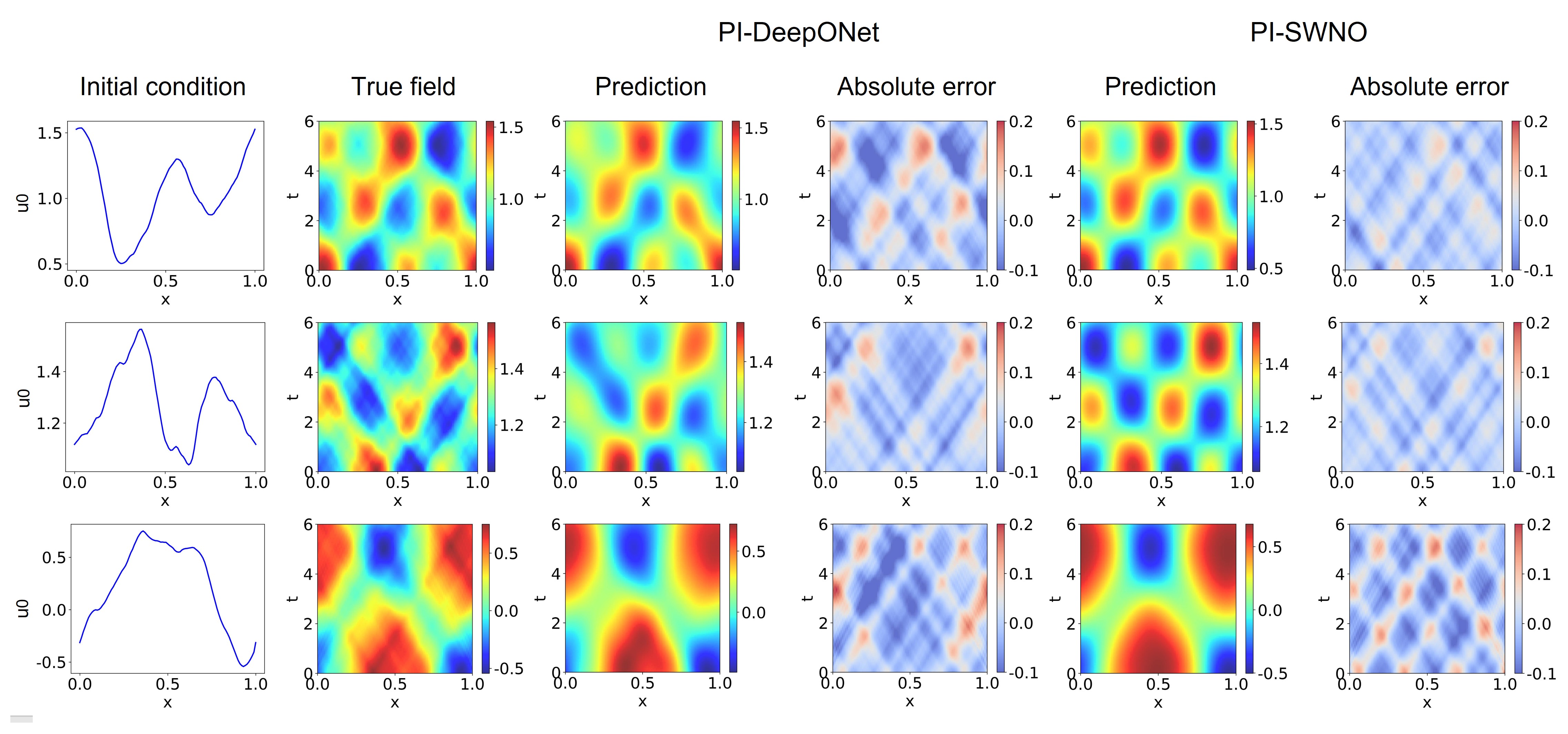}
	\caption{1D wave equation: The first column shows the randomly sampled initial condition; the second column shows the high-fidelity numerical solution; the third and fourth columns show the predicted solution and its absolute error from PI-DeepONet; the fifth and sixth columns show the predicted solution and its absolute error from PI-SWNO.}
	\label{fig:wave-fig1}
\end{figure}

Fig.~\ref{fig:wave-fig3} compares the evolution of errors with time domain span for the two models. For such periodic evolution problems, the error of the spatiotemporally coupled PI-DeepONet increases overall with the length of the time domain, and undergoes an abrupt catastrophic jump between 8 s and 10 s: the median error rises from about 5\% to over 25\% , the whisker range expands sharply to 40\%, accompanied by a large number of high-error outliers. In contrast, the error of the spatiotemporally decoupled PI-SWNO grows smoothly with no abrupt jumps throughout the time horizon, with the median error consistently below 10\% and a compact distribution. Fig.~\ref{fig:wave-fig2} provides a more detailed view of the statistical characteristics of MSE and ANRL2E for the two models on the 1D wave equation dataset in the short time horizon $[0, 2]$ and long time horizon $[0, 20]$, where PI-SWNO yields a smaller error standard deviation and lower extreme values. These results demonstrate that the decoupled architecture of PI-SWNO accurately separates spatiotemporal features, avoids phase drift and mode aliasing common in periodic problems, and achieves a lower average error and more stable fitting performance.

\begin{figure}[htbp]
	\centering
	\includegraphics[width=0.85\linewidth]{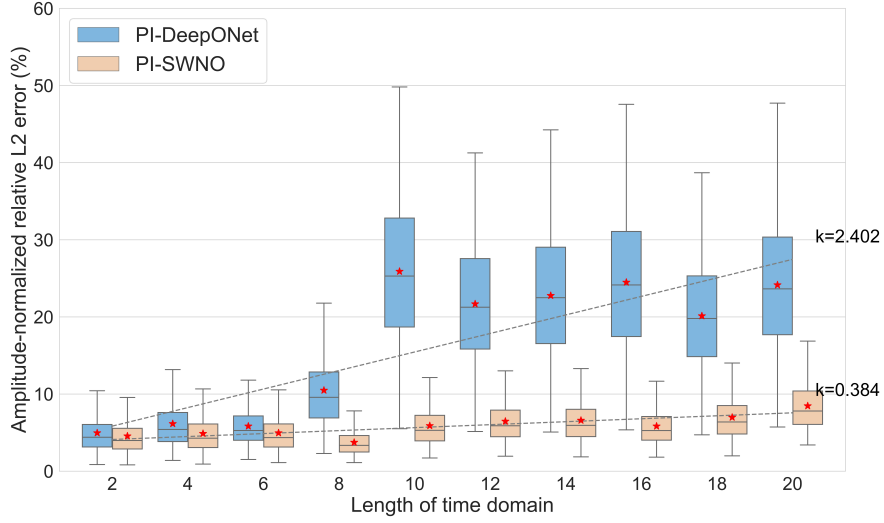}
	\caption{Comparison of long-time ANRL2E growth trends between the baseline PI-DeepONet and the proposed spatiotemporally decoupled PI-SWNO across expanding temporal domain spans on the 1D wave equation.}
	\label{fig:wave-fig3}
\end{figure}

\begin{figure}[htbp]
  \centering
  \includegraphics[width=0.98\linewidth]{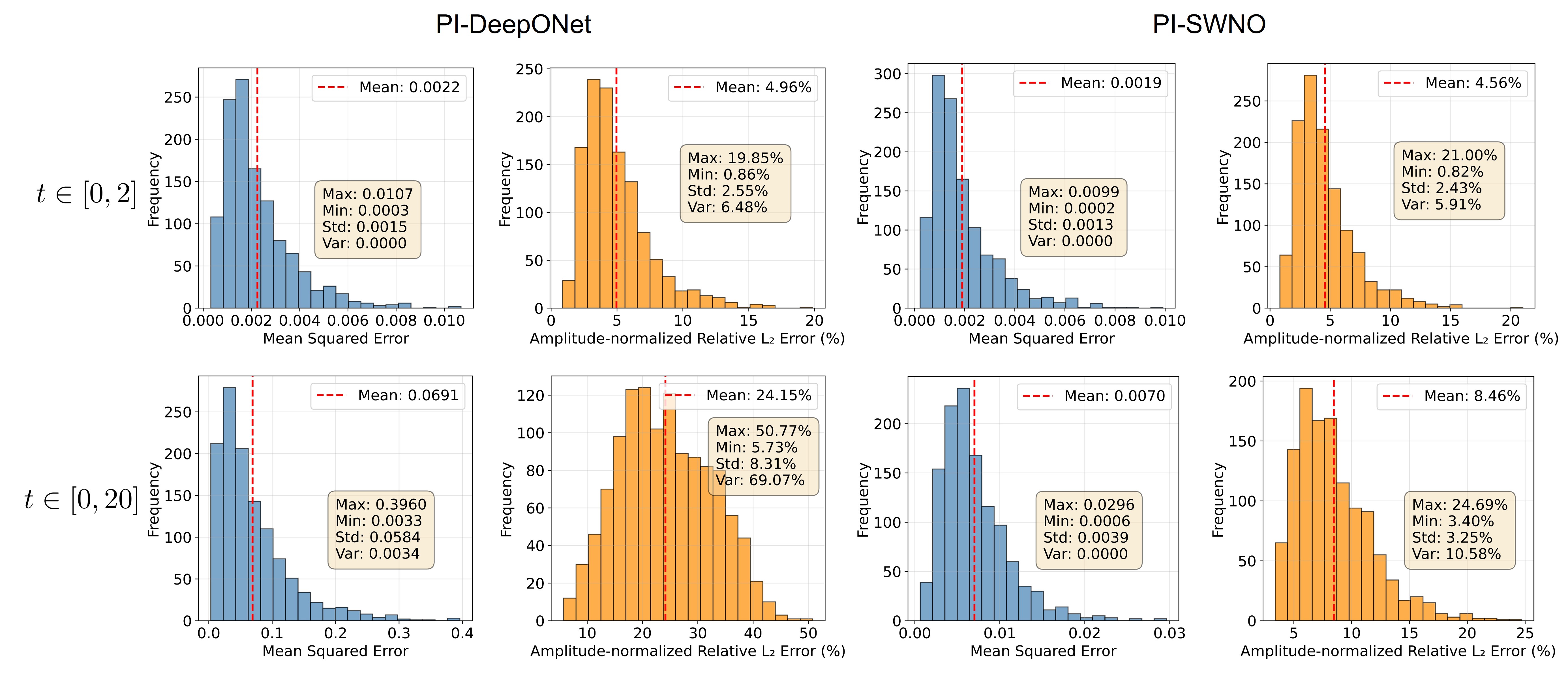}
  \caption{Statistical characteristics of MSE and ANRL2E for PI-DeepONet and PI-SWNO on the 1D wave equation dataset in the short time horizon $[0, 2]$ and long time horizon $[0, 20]$, respectively.}
  \label{fig:wave-fig2}
\end{figure}

\subsection{2-D wave equations}\label{case4}

Complementing the 1D wave benchmark, we extend our validation to the two-dimensional linear wave equation:

\begin{equation}
\begin{cases}
\frac{\partial^2 u}{\partial t^2}=c^2\left(\frac{\partial^2 u}{\partial x^2}+\frac{\partial^2 u}{\partial y^2}\right), & (x, y, t) \in[0,1] \times[0,1] \times(0,10] \\ 
u(0, y, t)=u(1, y, t),\frac{\partial u}{\partial x}(0, y, t)=\frac{\partial u}{\partial x}(1, y, t),   & (y, t) \in[0,1] \times(0,10] \\ 
u(x, 0, t)=u(x, 1, t), \frac{\partial u}{\partial y}(x, 0, t)=\frac{\partial u}{\partial y}(x, 1, t), & (x, t) \in[0,1] \times(0,10] \\ 
u(x, y, 0)=f(x, y), & (x, y) \in(0,1) \times(0,1)
\end{cases}
\end{equation}
where the meaning of each parameter is consistent with the 1D wave equation in Eq.~\eqref{eq:case3}, the wave speed is set to \(c=0.1\), the boundary conditions are periodic, the initial displacement field is generated using a Gaussian random field, and the initial velocity field is set to 0. A total of 1200 trajectories are generated, with 80\% randomly selected for training and the remaining 20\% for testing. The objective is to learn the operator $\mathcal{G}: u_0(x) \mapsto u(x,t)$ that maps the initial displacement field to the exact solution of the equation, with the loss function referring to Eq.~\eqref{eq:loss3}. The weights for each loss term in this case are set to: $\alpha_f=50$, $\alpha_b=50$, $\alpha_i=100$, and $\alpha_d=100$.

We randomly sample 8192 points inside the spatiotemporal domain \(\Omega\), uniformly select displacement data at 676 points at the initial time as input, randomly sample 1024 points on each of the two boundary surfaces \(x=0\) and \(y=0\) and 1024 corresponding points on the boundary surfaces \(x=1\) and \(y=1\) to satisfy the periodic condition, and use the displacement data at 676 uniformly sampled points at the final time as observation data. For the PI-DeepONet framework, both the branch network and trunk network are set to 6 fully connected layers, with a hidden dimension of 100 for the branch network, 200 for the trunk network, and 64 basis functions and corresponding coefficients. For the PI-SWNO framework, the input branch, temporal branch, and spatial branch are all set to 6 fully connected layers with a hidden dimension of 100, and 64 basis functions and corresponding coefficients. Both models are trained for 20,000 iterations with a batch size of 8.

As shown in Fig.~\ref{fig:wave2d-fig1}, both PI-DeepONet and PI-SWNO can accurately reproduce the spatiotemporal evolution of the wave field, precisely capturing the spatial position of the wave peak, propagation trend, and amplitude variation characteristics. This confirms that both physics-informed neural operators can effectively learn the intrinsic physical constraints of the wave equation. From the comparison of error fields, PI-DeepONet exhibits obvious error accumulation in regions with large wave amplitude gradients, while PI-SWNO yields lower absolute error values.

\begin{figure}[htbp]
	\centering
	\includegraphics[width=0.98\linewidth]{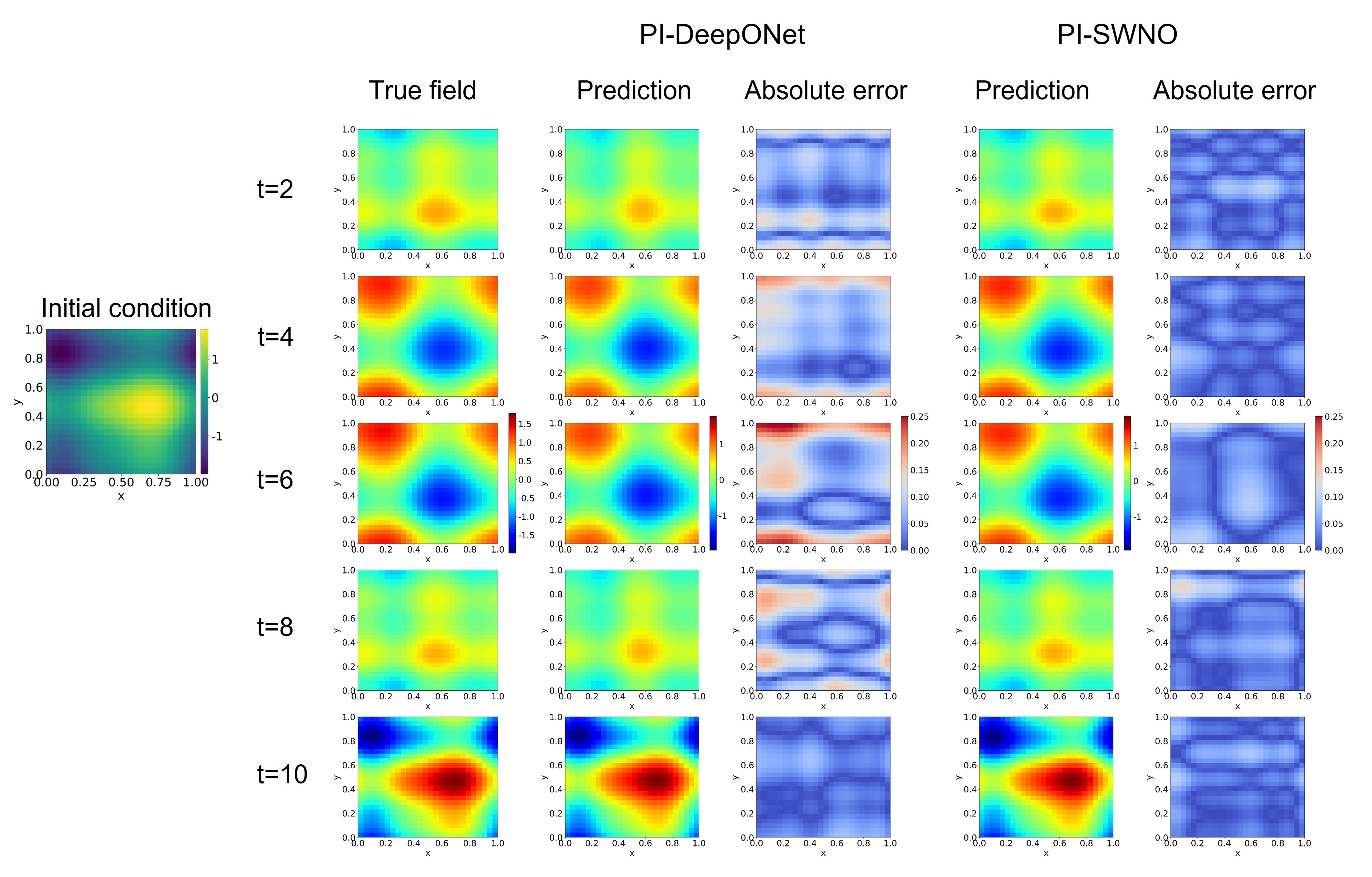}
	\caption{2D wave equation: The first plot shows the randomly sampled initial condition; the second column shows 2D slices of the high-fidelity numerical solution at $t = 2, 4, 6, 8, 10$; the third and fourth columns show the predicted solution and its absolute error from PI-DeepONet; the fifth and sixth columns show the predicted solution and its absolute error from PI-SWNO.}
	\label{fig:wave2d-fig1}
\end{figure}

As shown in Fig.~\ref{fig:wave2d-fig3}, the prediction errors of both models increase with the extension of the prediction time horizon, consistent with the non-decreasing fitting error property of fixed-parameter neural operators. In the short time horizon of 2 s, PI-DeepONet and PI-SWNO achieve comparable accuracy, with the performance gap widening continuously as the time horizon extends. Least-squares linear regression on the mean error shows that the error growth rate of PI-DeepONet is approximately 4 times that of PI-SWNO (1.582 vs. 0.620, respectively).

\begin{figure}[htbp]
	\centering
	\includegraphics[width=0.85\linewidth]{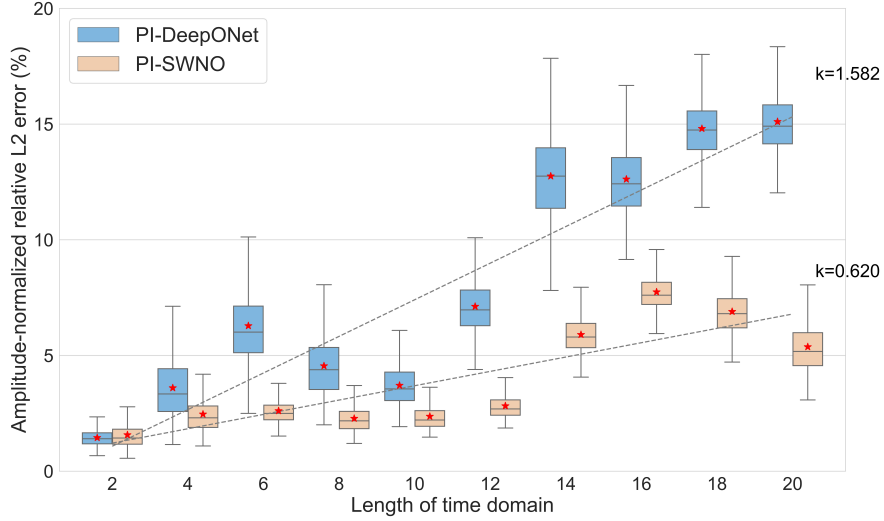}
	\caption{Comparison of long-time ANRL2E growth trends between the baseline PI-DeepONet and the proposed spatiotemporally decoupled PI-SWNO across expanding temporal domain spans on the 2D wave equation.}
	\label{fig:wave2d-fig3}
\end{figure}

As shown in Fig.~\ref{fig:wave2d-fig2}, PI-SWNO demonstrates significant advantages in periodic problems. In the short time horizon $[0, 2]$, the two models achieve comparable prediction accuracy, with PI-DeepONet yielding a slightly lower mean relative error and maximum error, while PI-SWNO has a more concentrated relative error distribution and smaller fluctuations. However, in the long time horizon $[0, 20]$, the mean error of PI-DeepONet increases to 15.11\%, three times that of PI-SWNO.

\begin{figure}[htbp]
  \centering
  \includegraphics[width=0.98\linewidth]{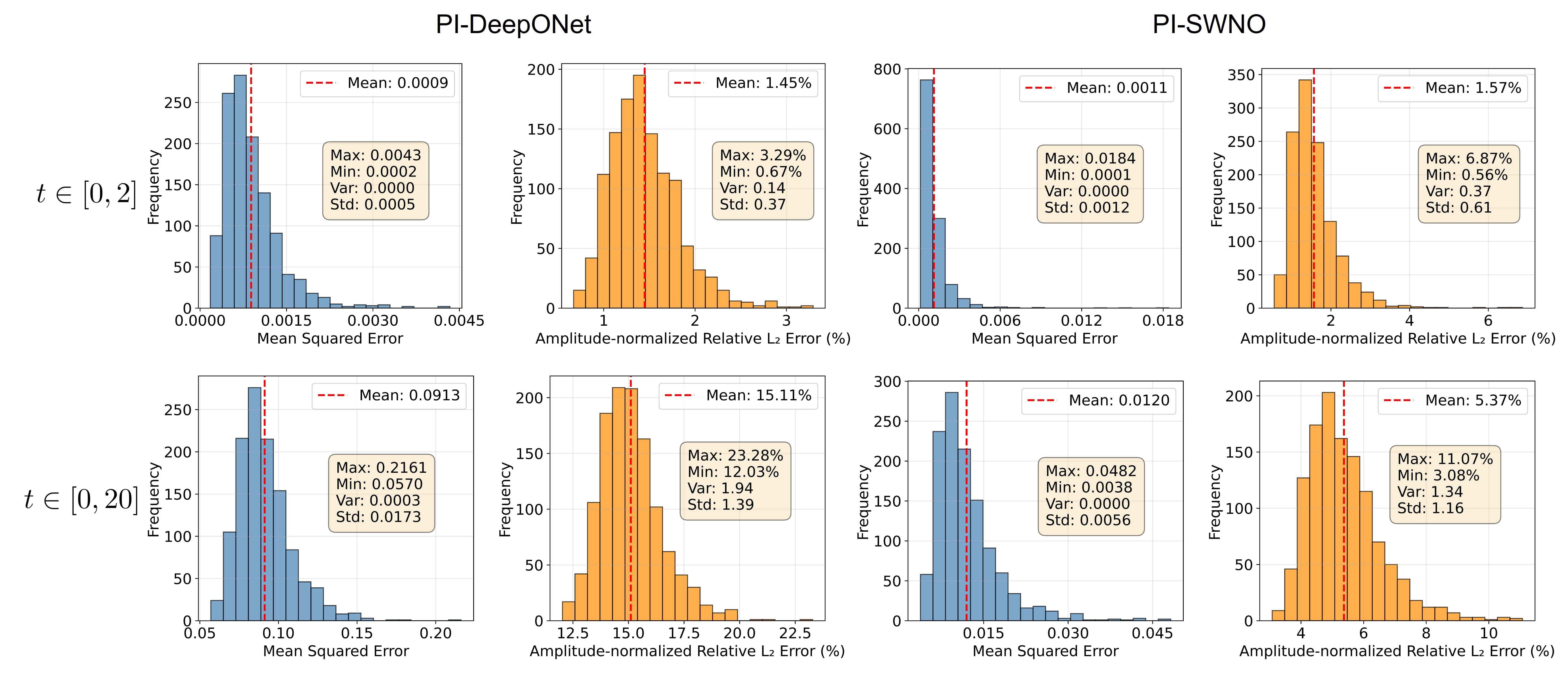}
  \caption{Statistical characteristics of MSE and ANRL2E for PI-DeepONet and PI-SWNO on the 2D wave equation dataset in the short time horizon $[0, 2]$ and long time horizon $[0, 20]$, respectively.}
  \label{fig:wave2d-fig2}
\end{figure}

The above four numerical cases validate the remarkable superiority of the spatiotemporally decoupled architecture of PI-SWNO in solving linear time-dependent PDEs. These linear autonomous PDEs admit solutions that can be rigorously or approximately decomposed into a linear superposition of time-invariant spatial basis functions and time-varying evolution coefficients. The decoupled representation of PI-SWNO, with separate input, spatial, and temporal branches, is mathematically isomorphic to the solution structure of such PDEs, thus enabling higher accuracy and slower error growth over long time horizons.
For strongly nonlinear unsteady PDEs such as the KdV and Burgers equations, the solutions involve competitive effects between nonlinear advection, dispersion, and diffusion terms, with stronger co-evolution of spatiotemporal features. While the universal approximation property of the spatiotemporally decoupled basis functions of the PI-SWNO model is proven in~\ref{appendix:piswno_approximation}, experimental validation is still required to investigate its performance in handling long-time nonlinear PDEs. We therefore design the following three nonlinear benchmark cases.

\subsection{1-D KdV equations}\label{case5}

For our next benchmark, we use the classical KdV equation with periodic boundary conditions:

\begin{equation}
\begin{cases}
&\frac{\partial u}{\partial t}+\eta u \frac{\partial u}{\partial x}+\gamma\frac{\partial^3 u}{\partial x^3}=0, \quad (x,t) \in [0,1] \times (0,1]\\
&u(0,t)=u(1,t), \quad t \in (0,1]  \\
& \frac{\partial u}{\partial x}(0,t)=\frac{\partial u}{\partial x}(1,t), \quad t \in (0,1] \\
&u(x,0)=f(x), \quad x \in (0,1) 
\end{cases}
\end{equation}
where \(\eta=1\) is the nonlinear coefficient, \(\gamma=\num{1e-4}\) is the dispersion coefficient. This equation describes the propagation and interaction of nonlinear waves, with the competition between the nonlinear term and dispersion term determined by \(K=\frac{\eta U L^2}{\gamma}\): \(L\) is the characteristic length of the wavelength and \(U\) is the characteristic amplitude of the wave. Nonlinear effects dominate when \(K \gg 1\), while dispersion effects dominate when \(K \ll 1\). \(u(x,t)\) denotes the wave field to be solved, with periodic boundary conditions. The initial field is generated using a periodic Gaussian random field, constructed as follows: standard Gaussian white noise is generated on a uniform structured grid and transformed to the frequency domain via fast Fourier transform (FFT); a 2D angular wavenumber grid matched to the computational mesh is constructed, and a Gaussian power spectral density matched to the length scale \(\sigma=0.2\) is defined directly in the frequency domain to perform weighted filtering on the frequency domain components of the white noise, with a minimum value constraint introduced to ensure non-negative power spectrum and avoid numerical singularities; the real-valued spatial random field is obtained via inverse 2D FFT. The initial velocity field is uniformly set to 0. A total of 1200 trajectories are generated, with 80\% randomly selected for training and the remaining 20\% for testing. The objective is to learn the operator $\mathcal{G}: u_0(x) \mapsto u(x,t)$ that maps the initial displacement field to the exact solution of the equation. The loss function is formulated as:

\begin{equation}
	\small
\begin{aligned}
	\mathcal{L}(\theta)=&\frac{\alpha_f}{BN_f}\sum_{b=1}^B \sum_{i=1}^{N_f}\left\|\frac{\partial \mathcal{G}_\theta\left(f_{b},x_i,t_i\right)}{\partial t}+\eta \mathcal{G}_\theta\left(f_{b},x_i,t_i\right)\frac{\partial \mathcal{G}_\theta\left(f_{b},x_i,t_i\right)}{\partial x}+\gamma\frac{\partial^3\mathcal{G}_\theta\left(f_{b},x_i,t_i\right)}{\partial x^3}\right\|^2 \\
	+&\frac{\alpha_b}{B N_b} \sum_{b=1}^B \sum_{j=1}^{N_b}\left\|\mathcal{G}_\theta\left(f_{b},0,t_j\right)-\mathcal{G}_\theta\left(f_{b},1,t_j\right)\right\|^2+\left\|\frac{\partial \mathcal{G}_\theta\left(f_{b},0,t_j\right)}{\partial x}-\frac{\partial \mathcal{G}_\theta\left(f_{b},1,t_j\right)}{\partial x} \right\|^2\\
	+&\frac{\alpha_i}{B N_i} \sum_{b=1}^B  \sum_{p=1}^{N_i}\left\|\mathcal{G}_\theta\left(f_{b},x_p,0\right)-f_b(x_p)\right\|^2  \\
	+&\frac{\alpha_d}{B N_d} \sum_{b=1}^B  \sum_{k=1}^{N_d}\left\|\mathcal{G}_\theta\left(f_{b},x_k,t_k\right)-u\left(x_k,t_k\right)\right\|^2 \\
\end{aligned}
\end{equation}
where the meaning of each parameter is consistent with the loss function of the 1D heat conduction case in Eq.~\eqref{eq:loss1}. The weights for each loss term in this case are set to: $\alpha_f=100$, $\alpha_b=10$, $\alpha_i=500$, and $\alpha_d=10$. We randomly sample 8192 points inside the spatiotemporal domain, uniformly select displacement data at 513 points at the initial time as input, randomly sample 200 points on the boundary \(x=0\) and 200 corresponding points on the boundary \(x=1\) to satisfy the periodic condition, and use the displacement at 513 uniformly sampled points at the final time as observation data. For the PI-DeepONet framework, both the branch network and trunk network are set to 4 fully connected layers, with a hidden dimension of 128 for the branch network, 256 for the trunk network, and 64 basis functions and corresponding coefficients. For the PI-SWNO framework, the input branch, temporal branch, and spatial branch are all set to 4 fully connected layers with a hidden dimension of 128, and 64 basis functions and corresponding coefficients. Both models are trained for 20,000 iterations with a batch size of 8.

The core feature of the KdV equation is the competitive effect between the nonlinear advection term and the third-order dispersion term: the nonlinear term tends to steepen the wave front, while the dispersion term tends to broaden the wave packet, and their dynamic balance forms steep gradient structures and solitary wave solutions. Based on the strong spatiotemporal coupling characteristics of nonlinear PDEs, it was theoretically expected that the fitting efficiency of the spatiotemporally decoupled architecture might decrease. However, the experimental results show that PI-SWNO does not exhibit an obvious performance bottleneck: its overall prediction accuracy is comparable to that of PI-DeepONet, and it demonstrates superior fitting capability in local high-gradient regions. As shown in Fig.~\ref{fig:kdv-fig1}, the predicted field of PI-SWNO is consistent with the ground truth in terms of spatial distribution, evolution trend, and amplitude, while the absolute error magnitude of PI-DeepONet is approximately twice that of PI-SWNO.

\begin{figure}[htbp]
	\centering
	\includegraphics[width=0.98\linewidth]{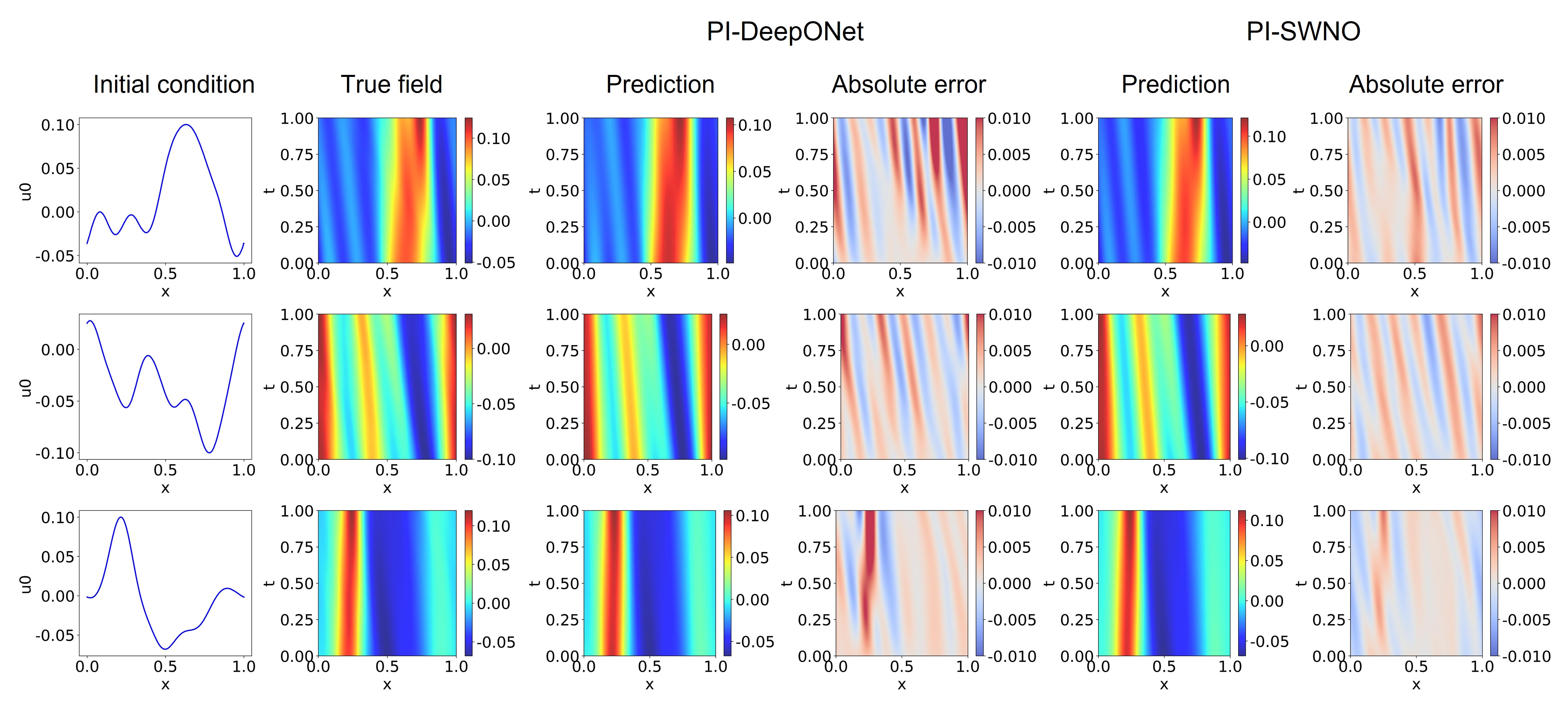}
	\caption{1D KdV equation: The first column shows the randomly sampled initial condition; the second column shows the high-fidelity numerical solution; the third and fourth columns show the predicted solution and its absolute error from PI-DeepONet; the fifth and sixth columns show the predicted solution and its absolute error from PI-SWNO.}
	\label{fig:kdv-fig1}
\end{figure}

\begin{figure}[htbp]
	\centering
	\includegraphics[width=0.85\linewidth]{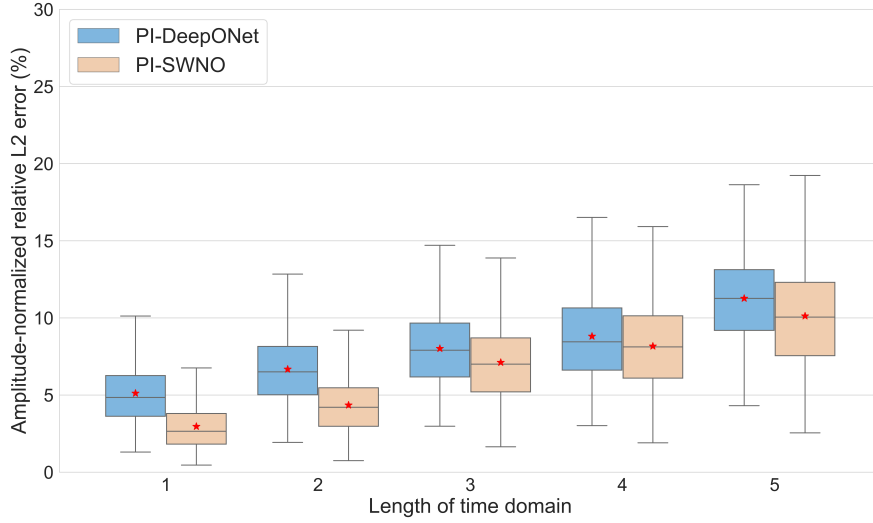}
	\caption{Comparison of long-time ANRL2E growth trends between the baseline PI-DeepONet and the proposed spatiotemporally decoupled PI-SWNO across expanding temporal domain spans on the 1D KdV equation.}
	\label{fig:kdv-fig3}
\end{figure}

As shown in Fig.~\ref{fig:kdv-fig3}, the prediction errors of both models show a gradual upward trend with the extension of the prediction time horizon, with an obvious temporal error accumulation phenomenon. Although the error accumulation rate of PI-SWNO is slightly faster than that of PI-DeepONet in the strongly spatiotemporally coupled KdV case, and it does not show a stability advantage in long-time prediction, the mean and median errors of PI-SWNO are consistently lower than those of PI-DeepONet across all 5 time domain nodes, maintaining a continuous and stable accuracy advantage. Fig.~\ref{fig:kdv-fig2} shows the statistical characteristics of MSE and ANRL2E for the two models on the 1D KdV equation dataset in the short time horizon $[0, 1]$ and long time horizon $[0, 5]$.

\begin{figure}[htbp]
	\centering
	\includegraphics[width=0.98\linewidth]{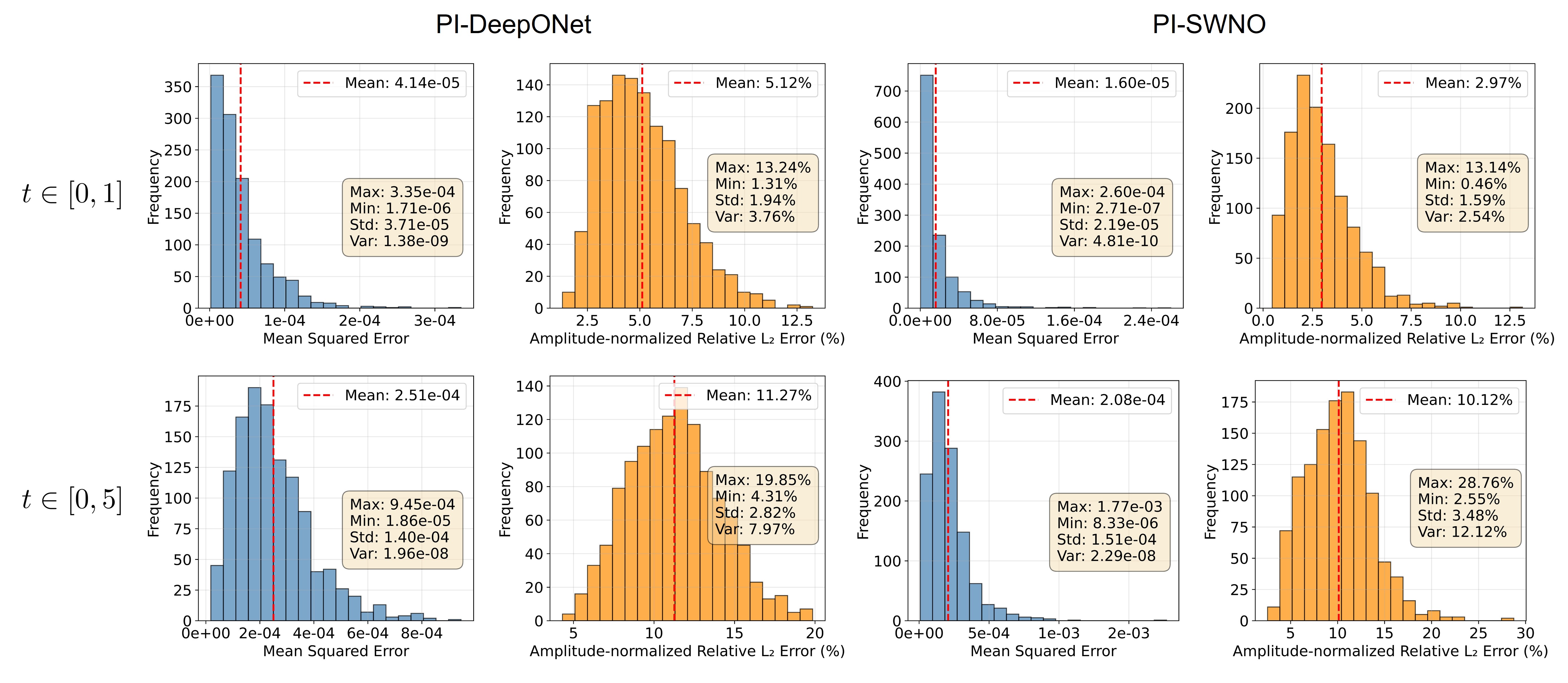}
  \caption{Statistical characteristics of MSE and ANRL2E for PI-DeepONet and PI-SWNO on the 1D KdV equation dataset in the short time horizon $[0, 1]$ and long time horizon $[0, 5]$, respectively.}
  \label{fig:kdv-fig2}
\end{figure}

\subsection{1-D Burgers equations}\label{case6}

Next, we examine a one-dimensional viscous Burgers equation with periodic boundary conditions:

\begin{equation}\label{eq:case5}
\begin{cases}
& \frac{\partial u}{\partial t}+u \frac{\partial u}{\partial x}-\nu \frac{\partial^2 u}{\partial x^2}=0, \quad (x,t) \in[0,1] \times(0,1] \\
& u(0, t)=u(1, t), \quad  t\in (0,1]\\
& \frac{\partial u}{\partial x}(0, t)=\frac{\partial u}{\partial x}(1, t), \quad  t\in (0,1]\\
& u(x, 0)=f(x), \quad x \in(0,1] 
\end{cases}
\end{equation}
where \(u(x,t)\) denotes the velocity field to be solved. This equation describes the long-term competitive effect between the nonlinear advection term and the viscous diffusion term, with the dominant mechanism distinguished by the Reynolds number \(Re=\frac{UL}{\nu}\): \(U\) is the characteristic velocity, \(L\) is the characteristic length, and \(\nu=0.01\) is the viscosity coefficient. Advection dominates when \(Re \gg 1\), while diffusion dominates when \(Re \ll 1\). Periodic boundary conditions are set, the initial velocity field is generated using a periodic Gaussian random field, and the initial acceleration field is set to 0. A total of 1200 trajectories are generated, with 80\% randomly selected for training and the remaining 20\% for testing. The objective is to learn the operator $\mathcal{G}: u_0(x) \mapsto u(x,t)$ that maps the initial displacement field to the exact solution of the equation. The loss function is formulated as:

\begin{equation}\label{eq:loss5}
	\small
\begin{aligned}
	\mathcal{L}(\theta)=&\frac{\alpha_f}{BN_f}\sum_{b=1}^B \sum_{i=1}^{N_f}\left\|\frac{\partial \mathcal{G}_\theta\left(f_{b},x_i,t_i\right)}{\partial t}+ \mathcal{G}_\theta\left(f_{b},x_i,t_i\right)\frac{\partial \mathcal{G}_\theta\left(f_{b},x_i,t_i\right)}{\partial x}-\nu\frac{\partial^2\mathcal{G}_\theta\left(f_{b},x_i,t_i\right)}{\partial x^2}\right\|^2 \\
	+&\frac{\alpha_b}{B N_b} \sum_{b=1}^B \sum_{j=1}^{N_b}\left\|\mathcal{G}_\theta\left(f_{b},0,t_j\right)-\mathcal{G}_\theta\left(f_{b},1,t_j\right)\right\|^2+\left\|\frac{\partial \mathcal{G}_\theta\left(f_{b},0,t_j\right)}{\partial x}-\frac{\partial \mathcal{G}_\theta\left(f_{b},1,t_j\right)}{\partial x} \right\|^2\\
	+&\frac{\alpha_i}{B N_i} \sum_{b=1}^B  \sum_{p=1}^{N_i}\left\|\mathcal{G}_\theta\left(f_{b},x_p,0\right)-f_b(x_p)\right\|^2  \\
	+&\frac{\alpha_d}{B N_d} \sum_{b=1}^B  \sum_{k=1}^{N_d}\left\|\mathcal{G}_\theta\left(f_{b},x_k,t_k\right)-u\left(x_k,t_k\right)\right\|^2 \\
\end{aligned}    
\end{equation}
where the meaning of each parameter is consistent with Eq.~\eqref{eq:loss1}. The weights for each loss term in this case are set to: $\alpha_f=100$, $\alpha_b=50$, $\alpha_i=500$, and $\alpha_d=200$. We randomly sample 4800 points inside the spatiotemporal domain, uniformly select velocity data at 101 points at the initial time as input, randomly sample 256 points on the boundary \(x = 0\) and 256 corresponding points on the boundary \(x = 1\) to satisfy the periodic condition, and use the displacement at 101 uniformly sampled points at the final time as observation data. For the PI-DeepONet framework, both the branch network and trunk network are set to 4 fully connected layers with a hidden dimension of 64, and 32 basis functions and corresponding coefficients. For the PI-SWNO framework, the input branch, temporal branch, and spatial branch are all set to 4 fully connected layers, with a hidden dimension of 64 for the input branch, 32 for the temporal and spatial branches, and 32 basis functions and corresponding coefficients. Both models are trained for 50,000 iterations with a batch size of 128.

\begin{figure}[htbp]
	\centering
	\includegraphics[width=0.98\linewidth]{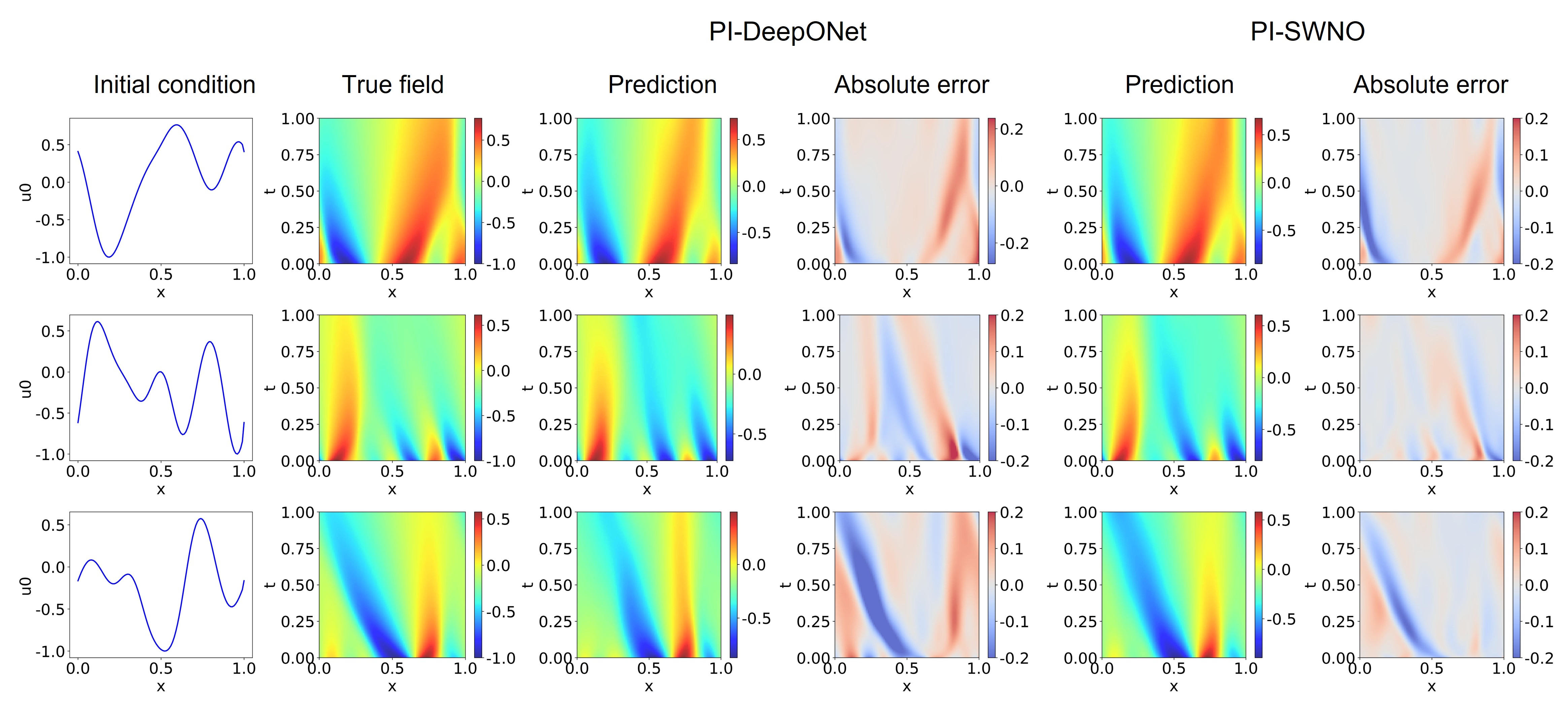}
	\caption{1D Burgers equation: The first column shows the randomly sampled initial condition; the second column shows the high-fidelity numerical solution; the third and fourth columns show the predicted solution and its absolute error from PI-DeepONet; the fifth and sixth columns show the predicted solution and its absolute error from PI-SWNO.}
	\label{fig:burgers1d-fig1}
\end{figure}

\begin{figure}[htbp]
	\centering
	\includegraphics[width=0.85\linewidth]{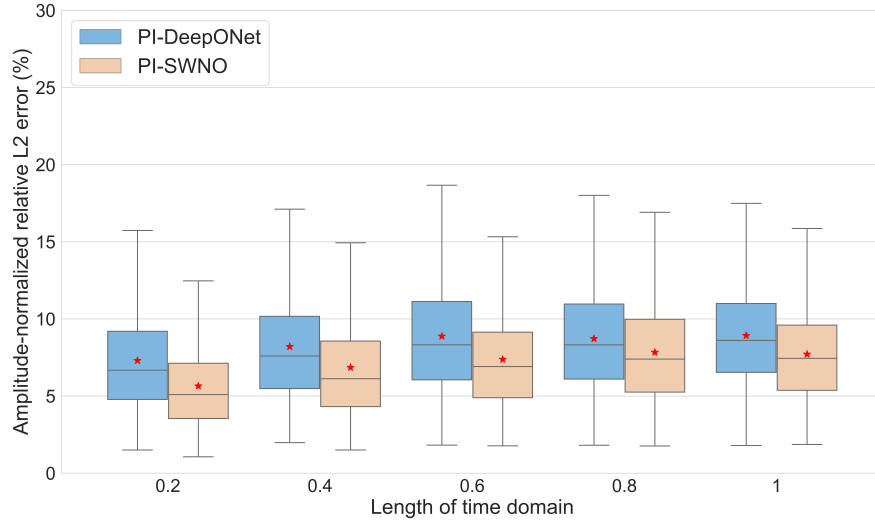}
	\caption{Comparison of long-time ANRL2E growth trends between the baseline PI-DeepONet and the proposed spatiotemporally decoupled PI-SWNO across expanding temporal domain spans on the 1D Burgers equation.}
	\label{fig:burgers1d-fig3}
\end{figure}

\begin{figure}[htbp]
  \centering
  \includegraphics[width=0.98\linewidth]{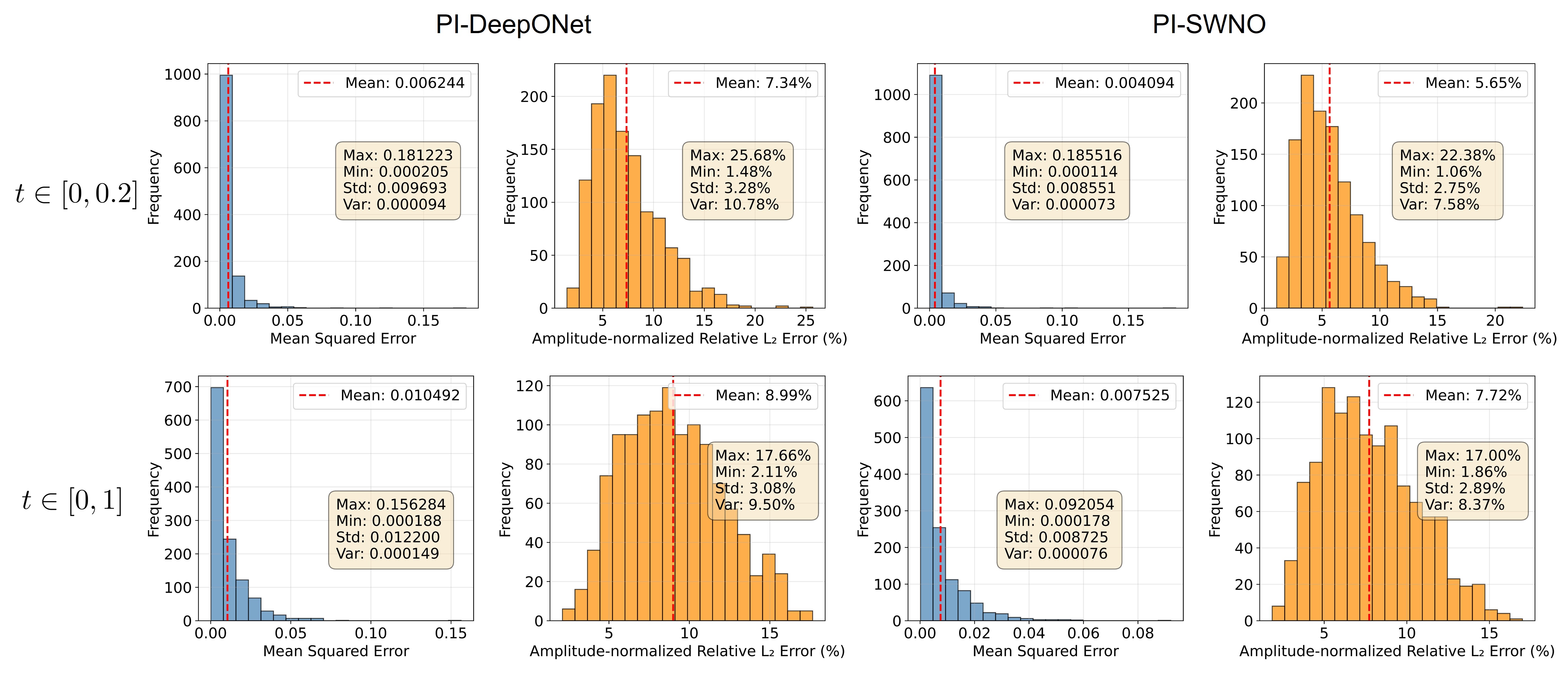}
  \caption{Statistical characteristics of MSE and ANRL2E for PI-DeepONet and PI-SWNO on the 1D Burgers equation dataset in the short time horizon $[0, 0.2]$ and long time horizon $[0, 1]$, respectively.}
  \label{fig:burgers1d-fig2}
\end{figure}

As shown in Fig.~\ref{fig:burgers1d-fig1}, the solution of the equation forms an obliquely propagating shock structure over time, with high-gradient regions concentrated at the wave front. Both models can accurately predict the spatiotemporal evolution of the velocity field, but exhibit varying degrees of morphological deviation and structural blurring, especially in local regions with large gradients. Contrary to expected accuracy degradation, PI-SWNO does not show deterioration in its ability to capture steep gradients; instead, its overall absolute error is lower, it has stronger capture capability for steep structures, and its prediction results have better consistency and stability. 

As shown in Fig.~\ref{fig:burgers1d-fig3}, the PI-SWNO model does not exhibit rapid error growth with time extension when handling strong spatiotemporally coupled shock evolution problems, and slightly outperforms PI-DeepONet at every time domain length. As shown in Fig.~\ref{fig:burgers1d-fig2}, the mean MSE of the PI-SWNO model is 44.5\% lower than that of PI-DeepONet, and the mean ANRL2E is reduced by 26.1\%, indicating that it not only achieves higher average accuracy but also performs more stably across samples with different initial conditions, with stronger generalization ability.

\subsection{2-D Burgers equations}\label{case7}

Our final benchmark is the two-dimensional viscous Burgers equation with complex shock interactions:

\begin{equation}
	\small
\begin{cases}
& \frac{\partial u}{\partial t} + u (\frac{\partial u}{\partial x} + \frac{\partial u}{\partial y}) - \nu \left( \frac{\partial^2 u}{\partial x^2} + \frac{\partial^2 u}{\partial y^2} \right) = 0, \quad (x, y, t) \in [0, 1] \times [0, 1] \times (0, 1] \\
& u(0, y, t) = u(1, y, t), \quad \frac{\partial u}{\partial x}(0, y, t) = \frac{\partial u}{\partial x}(1, y, t), \quad  (y,t) \in[0,1] \times(0,1]\\
& u(x, 0, t) = u(x, 1, t), \quad \frac{\partial u}{\partial y}(x, 0, t) = \frac{\partial u}{\partial y}(x, 1, t), \quad  (x,t) \in[0,1] \times(0,1]\\
& u(x, y, 0) = f(x, y), \quad (x, y) \in [0, 1] \times [0, 1]
\end{cases}
\end{equation}
where the meaning of each parameter is consistent with Eq.~\eqref{eq:case5}, the initial velocity field is generated using a periodic Gaussian random field, and the initial acceleration field is uniformly set to 0. A total of 1200 trajectories are generated, with 80\% randomly selected for training and the remaining 20\% for testing. The objective is to learn the operator \(\mathcal{G}:u_0(x,y) \mapsto u(x,y,t)\) that maps the initial velocity field to the exact solution of the equation, with the loss function referring to Eq.~\eqref{eq:loss5}.

The weights for each loss term in this case are set to: $\alpha_f=50$, $\alpha_b=20$, $\alpha_i=150$, and $\alpha_d=5$. We randomly sample 8192 points inside the spatiotemporal domain, uniformly select velocity data at 676 points at the initial time as input, randomly sample 1024 points on each of the two boundary surfaces \(x=0\) and \(y=0\) and 1024 corresponding points on the boundary surfaces \(x=1\) and \(y=1\)to satisfy the periodic condition, and use the velocity data at 676 uniformly sampled points at the final time as observation data. For the PI-DeepONet framework, both the branch network and trunk network are set to 6 fully connected layers, with a hidden dimension of 100 for the branch network, 200 for the trunk network, and 64 basis functions and corresponding coefficients. For the PI-SWNO framework, the input branch, temporal branch, and spatial branch are all set to 6 fully connected layers with a hidden dimension of 100, and 64 basis functions and corresponding coefficients. Both models are trained for 50,000 iterations with a batch size of 8.

\begin{figure}[htbp]
	\centering
	\includegraphics[width=0.98\linewidth]{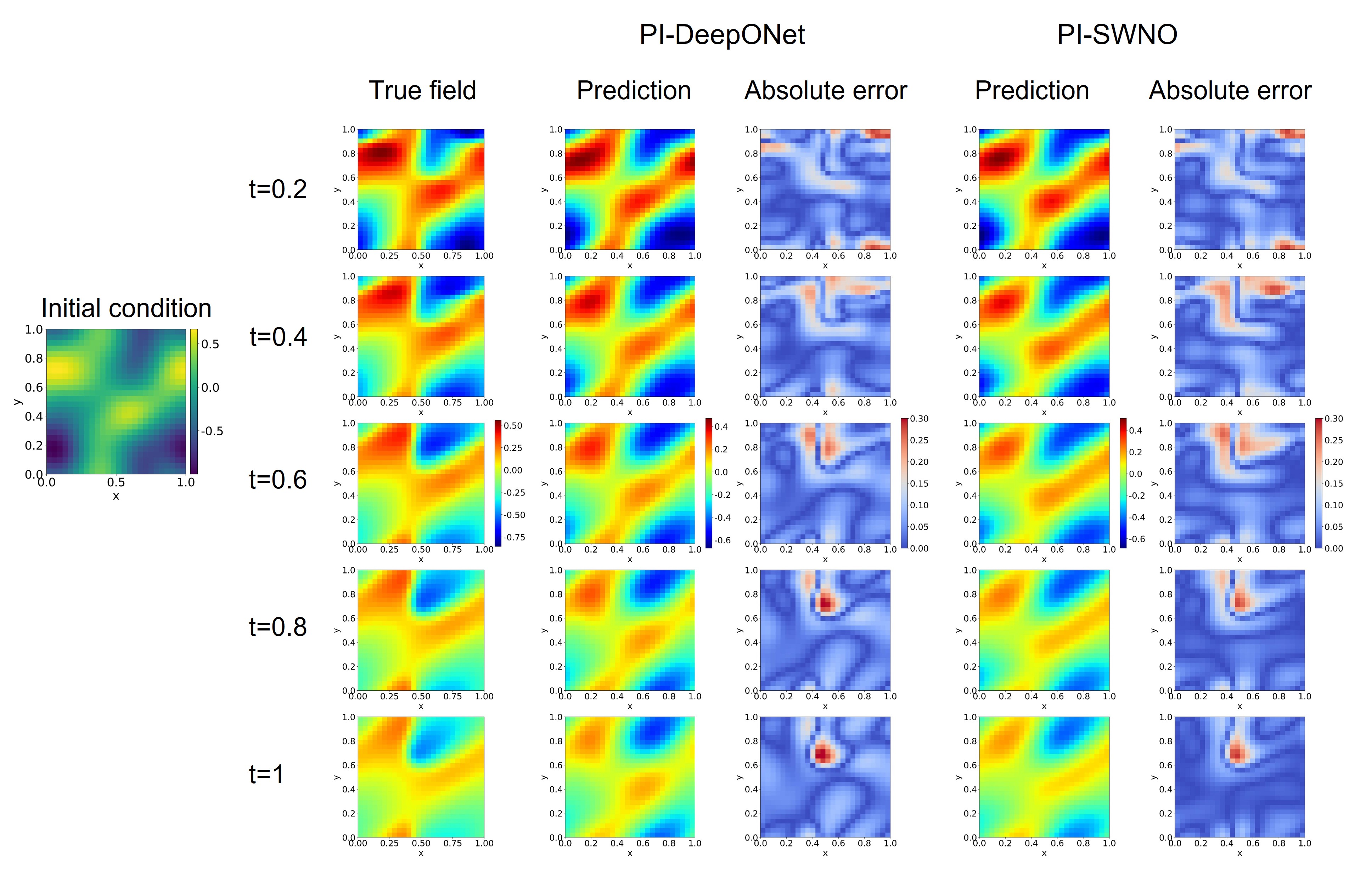}
	\caption{2D Burgers equation: The first plot shows the randomly sampled initial condition; the second column shows 2D slices of the high-fidelity numerical solution at $t = 0.2, 0.4, 0.6, 0.8, 1$; the third and fourth columns show the predicted solution and its absolute error from PI-DeepONet; the fifth and sixth columns show the predicted solution and its absolute error from PI-SWNO.}
	\label{fig:burgers2d-fig1}
\end{figure}

The 2D shock structure is more complex with higher spatiotemporal coupling strength. As shown in Fig.~\ref{fig:burgers2d-fig1}, both models can roughly capture the distribution of high-value and low-value regions and the overall evolution trend of the velocity field, but have significant defects in detailed reconstruction. Neither model can perfectly reproduce the fine wave front structure in the ground truth, especially in regions with drastic gradient changes. A large error band appears at the wave front position, indicating that the model has difficulty accurately locking the position of shocks or steep gradients when handling strong nonlinear advection terms, leading to phase deviation or amplitude attenuation. In comparison, the prediction accuracy of PI-SWNO is slightly improved, with a smaller range and lower intensity of high-error regions in the absolute error map.

\begin{figure}[htbp]
	\centering
	\includegraphics[width=0.85\linewidth]{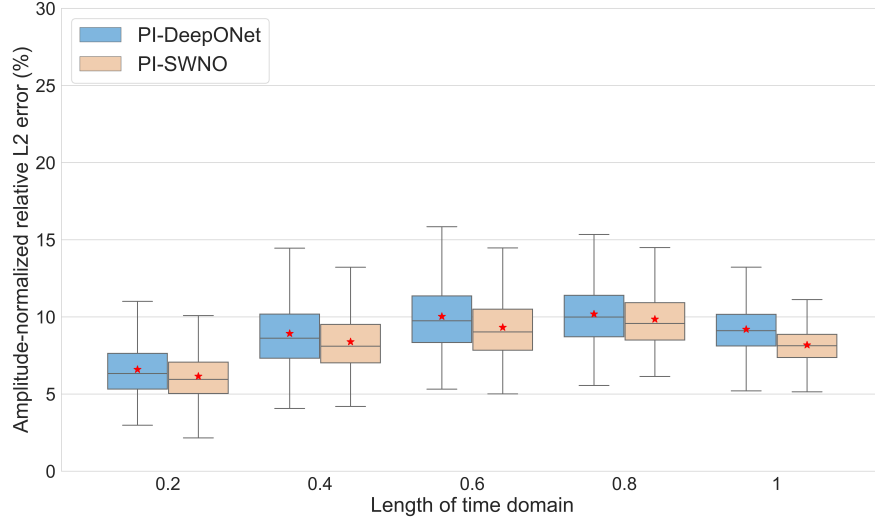}
	\caption{Comparison of long-time ANRL2E growth trends between the baseline PI-DeepONet and the proposed spatiotemporally decoupled PI-SWNO across expanding temporal domain spans on the 2D Burgers equation.}
	\label{fig:burgers2d-fig3}
\end{figure}

As shown in Fig.~\ref{fig:burgers2d-fig3}, the mean error of PI-SWNO is lower than that of PI-DeepONet across all time domain lengths, with the overall boxplot of PI-SWNO below that of PI-DeepONet and a narrower extreme value range. This indicates that PI-SWNO has smaller error dispersion, does not exhibit the theoretically expected performance degradation, and instead maintains more stable solution accuracy.

\begin{figure}[htbp]
  \centering
  \includegraphics[width=0.98\linewidth]{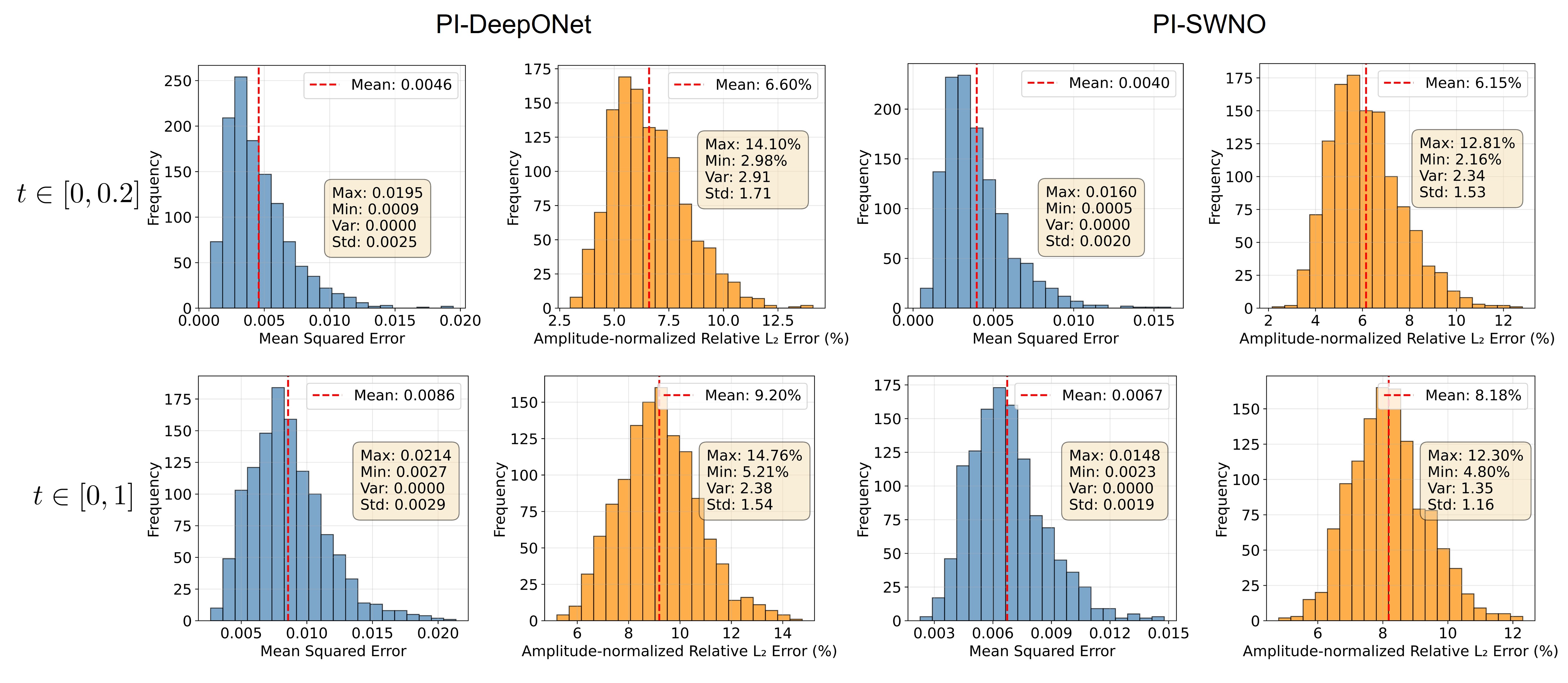}
  \caption{Statistical characteristics of MSE and ANRL2E for PI-DeepONet and PI-SWNO on the 2D Burgers equation dataset in the short time horizon $[0, 0.2]$ and long time horizon $[0, 1]$, respectively.}
  \label{fig:burgers2d-fig2}
\end{figure}

Fig.~\ref{fig:burgers2d-fig2} shows the error statistics of the two models in the short time horizon \(t \in [0,0.2]\) and long time horizon \(t \in [0,1]\), respectively. PI-SWNO slightly outperforms PI-DeepONet in all error metrics, with overall comparable accuracy between the two models.

Table~\ref{tab:Allcases} summarizes the core performance metrics of PI-DeepONet and PI-SWNO on all benchmark cases in this paper, covering four core dimensions: model scale, memory overhead, training efficiency, and prediction accuracy.

\begin{table}[htbp]
	\centering
	\small
	\setlength{\tabcolsep}{3pt}
	\caption{Comprehensive performance comparison between PI-DeepONet and PI-SWNO on all partial differential equation benchmark problems. The bold values represent better model performance.}
	\label{tab:Allcases}
	\begin{tabular}{cccccccc}
		\toprule
		\multirow{2}{*}{Case} & \multirow{2}{*}{Model} & \multirow{2}{*}{Parameters} &\multirow{2}{*}{Memory}&\multirow{2}{*}{Training time} & \multicolumn{3}{c}{ANRL2E(\%)} \\
		\cmidrule(lr){6-8}
		& & & (MB)&(s/epoch)&min&max&mean\\
		\midrule
		\multirow{2}{*}{1D HC} & PI-DeepONet &13,056 &$\bold{2,421}$ &$\bold{0.91}$ &3.47 &13.35 &7.91  \\
		& PI-SWNO &$\bold{10,530}$ &2,731 &1.52 &$\bold{3.37}$ &$\bold{5.15}$&$\bold{4.65}$   \\
		\midrule
		\multirow{2}{*}{2D HC} & PI-DeepONet &339,125 &$\bold{12,179}$ &$\bold{3.12}$ &$\bold{0.72}$ &13.49&4.99  \\
		& PI-SWNO & $\bold{258,486}$&14,505 & 3.69& 0.99 &$\bold{4.02}$&$\bold{2.49}$   \\
		\midrule
		\multirow{2}{*}{1D Wave} & PI-DeepONet & 35,840& $\bold{4,295}$&$\bold{1.11}$ &4.96 &25.90&16.66  \\
		& PI-SWNO &$\bold{33,890}$ &6,697 & 2.20&$\bold{3.74}$ &$\bold{8.46}$&$\bold{5.84}$   \\
        \midrule
		\multirow{2}{*}{2D Wave} & PI-DeepONet & 339,125& $\bold{14,595}$ &$\bold{3.69}$ & $\bold{1.45}$ &15.11 &8.20  \\
		& PI-SWNO & $\bold{258,486}$ &15,879 &4.60 &1.57 &$\bold{7.75}$ &$\bold{4.00}$   \\
		\midrule
		\multirow{2}{*}{1D KdV} & PI-DeepONet &338,176 &$\bold{11,195}$ &$\bold{3.53}$ &5.12 &11.27&7.98  \\
		& PI-SWNO &$\bold{264,578}$ &13,197 &4.13 &$\bold{2.97}$ &$\bold{10.12}$&$\bold{6.54}$   \\
		\midrule
		\multirow{2}{*}{1D Burgers} & PI-DeepONet &35,840 &$\bold{9,107}$ &$\bold{0.78}$ &7.30 &8.92&8.41  \\
		& PI-SWNO &$\bold{33,890}$ &14,977 &1.41 & $\bold{5.65}$&$\bold{7.83}$&$\bold{7.09}$  \\
		\midrule
		\multirow{2}{*}{2D Burgers} & PI-DeepONet &339,125 &$\bold{10,507}$ &$\bold{2.66}$ & 6.60 & 10.19 & 8.99  \\
		& PI-SWNO &$\bold{258,486}$ &11,369 &3.63 & $\bold{6.15}$ & $\bold{9.85}$ & $\bold{8.38}$   \\
		\bottomrule
	\end{tabular}
\end{table}

\subsection{Validation of the time-marching batch-wise sampling strategy}

To validate the effectiveness of the time-marching batch-wise sampling strategy, we perform an ablation study on this strategy using the PI-SWNO model in this section. Based on four benchmark cases---the 1D/2D heat conduction (HC) equations and 1D/2D wave equations---we compare the error convergence curves during model training with and without the proposed strategy, and systematically investigate how different numbers of intervals affect prediction accuracy, convergence rate, and memory footprint.

As shown in Fig.~\ref{fig:sample-fig1}, the proposed strategy delivers varying degrees of performance gains across all four benchmark cases. Relative to the baseline control group without the strategy, all experimental setups with interval partitioning (intervals $\geq$ 2) achieve consistent improvements in both convergence rate and final prediction accuracy. A clear positive correlation is observed between model convergence behavior and the number of time intervals. Across all four cases, the configuration with intervals=10 yields the best convergence performance: by the end of 10,000 training iterations, it achieves the largest reduction in Amplitude-Normalized Relative $L_2$ Error (ANRL2E) compared with the baseline. The relative error reductions reach approximately 49.2\% for the 1D heat conduction case, 66.3\% for the 2D heat conduction case, 51.1\% for the 1D wave equation case, and 35.2\% for the 2D wave equation case. Taking the 1D heat conduction case as an example, the model equipped with the proposed strategy and intervals=10 reduces the ANRL2E to around 8\% within only 1,000 epochs, whereas the baseline model requires 5,000 epochs to reach the same error level. For the 2D heat conduction case, the strategy accelerates convergence even more prominently: a larger number of intervals leads to a markedly faster error decay and higher final prediction accuracy. Moreover, the 2D heat conduction case exhibits stronger sensitivity to the number of intervals, such that increasing intervals brings continuous gains in both convergence speed and accuracy.

\begin{figure}[htbp]
	\centering
	\includegraphics[width=0.98\linewidth]{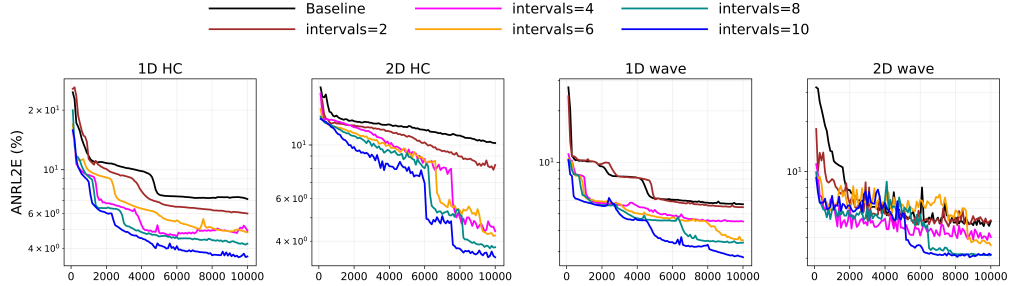}
	\caption{Ablation study of the time‑stepping batch-wise sampling strategy: We validate the convergence curves of ANRL2E with/without the proposed strategy and under different interval settings across four benchmark cases, including the 1D/2D HC equations and 1D/2D wave equations.}
	\label{fig:sample-fig1}
\end{figure}

Table~\ref{tab:interval_strategy_error} quantifies the final ANRL2E values across the four benchmark cases after 10,000 training iterations, under different numbers of time intervals. From the perspective of parameter sensitivity, model prediction accuracy generally rises with the number of intervals, accompanied by diminishing marginal gains. The most pronounced error reduction occurs when increasing intervals from 1 to 4; for the 2D heat conduction case, this stage alone delivers a relative error drop of 56.8\%, capturing the core performance improvement. Beyond 4 intervals, the accuracy gain per additional interval gradually tapers off. Furthermore, the strategy shows case-dependent effectiveness for different types of time-dependent PDEs: it delivers the strongest optimization for the 2D heat conduction case, while still providing meaningful improvements even for strongly coupled systems such as the 2D wave equation.

\begin{table}[htbp]
	\centering
	\caption{Amplitude-normalized relative $L_2$ error after 10000 training epochs of 4 numerical cases under different numbers of intervals. Intervals=1 indicates baseline, where the time-marching batch-wise sampling strategy is not applied. The bold values represent better model performance.}
	\label{tab:interval_strategy_error}
	\begin{tabular}{ccccc}
		\toprule
		\multirow{2}{*}{Intervals} & \multicolumn{4}{c}{ANRL2E(\%)} \\
		\cmidrule(lr){2-5}
		  & 1D HC & 2D HC & 1D wave & 2D wave \\
        \midrule
		1 &7.11  &   10.19 &  5.68    & 4.69  \\
		2 &6.00  &   7.90 & 5.46  &  4.66 \\
		4 &4.59  &   4.40  & 4.51 &  3.87 \\
		6 &4.82 &    4.21  & 3.50 &  3.59  \\
		8 &4.18 &    3.74  & 3.37 &   3.11 \\
		10 & $\bold{3.61}$ &  $\bold{3.43}$  & $\bold{2.78}$ &  $\bold{3.04}$    \\
		\bottomrule
	\end{tabular}
\end{table}

As presented in Table~\ref{tab:interval_strategy_memory}, training memory overhead decreases monotonically with the number of intervals. The memory saving is most pronounced when increasing intervals from 1 to 4, after which the rate of reduction slows down. At intervals=10, the memory compression ratios relative to the baseline reach 54.2\% (1D heat conduction), 54.2\% (2D heat conduction), 77.8\% (1D wave equation), and 53.3\% (2D wave equation). The 1D wave case achieves the highest memory saving, confirming the diminishing marginal returns of the proposed batch sampling strategy for video memory optimization.

\begin{table}[htbp]
	\centering
	\caption{Memory footprint of 4 numerical cases under different numbers of intervals. intervals=1 indicates baseline, where the time-marching batch-wise sampling strategy is not applied. The bold values represent better model performance.}
	\label{tab:interval_strategy_memory}
	\begin{tabular}{ccccc}
		\toprule
		\multirow{2}{*}{Intervals} & \multicolumn{4}{c}{Memory(MB)  } \\
		\cmidrule(lr){2-5}
		  & 1D HC & 2D HC & 1D wave & 2D wave \\
		\midrule
		1  & 2,722  & 21,495 & 6,688 &  15,879  \\
		2  & 1,906  &  16,233  & 2,469  &12,134  \\
		4  & 1,490 &  11,803  &2,042 &    8,914  \\
		6  & 1,360  &  10,957   &1,711&   8,252  \\
		8  & 1,290  &   9,935 &1,541&    7,512   \\
		10 & $\bold{1,248}$  &  $\bold{9,835}$   &$\bold{1,487}$&     $\bold{7,417}$   \\
		\bottomrule
	\end{tabular}
\end{table}

\section{Summary and Discussion}
\label{sec5}

This paper addresses three core challenges faced by fixed-parameter neural operators in solving long-time time-dependent PDEs: the non-decreasing growth of fitting error, error accumulation in spatiotemporally coupled architectures, and the memory bottleneck in long-time domain training. Based on the inherent physical property of unsteady systems—time-invariant spatial basis functions plus time-varying evolution coefficients—and the Stone-Weierstrass approximation theorem, we propose the spatiotemporally decoupled physics-informed Stone-Weierstrass neural operator (PI-SWNO), and design a time-marching batch-wise sampling strategy. Through systematic experiments on 7 representative time-dependent PDE benchmark cases, we validate the performance of the proposed method in terms of solution accuracy, long-time domain stability, computational efficiency, and framework universality. The core contributions and research conclusions of this paper are as follows:
In terms of theoretical foundation, we first establish the non-decreasing fitting error theorem for fixed-parameter neural operators, rigorously proving that the global minimum fitting error of a fixed-parameter neural operator is non-decreasing with the expansion of the time interval. This clarifies that the core scientific problem of long-time neural operator modeling is to control the growth rate of error with the time interval, providing core theoretical support for subsequent method design.

In terms of model architecture, we propose the spatiotemporally decoupled PI-SWNO architecture, which decomposes the coupled spatiotemporal trunk network of conventional DeepONet into three independent subnetworks: an input branch, a spatial branch, and a temporal branch. The experimental results show that, through independent encoding of spatial and temporal features, this architecture effectively avoids spatiotemporal feature aliasing, phase drift, and mode aliasing over long time horizons: in the 100 s long time horizon test of 1D/2D heat conduction equations, PI-DeepONet exhibits significant nonlinear error growth with a maximum mean relative error exceeding 13\%, while the mean error of PI-SWNO remains stable in the range of 3\% to 5\% with no error explosion phenomenon. In the long-time periodic prediction of the 1D wave equation, PI-DeepONet undergoes abrupt catastrophic divergence between $[0, 8]$ and $[0, 10]$, with the median error jumping from 5\% to over 25\%, while PI-SWNO maintains smooth error growth with no abrupt changes throughout the time horizon, demonstrating extremely strong robustness in long-time prediction. For strongly nonlinear PDEs, PI-SWNO also achieves improved accuracy compared with the conventional spatiotemporally coupled DeepONet architecture.

In terms of training strategy, to address the core conflict between memory footprint control and global solution consistency in long-time PDE training, we propose a time-marching batch-wise sampling strategy, constructing a two-stage iterative paradigm of staged local optimization plus global consistency calibration. The ablation study results show that the strategy has strong framework universality and can be seamlessly integrated into two mainstream physics-informed neural operator frameworks, PI-SWNO and PI-DeepONet. When the number of time blocks increases to 10, the training memory footprint of PI-DeepONet and PI-SWNO is reduced by 49.5\% and 54.2\%, respectively, while the model convergence speed is accelerated with no loss in final solution accuracy. This strategy provides an efficient training mode for solving time-dependent PDEs with large time spans and high spatial dimensions.

In terms of experimental validation, we systematically evaluate the performance of PI-SWNO through 7 representative time-dependent PDE benchmark cases, clarifying its advantage boundaries and applicable scenarios. For long-time prediction tasks of linear autonomous PDEs such as parabolic heat conduction equations and hyperbolic wave equations, the architecture of PI-SWNO is naturally compatible with the solution structure of PDEs, with significantly better prediction accuracy and long-time domain stability than PI-DeepONet. For the 1D KdV equation with strong nonlinearity and high-order derivatives, and the 1D/2D Burgers equations with shock structures, PI-SWNO does not exhibit the theoretically expected performance bottleneck, and its ability to capture steep gradients and nonlinear structures is still superior to PI-DeepONet. However, due to the more complex solution structure and higher spatiotemporal coupling degree of nonlinear PDEs, PI-SWNO does not show an advantage in the long-time error growth rate.

Our work addresses a critical gap in the field of neural operators for long-time prediction of time-dependent PDEs, and offers a powerful tool for studying complex dynamical systems with long-range spatiotemporal dependencies. The PI-SWNO framework offers the potential to develop physics-informed models for both complex natural phenomena and digital twins of engineering systems. For instance, by integrating a priori physical knowledge, the proposed model can effectively describe and forecast extreme environmental conditions across a given region, thereby facilitating preventive warnings for practical engineering scenarios.

\appendix

\section{A proof of the non-decreasing fitting error theorem for fixed-parameter neural operators} \label{wucha_proof}

This appendix proves the non-decreasing fitting error theorem for fixed-parameter neural operators presented in Section~\ref{wucha}. Notation and definitions follow those in Section~\ref{wucha}.
\begin{theorem}[Non-Decreasing Fitting Error for Fixed-Parameter Neural Operators]
Suppose the following two assumptions hold:
\begin{enumerate}
    \item The true solution of the spatiotemporal system satisfies \(u \in L^2_{\text{loc}}(\Omega \times  \mathcal{T} )\);
    \item For any fixed \(\theta \in \Theta\), the mapping \(\tau \mapsto \left\| \mathcal{G}_\theta(\cdot, \tau) - u(\cdot, \tau) \right\|_{L^2(\Omega)}^2\) is Lebesgue measurable, and the integral \(J_t(\theta)\) is finite for any finite \(t>0\).
\end{enumerate}
Then the global minimum fitting error \(L_t\) is a non-decreasing function of time \(t\). That is, for any \(0 \leq t_1 < t_2 < +\infty\),
\[
    L_{t_1} \leq L_{t_2}
\]
\end{theorem}

\begin{proof}

For any \(0 \leq t_1 < t_2 < +\infty\) and any fixed \(\theta \in \Theta\), the error functional can be decomposed by the interval additivity of the Lebesgue integral: 
\[
   J_{t_2}(\theta) = \int_{0}^{t_2} \left\| \mathcal{G}_\theta - u \right\|_{L^2(\Omega)}^2 \mathrm{d}\tau = \int_{0}^{t_1} \left\| \mathcal{G}_\theta - u \right\|_{L^2(\Omega)}^2 \mathrm{d}\tau + \int_{t_1}^{t_2} \left\| \mathcal{G}_\theta - u \right\|_{L^2(\Omega)}^2 \mathrm{d}\tau
\]
By non-negativity of the squared-integrand Lebesgue integral, the second term satisfies
\[
    \int_{t_1}^{t_2} \left\| \mathcal{G}_\theta - u \right\|_{L^2(\Omega)}^2 \mathrm{d}\tau \geq 0
\]
Thus for all \(\theta \in \Theta\), 
\[
    J_{t_2}(\theta) \geq J_{t_1}(\theta)
\]

By definition of the minimum fitting error, \(L_{t_1} = \inf_{\theta \in \Theta} J_{t_1}(\theta)\) is the greatest lower bound of \(J_{t_1}(\theta)\) over \(\Theta\). Therefore, for all \(\theta \in \Theta\),
\[
J_{t_1}(\theta) \geq L_{t_1}    
\]
Therefore, for all \(\theta \in \Theta\),
\[
J_{t_2}(\theta) \geq L_{t_1}    
\]
meaning \(L_{t_1}\) is a lower bound of the set \(\{ J_{t_2}(\theta) \mid \theta \in \Theta \}\). By definition of the infimum,
\[
L_{t_2} = \inf_{\theta \in \Theta} J_{t_2}(\theta) \geq L_{t_1}    
\]

In summary, \(L_{t_1} \leq L_{t_2}\) holds for all \(0 \leq t_1 < t_2 < +\infty\), so \(L_t\) is non-decreasing in \(t\).
\end{proof}

\section{Theoretical proof of the universal approximation property of the PI-SWNO architecture}
\label{appendix:piswno_approximation}

This appendix establishes the universal approximation capability of the proposed PI-SWNO.

\subsection{Preliminaries and basic definitions}
\begin{definition}[Input and output function spaces]
    Let the domain of input functions for unsteady PDEs be a compact set \(K \subset \mathbb{R}^d\), corresponding to initial conditions, boundary conditions, source terms, and other inputs. The input function space is defined as the space of continuous functions on the compact set:
    \[
    \mathcal{X} = C(K, \mathbb{R})
    \]
    Let the spatial computational domain of the PDE be a non-empty bounded closed set \(\Omega \subset \mathbb{R}^n\), and the total time interval be a non-empty bounded closed interval \([0,T] \subset \mathbb{R}\). The output solution function space is defined as the space of continuous functions on the spatiotemporal product compact set:
    \[
    \mathcal{Y} = C(\Omega \times [0,T], \mathbb{R})
    \]
    By the Heine-Borel theorem in Euclidean spaces, \(\Omega\), \([0,T]\), and their product \(\Omega \times [0,T]\) are all compact Hausdorff spaces. 
\end{definition}

\begin{definition}[PDE solution operator]
    For a well-posed unsteady PDE, the mapping from input functions to spatiotemporal solutions defines the solution operator:
    \[
    \mathcal{G}: \mathcal{X} \to \mathcal{Y}
    \]
    We assume \(\mathcal{G}\) is continuous and bounded, which holds for nearly all well-posed unsteady PDEs in engineering applications.
\end{definition}

\begin{definition}[Operator form of PI-SWNO]
    The proposed spatiotemporally decoupled PI-SWNO has the unified mathematical form:
    \begin{equation}
    \mathcal{G}_\theta(f)(x, t) = \sum_{i=1}^n b_i(f;\theta_b) \cdot \phi_i(x;\theta_s) \cdot \varphi_i(t;\theta_t),
    \end{equation}
    where
    \begin{itemize}
        \item \(\theta = \{\theta_b, \theta_s, \theta_t\}\) denotes the learnable parameters of the input, spatial, and temporal branches;
        \item \(n\) is a positive integer representing the order of modal expansion;
        \item \(b_i: \mathcal{X} \to \mathbb{R}\) is the modal coefficient generated by the input branch, a continuous function of the input \(f\);
        \item \(\phi_i: \Omega \to \mathbb{R}\) is the time-invariant spatial basis function from the spatial branch, a continuous function of the spatial coordinate \(x\);
        \item \(\varphi_i: [0,T] \to \mathbb{R}\) is the temporal basis function from the temporal branch, a continuous function of the time coordinate \(t\).
    \end{itemize}
\end{definition}

\subsection{Three fundamental theorems}
We first state three core theorems used in the proof: the universal approximation theorem for operators, the Stone-Weierstrass theorem, and the universal approximation theorem for neural networks.

The universal approximation theorem for operators (UATO)~\cite{23Lu2021} underpins the DeepONet architecture by rigorously proving that branch-trunk neural operators possess universal approximation capability for continuous bounded operators, which is formally stated as follows. 

\begin{theorem}[Universal Approximation Theorem for Operators]
    Let \(\mathcal{X} = C(K, \mathbb{R})\) and \(\mathcal{Y} = C(\Omega \times [0,T], \mathbb{R})\) be real-valued continuous function spaces over compact sets. Let \(\mathcal{G}: \mathcal{X} \to \mathcal{Y}\) be a continuous bounded operator. Then for any compact subset \(\mathcal{K} \subset \mathcal{X}\) and any \(\varepsilon > 0\), there exist a positive integer \(n\), continuous functions \(b_i: \mathcal{X} \to \mathbb{R}\), and continuous functions \(\psi_i \in C(\Omega \times [0,T], \mathbb{R})(i=1,2,\dots,n)\) such that the operator
    \[
    \tilde{\mathcal{G}}(f)(\cdot) = \sum_{i=1}^n b_i(f) \cdot \psi_i(\cdot)
    \]
    satisfies the uniform approximation condition:
    \begin{equation}
    \sup_{f \in \mathcal{K}} \left\| \mathcal{G}(f) - \tilde{\mathcal{G}}(f) \right\|_\infty < \varepsilon,
    \end{equation}
    where \(\|\cdot\|_\infty\) is the uniform norm defined by \(\|u\|_\infty = \sup_{z \in X} |u(z)|\).
    \label{theorem:uato}
\end{theorem}

The Stone-Weierstrass theorem~\cite{Rudin1987Real} serves as a topological generalization of the classical Weierstrass approximation theorem, and constitutes a cornerstone of function approximation theory. Its formal statement is as follows. 
\begin{theorem}[Stone-Weierstrass theorem]
    Let \(X\) be a non-empty compact Hausdorff space, and let \(\mathcal{A}\) be a subalgebra of \(C(\Omega \times [0,T], \mathbb{R})\) (i.e., a subset closed under linear combinations and pointwise products). If \(\mathcal{A}\) satisfies:
    \begin{enumerate}
        \item Point-separating property: for any distinct \(x \neq y \in X\), there exists \(f \in \mathcal{A}\) such that \(f(x) \neq f(y)\);
        \item Constant-function property: the constant function \(1(x) \equiv 1\) belongs to \(\mathcal{A}\).
    \end{enumerate}
    then \(\mathcal{A}\) is uniformly dense in \(C(X, \mathbb{R})\). That is, for any continuous function \(u \in C(X, \mathbb{R})\) and any \(\varepsilon > 0\), there exists \(g \in \mathcal{A}\) such that
    \begin{equation}
    \|u - g\|_\infty = \sup_{x \in X} |u(x) - g(x)| < \varepsilon
    \end{equation}
    \label{theorem:stone_weierstrass}
\end{theorem}

For the spatiotemporal scenario considered in this paper, the following key corollary follows directly from Theorem~\ref{theorem:stone_weierstrass}:

\begin{corollary}[Stone-Weierstrass approximation for spatiotemporal product spaces]
     Let the spatiotemporal product space be \(X = \Omega \times [0,T]\). Define the subalgebra
    \[
    \mathcal{A} = \operatorname{span}\left\{ \phi(x) \cdot \varphi(t) \mid \phi \in C(\Omega, \mathbb{R}),\ \varphi \in C([0,T], \mathbb{R}) \right\}
    \]
    Then \(\mathcal{A}\) is uniformly dense in \(C(\Omega \times [0,T], \mathbb{R})\). For any \(\psi(x,t) \in C(\Omega \times [0,T], \mathbb{R})\) and any \(\delta > 0\), there exist a positive integer \(m\), spatial continuous functions \(\phi_j \in C(\Omega, \mathbb{R})\), and temporal continuous functions \(\varphi_j \in C([0,T], \mathbb{R})(j=1,2,\dots,m)\) such that
    \begin{equation}
    \sup_{(x,t) \in \Omega \times [0,T]} \left| \psi(x,t) - \sum_{j=1}^m \phi_j(x) \varphi_j(t) \right| < \delta
    \end{equation}
    \label{corollary:sw_spacetime}
\end{corollary}

\begin{proof}
    By the Heine-Borel theorem, \(\Omega\) and \([0,T]\) are compact Hausdorff, so their product \(X = \Omega \times [0,T]\) is a non-empty compact Hausdorff space. It is straightforward to verify that \(\mathcal{A}\) is a subalgebra of \(C(X, \mathbb{R})\) and satisfies the two key conditions of Theorem~\ref{theorem:stone_weierstrass}:
    \begin{enumerate}
        \item Constant-function property: set \(\phi(x) \equiv 1\) and \(\varphi(t) \equiv 1\); then the constant function \(1(x,t) \equiv 1 = \phi(x)\varphi(t) \in \mathcal{A}\);
        \item Point-separating property: for any distinct \((x_1,t_1) \neq (x_2,t_2) \in X\), if \(x_1 \neq x_2\), take \(\phi(x) = \|x-x_1\|\) and \(\varphi(t) \equiv 1\); if \(t_1 \neq t_2\), take \(\phi(x) \equiv 1\) and \(\varphi(t) = |t-t_1|\). In both cases, the two points are separated.
    \end{enumerate}
    By Theorem~\ref{theorem:stone_weierstrass}, \(\mathcal{A}\) is uniformly dense in \(C(X, \mathbb{R})\). This completes the proof. 
\end{proof}

The universal approximation theorem (UAT)~\cite{Cybenko1989Approximation} for neural networks serves as the fundamental theoretical cornerstone of deep learning. It rigorously proves the universal approximation capability of feedforward neural networks for continuous functions defined on compact sets. Its standard formulation, specifically for feedforward networks with a single hidden layer, is stated as follows.
\begin{theorem}[Universal approximation theorem for neural networks]
    Let \(\sigma: \mathbb{R} \to \mathbb{R}\) be a continuous, non-constant sigmoidal activation (e.g., Sigmoid, Tanh) or a non-polynomial continuous activation satisfying universal approximation conditions (e.g., ReLU, GELU). For any positive integer \(d\), any compact set \(K \subset \mathbb{R}^d\), any continuous function \(f: K \to \mathbb{R}\), and any \(\varepsilon > 0\), there exist a positive integer \(m\), weights \(w_j \in \mathbb{R}^d, b_j \in \mathbb{R}\), and coefficients \(\beta_j \in \mathbb{R}(j=1,2,\dots,m)\) such that the single-hidden-layer feedforward neural network:
    \[
    \hat{f}(\boldsymbol{z}) = \sum_{j=1}^m \beta_j \cdot \sigma\left( w_j^\mathrm{T} \boldsymbol{z} + b_j \right)
    \]
    satisfies:
    \[
    \sup_{\boldsymbol{z} \in K} \left| f(\boldsymbol{z}) - \hat{f}(\boldsymbol{z}) \right| < \varepsilon
    \]
    That is: any continuous function defined on a compact set can be uniformly approximated to arbitrary precision by a single-hidden-layer feedforward neural network equipped with a suitable activation function.
    \label{theorem:nn_uat}
\end{theorem}

This theorem extends directly to multi-output feedforward neural networks: for any continuous vector-valued function \(\boldsymbol{f}: K \to \mathbb{R}^n\) defined on the compact set \(K \subset \mathbb{R}^d\), each of its components is continuous. Consequently, \(\boldsymbol{f}\) can be uniformly approximated to arbitrary precision by multi-output feedforward neural networks.

\subsection{universal approximation theorem for PI-SWNO and proof}
Based on the three foundational theorems above, this paper proposes the following universal approximation theorem for the PI-SWNO framework, along with its rigorous mathematical proof.

\begin{theorem}[Universal approximation theorem for PI-SWNO]
     Let the solution operator \(\mathcal{G}: \mathcal{X} \to \mathcal{Y}\) of the unsteady PDE be continuous and bounded, where \(\mathcal{X} = C(K, \mathbb{R})\) and \(\mathcal{Y} = C(\Omega \times [0,T], \mathbb{R})\) are continuous function spaces over compact sets. Then for any compact subset \(\mathcal{X} \subset \mathcal{K}\) and any \(\varepsilon > 0\), there exist a positive integer \(N\) and learnable parameters \(\theta\) such that the PI-SWNO operator \(\mathcal{G}_\theta\) uniformly approximates the target solution operator \(\mathcal{G}\) on \(\mathcal{K}\):
    \begin{equation}
    \sup_{f \in \mathcal{K}} \sup_{(x,t) \in \Omega \times [0,T]} \left| \mathcal{G}(f)(x,t) - \mathcal{G}_\theta(f)(x,t) \right| < \varepsilon
    \end{equation}
    \label{theorem:piswno_approximation}
\end{theorem}

\begin{proof}

    (1)Decompose the target solution operator into the standard branch-trunk form of DeepONet via the Universal Approximation Theorem for Operators.
    
    For the target solution operator \(\mathcal{G}\), any compact \(\mathcal{K} \subset \mathcal{X}\), and any \(\varepsilon > 0\), by Theorem~\ref{theorem:uato}, there exist a positive integer \(n\), continuous coefficient functions \(b_i: \mathcal{X} \to \mathbb{R}\), and spatiotemporally continuous basis functions \(\psi_i(x,t) \in C(\Omega \times [0,T], \mathbb{R})(i=1,2,\dots,n)\) such that
    \[
    \tilde{\mathcal{G}}(f)(x,t) = \sum_{i=1}^n b_i(f) \cdot \psi_i(x,t)
    \]
    satisfies:
    \begin{equation}
    \sup_{f \in \mathcal{K}} \sup_{(x,t) \in \Omega \times [0,T]} \left| \mathcal{G}(f)(x,t) - \tilde{\mathcal{G}}(f)(x,t) \right| < \frac{\varepsilon}{3}
    \end{equation}

    Since \(\mathcal{K}\) is compact on \(\mathcal{X}\), each continuous function \(b_i(f)\) is bounded on \(\mathcal{K}\). That is, there exists a constant \(M > 0\) such that for all \(f \in \mathcal{K}\) and all \(i=1,2,\dots,n\), have
    \begin{equation}
    |b_i(f)| \leq M
    \end{equation}

    (2)Decompose each spatiotemporally coupled basis function \(\psi_i(x,t)\) into a linear combination of decoupled spatial and temporal basis functions via the Stone-Weierstrass Theorem.
    
    For each spatiotemporally coupled basis function \(\psi_i(x,t)\), let the approximation accuracy be \(\delta = \frac{\varepsilon}{3nM}\). By Corollary~\ref{corollary:sw_spacetime}, there exist a positive integer \(m\), spatial continuous functions \(\phi_{i,j} \in C(\Omega, \mathbb{R})\), and temporal continuous functions \(\varphi_{i,j} \in C([0,T], \mathbb{R})(j=1,2,\dots,m)\) such that
    \begin{equation}
    \sup_{(x,t) \in \Omega \times [0,T]} \left| \psi_i(x,t) - \sum_{j=1}^m \phi_{i,j}(x) \varphi_{i,j}(t) \right| < \delta = \frac{\varepsilon}{3nM}
    \end{equation}

    (3)Show that the three decoupled branches of PI-SWNO can approximate the required continuous coefficient functions, spatial basis functions, and temporal basis functions to arbitrary accuracy via the Universal Approximation Theorem for Neural Networks.
    
    \begin{itemize}
        \item Approximation of the input branch: The input branch takes discrete sampled values of \(f\) (a finite-dimensional vector \(\boldsymbol{z}_f \in \mathbb{R}^{M \times N}\)), and outputs an n-dimensional modal coefficient vector \([b_1(f), \dots, b_n(f)]^\mathrm{T}\). Since \(b_i(f)\) is continuous, its discrete counterpart is a continuous vector-valued function on a compact set. By Theorem~\ref{theorem:nn_uat}, there exist parameters \(\theta_b\) such that the network outputs \(\hat{b}_i(f)\) satisfy:
        \[
        \sup_{f \in \mathcal{K}} |b_i(f) - \hat{b}_i(f)| < \frac{\varepsilon}{3n \cdot \|\psi_i\|_\infty}, \quad \forall i=1,2,\dots,n \quad ,
        \]
        where \(\|\psi_i\|_\infty = \sup_{(x,t)} |\psi_i(x,t)|\) denotes the uniform norm of the basis function.

        \item Approximation of the spatial branch: The spatial branch takes \(x \in \Omega\) (finite-dimensional vector on a compact set) and outputs an \(n \times m\) dimensional vector of spatial basis functions \([\phi_{1,1}(x), \dots, \phi_{n,m}(x)]^\mathrm{T}\). By Theorem~\ref{theorem:nn_uat}, there exist parameters \(\theta_s\) such that the network outputs \(\hat{\phi}_{i,j}(x)\) satisfy:
        \[
        \sup_{x \in \Omega} |\phi_{i,j}(x) - \hat{\phi}_{i,j}(x)| < \frac{\varepsilon}{3n m M \cdot \|\varphi_{i,j}\|_\infty}, \quad \forall i,j
        \]

        \item Approximation of the temporal branch: The temporal branch takes \(t \in [0,T]\) (1-dimensional vector on a compact set) and outputs an \(n \times m\)dimensional vector of temporal basis functions \([\varphi_{1,1}(t), \dots, \varphi_{n,m}(t)]^\mathrm{T}\). By Theorem~\ref{theorem:nn_uat}, there exist parameters \(\theta_t\) such that the network outputs \(\hat{\varphi}_{i,j}(t)\) satisfy:
        \[
        \sup_{t \in [0,T]} |\varphi_{i,j}(t) - \hat{\varphi}_{i,j}(t)| < \frac{\varepsilon}{3n m M \cdot \|\phi_{i,j}\|_\infty}, \quad \forall i,j
        \]
    \end{itemize}

    In summary, the three branch networks of PI-SWNO can jointly construct an operator that fully conforms to the architectural form:
    \begin{equation}
    \mathcal{G}_\theta(f)(x,t) = \sum_{i=1}^n \sum_{j=1}^m \hat{b}_i(f) \cdot \hat{\phi}_{i,j}(x) \cdot \hat{\varphi}_{i,j}(t),
    \end{equation}
    where $\theta = \{\theta_b, \theta_s, \theta_t\}$ denotes the learnable parameters of the network, which is fully consistent with the spatiotemporally decoupled architecture of the proposed PI-SWNO.

    (4)Combine the approximation errors via the triangle inequality.
     
    For any \(f \in \mathcal{K}\) and any \((x,t) \in \Omega \times [0,T]\), combine the approximation errors of the three stages by the triangle inequality:

    \[
    \begin{aligned}
    &\left| \mathcal{G}(f)(x,t) - \mathcal{G}_\theta(f)(x,t) \right| \\
    \leq& \left| \mathcal{G}(f)(x,t) - \tilde{\mathcal{G}}(f)(x,t) \right| 
    + \left| \tilde{\mathcal{G}}(f)(x,t) - \sum_{i,j} b_i(f) \phi_{i,j}(x) \varphi_{i,j}(t) \right| \\
    &+ \left| \sum_{i,j} b_i(f) \phi_{i,j}(x) \varphi_{i,j}(t) - \sum_{i,j} \hat{b}_i(f) \hat{\phi}_{i,j}(x) \hat{\varphi}_{i,j}(t) \right| \\
    \leq& \left| \mathcal{G}(f)(x,t) - \tilde{\mathcal{G}}(f)(x,t) \right| 
    + \left| \tilde{\mathcal{G}}(f)(x,t) - \sum_{i,j} b_i(f) \phi_{i,j}(x) \varphi_{i,j}(t) \right| \\
    &+ \left| \sum_{i,j} (b_i(f) - \hat{b}_i(f)) \phi_{i,j}(x) \varphi_{i,j}(t) \right| + \left| \sum_{i,j} \hat{b}_i(f) (\phi_{i,j}(x) - \hat{\phi}_{i,j}(x)) \varphi_{i,j}(t) \right|  \\
    &+ \left| \sum_{i,j} \hat{b}_i(f) \hat{\phi}_{i,j}(x) (\varphi_{i,j}(t) - \hat{\varphi}_{i,j}(t)) \right|\\
    <& \frac{\varepsilon}{3} + \frac{\varepsilon}{3} + \frac{\varepsilon}{3} \\
    =& \varepsilon
    \end{aligned}
    \]

    Taking the supremum on both sides of the inequality:
    \[
    \sup_{f \in \mathcal{K}} \sup_{(x,t) \in \Omega \times [0,T]} \left| \mathcal{G}(f)(x,t) - \mathcal{G}_\theta(f)(x,t) \right| < \varepsilon
    \]
    This completes the proof of Theorem~\ref{theorem:piswno_approximation}.
\end{proof}

This theorem shows that the spatiotemporally decoupled PI-SWNO can uniformly approximate the continuous bounded solution operator of unsteady PDEs to arbitrary accuracy on compact subsets of function spaces, possessing the same universal approximation capability as the classical DeepONet. Meanwhile, the architecture matches the inherent physical property of unsteady PDEs --- time-invariant spatial bases plus time-varying evolution coefficients --- and offers clearer physical interpretability and superior long-time error control compared with spatiotemporally coupled DeepONet.

A supplementary justification for the discrete sampling in the input branch is provided below based on the Arzelà-Ascoli theorem.

\begin{theorem}[Uniform Characterization of Compact Function Families by Finite Samples]
    Let the input function domain \(K \subset \mathbb{R}^d\) be a non-empty compact set, and let the input function space \(\mathcal{X}=C(K, \mathbb{R})\) be equipped with the uniform norm \(\|\cdot\|_{C(K)}\). Let \(\mathcal{K} \subset \mathcal{X}\) be any compact subset. Then for any \(\varepsilon>0\), there exist finitely many sampling points \(\{x_i\}_{i=1}^m \subset K\) such that for any input function \(f,g \in \mathcal{K}\), if \(f(x_i)=g(x_i), \ \forall i=1,2,\dots,m\), then
    \[
    \|f - g\|_{C(K)} = \sup_{x \in K} |f(x) - g(x)| < \varepsilon
    \]
\end{theorem}

This theorem guarantees that all observable input cases in engineering lie within compact subsets of the input function space, and finite discrete samples suffice to characterize the input function to arbitrary accuracy. Thus, using discrete sampled values as input to the PI-SWNO input branch does not compromise its theoretical universal approximation capability.






\bibliographystyle{elsarticle-num}
\bibliography{refs} 

@article{RAISSI2019,
title = {Physics-informed neural networks: A deep learning framework for solving forward and inverse problems involving nonlinear partial differential equations},
journal = {Journal of Computational Physics},
volume = {378},
pages = {686-707},
year = {2019},
issn = {0021-9991},
doi = {https://doi.org/10.1016/j.jcp.2018.10.045},
url = {https://www.sciencedirect.com/science/article/pii/S0021999118307125},
author = {M. Raissi and P. Perdikaris and G.E. Karniadakis}
}

@article{li2023transformer,
title={Transformer for Partial Differential Equations{\textquoteright} Operator Learning},
author={Zijie Li and Kazem Meidani and Amir Barati Farimani},
journal={Transactions on Machine Learning Research},
issn={2835-8856},
year={2023},
url={https://openreview.net/forum?id=EPPqt3uERT}
}

@article{MUCKE2021,
title = {Reduced order modeling for parameterized time-dependent PDEs using spatially and memory aware deep learning},
journal = {Journal of Computational Science},
volume = {53},
pages = {101408},
year = {2021},
issn = {1877-7503},
doi = {https://doi.org/10.1016/j.jocs.2021.101408},
url = {https://www.sciencedirect.com/science/article/pii/S1877750321000934},
author = {Nikolaj T. Mücke and Sander M. Bohté and Cornelis W. Oosterlee}
}

@article{VARGAS2022,
title = {Finite difference method for solving fractional differential equations at irregular meshes},
journal = {Mathematics and Computers in Simulation},
volume = {193},
pages = {204-216},
year = {2022},
issn = {0378-4754},
doi = {https://doi.org/10.1016/j.matcom.2021.10.010},
url = {https://www.sciencedirect.com/science/article/pii/S037847542100361X},
author = {Antonio M. Vargas}
}

@article{KERGRENE2016,
title = {Stable Generalized Finite Element Method and associated iterative schemes; application to interface problems},
journal = {Computer Methods in Applied Mechanics and Engineering},
volume = {305},
pages = {1-36},
year = {2016},
issn = {0045-7825},
doi = {https://doi.org/10.1016/j.cma.2016.02.030},
url = {https://www.sciencedirect.com/science/article/pii/S0045782516300603},
author = {Kenan Kergrene and Ivo Babuška and Uday Banerjee}
}

@article{BUCHMULLER2016,
title = {Finite volume WENO methods for hyperbolic conservation laws on Cartesian grids with adaptive mesh refinement},
journal = {Applied Mathematics and Computation},
volume = {272},
pages = {460-478},
year = {2016},
issn = {0096-3003},
doi = {https://doi.org/10.1016/j.amc.2015.03.078},
url = {https://www.sciencedirect.com/science/article/pii/S0096300315003926},
author = {Pawel Buchmüller and Jürgen Dreher and Christiane Helzel}
}

@article{Arnold2001,
author = {Arnold, Douglas N. and Brezzi, Franco and Cockburn, Bernardo and Marini, L. Donatella},
title = {Unified Analysis of Discontinuous Galerkin Methods for Elliptic Problems},
year = {2001},
issue_date = {2001},
publisher = {Society for Industrial and Applied Mathematics},
address = {USA},
volume = {39},
number = {5},
issn = {0036-1429},
url = {https://doi.org/10.1137/S0036142901384162},
doi = {10.1137/S0036142901384162},
journal = {SIAM J. Numer. Anal.},
month = may,
pages = {1749–1779},
numpages = {31}
}

@article{VILLADSEN1995,
title = {Solution of boundary-value problems by orthogonal collocation},
journal = {Chemical Engineering Science},
volume = {50},
number = {24},
pages = {3981-3996},
year = {1995},
issn = {0009-2509},
doi = {https://doi.org/10.1016/0009-2509(96)81831-8},
url = {https://www.sciencedirect.com/science/article/pii/0009250996818318},
author = {J.V. Villadsen and W.E. Stewart}
}

@article{4Diab2025,
   author = {Diab, Waleed and Al Kobaisi, Mohammed},
   title = {Temporal neural operator for modeling time-dependent physical phenomena},
   journal = {Scientific Reports},
   volume = {15},
   number = {1},
   pages = {32791},
   ISSN = {2045-2322},
   DOI = {10.1038/s41598-025-16922-5},
   url = {https://doi.org/10.1038/s41598-025-16922-5},
   year = {2025},
   type = {Journal Article}
}

@article{5Hu2025,
   author = {Hu, Zheyuan and Daryakenari, Nazanin Ahmadi and Shen, Qianli and Kawaguchi, Kenji and Karniadakis, George Em},
   title = {State-space models are accurate and efficient neural operators for dynamical systems},
   journal = {Neural Networks},
   pages = {108496},
   ISSN = {0893-6080},
   DOI = {https://doi.org/10.1016/j.neunet.2025.108496},
   url = {https://www.sciencedirect.com/science/article/pii/S0893608025013772},
   year = {2025},
   type = {Journal Article}
}

@misc{34cho2025,
	title={Physics-Informed Deep Inverse Operator Networks for Solving PDE Inverse Problems}, 
	author={Sung Woong Cho and Hwijae Son},
	year={2025},
	eprint={2412.03161},
	archivePrefix={arXiv},
	primaryClass={math.NA},
	url={https://arxiv.org/abs/2412.03161}, 
}

@article{41KORIC2023,
	title = {Data-driven and physics-informed deep learning operators for solution of heat conduction equation with parametric heat source},
	journal = {International Journal of Heat and Mass Transfer},
	volume = {203},
	pages = {123809},
	year = {2023},
	issn = {0017-9310},
	doi = {https://doi.org/10.1016/j.ijheatmasstransfer.2022.123809},
	url = {https://www.sciencedirect.com/science/article/pii/S0017931022012777},
	author = {Seid Koric and Diab W. Abueidda}
}

@misc{6Lei2025,
      title={Long-time Integration of Nonlinear Wave Equations with Neural Operators}, 
      author={Guanhang Lei and Zhen Lei and Lei Shi},
      year={2025},
      eprint={2410.15617},
      archivePrefix={arXiv},
      primaryClass={math.NA},
      url={https://arxiv.org/abs/2410.15617}, 
}

@misc{7Mandl2025,
   author = {Mandl, Luis and Nayak, Dibyajyoti and Ricken, Tim and Goswami, Somdatta},
   title = {Physics-Informed Time-Integrated DeepONet: Temporal Tangent Space Operator Learning for High-Accuracy Inference},
   pages = {arXiv:2508.05190},
   month = {August 01, 2025},
   DOI = {10.48550/arXiv.2508.05190},
   url = {https://ui.adsabs.harvard.edu/abs/2025arXiv250805190M},
   year = {2025},
   type = {Electronic Article}
}

@misc{8Michałowska2024,
      title={Neural Operator Learning for Long-Time Integration in Dynamical Systems with Recurrent Neural Networks}, 
      author={Katarzyna Michałowska and Somdatta Goswami and George Em Karniadakis and Signe Riemer-Sørensen},
      year={2024},
      eprint={2303.02243},
      archivePrefix={arXiv},
      primaryClass={cs.LG},
      url={https://arxiv.org/abs/2303.02243}, 
}

@misc{9Nayak2025,
      title={TI-DeepONet: Learnable Time Integration for Stable Long-Term Extrapolation}, 
      author={Dibyajyoti Nayak and Somdatta Goswami},
      year={2025},
      eprint={2505.17341},
      archivePrefix={arXiv},
      primaryClass={cs.LG},
      url={https://arxiv.org/abs/2505.17341}, 
}

@inproceedings{10Buitrago2025,
title={On the Benefits of Memory for Modeling Time-Dependent {PDE}s},
author={Ricardo Buitrago and Tanya Marwah and Albert Gu and Andrej Risteski},
booktitle={The Thirteenth International Conference on Learning Representations},
year={2025},
url={https://openreview.net/forum?id=o9kqa5K3tB}
}

@article{11He2024,
  author = {He, Junyan and Kushwaha, Shashank and Park, Jaewan and Koric, Seid and Abueidda, Diab and Jasiuk, Iwona},
  title = {Sequential Deep Operator Networks (S-DeepONet) for predicting full-field solutions under time-dependent loads},
  journal = {Engineering Applications of Artificial Intelligence},
  volume = {127},
  pages = {107258},
  ISSN = {0952-1976},
  DOI = {https://doi.org/10.1016/j.engappai.2023.107258},
  year = {2024},
  type = {Journal Article}
}

@article{12Hu2025,
  author = {Hu, Zheyuan and Cao, Qianying and Kawaguchi, Kenji and Karniadakis, George Em},
  title = {DeepOMamba: State-space model for spatio-temporal PDE neural operator learning},
  journal = {Journal of Computational Physics},
  volume = {540},
  pages = {114272},
  ISSN = {0021-9991},
  DOI = {https://doi.org/10.1016/j.jcp.2025.114272},
  year = {2025},
  type = {Journal Article}
}

@article{13Jin2022,
author = {Jin, Pengzhan and Meng, Shuai and Lu, Lu},
year = {2022},
month = {11},
pages = {A3490-A3514},
title = {MIONet: Learning Multiple-Input Operators via Tensor Product},
volume = {44},
journal = {SIAM Journal on Scientific Computing},
doi = {10.1137/22M1477751}
}

@misc{14Gu2024,
title={Mamba: Linear-Time Sequence Modeling with Selective State Spaces},
author={Albert Gu and Tri Dao},
year={2024},
url={https://openreview.net/forum?id=AL1fq05o7H}
}

@misc{15Dao2024,
     title={Transformers are SSMs: Generalized Models and Efficient Algorithms Through Structured State Space Duality},
     author={Tri Dao and Albert Gu},
     year={2024},
     eprint={2405.21060},
     archivePrefix={arXiv},
     primaryClass={cs.LG},
     url={https://arxiv.org/abs/2405.21060},
}

@article{16KARUMURI2026,
title = {Physics-informed latent neural operator for real-time predictions of time-dependent parametric PDEs},
journal = {Computer Methods in Applied Mechanics and Engineering},
volume = {450},
pages = {118599},
year = {2026},
issn = {0045-7825},
doi = {https://doi.org/10.1016/j.cma.2025.118599},
url = {https://www.sciencedirect.com/science/article/pii/S0045782525008710},
author = {Sharmila Karumuri and Lori Graham-Brady and Somdatta Goswami}
}

@article{17CHEN2025,
title = {Enforcing the principle of locality for physical simulations with neural operators},
journal = {Journal of Computational Physics},
volume = {538},
pages = {114131},
year = {2025},
issn = {0021-9991},
doi = {https://doi.org/10.1016/j.jcp.2025.114131},
url = {https://www.sciencedirect.com/science/article/pii/S0021999125004140},
author = {Jiangce Chen and Wenzhuo Xu and Zeda Xu and Noelia {Grande Gutiérrez} and Sneha Prabha Narra and Christopher McComb}
}

@article{18WANG2025,
title = {Time-marching neural operator–FE coupling: AI-accelerated physics modeling},
journal = {Computer Methods in Applied Mechanics and Engineering},
volume = {446},
pages = {118319},
year = {2025},
issn = {0045-7825},
doi = {https://doi.org/10.1016/j.cma.2025.118319},
url = {https://www.sciencedirect.com/science/article/pii/S0045782525005912},
author = {Wei Wang and Maryam Hakimzadeh and Haihui Ruan and Somdatta Goswami}
}

@article{19CHO2026,
title = {Learning time-dependent PDE via graph neural networks and deep operator network for robust accuracy on irregular grids},
journal = {Journal of Computational Physics},
volume = {544},
pages = {114430},
year = {2026},
issn = {0021-9991},
doi = {https://doi.org/10.1016/j.jcp.2025.114430},
url = {https://www.sciencedirect.com/science/article/pii/S0021999125007120},
author = {Sung Woong Cho and Jae Yong Lee and Hyung Ju Hwang}
}

@article{20CHEN2025,
title = {Tensor decomposition-based neural operator with dynamic mode decomposition for parameterized time-dependent problems},
journal = {Journal of Computational Physics},
volume = {533},
pages = {113996},
year = {2025},
issn = {0021-9991},
doi = {https://doi.org/10.1016/j.jcp.2025.113996},
url = {https://www.sciencedirect.com/science/article/pii/S0021999125002797},
author = {Yuanhong Chen and Yifan Lin and Xiang Sun and Chunxin Yuan and Zhen Gao}
}

@InProceedings{21wang2025,
author={Wang, Tian
and Wang, Chuang},
editor={Mahmud, Mufti
and Doborjeh, Maryam
and Wong, Kevin
and Leung, Andrew Chi Sing
and Doborjeh, Zohreh
and Tanveer, M.},
title={Latent Neural Operator Pretraining for Solving Time-Dependent PDEs},
booktitle={Neural Information Processing},
year={2025},
publisher={Springer Nature Singapore},
address={Singapore},
pages={163--178}
}

@article{ding2026,
title = {Physics-informed hierarchical neural operator for solving inverse problem of unsteady heat conduction},
journal = {International Journal of Heat and Mass Transfer},
volume = {258},
pages = {128335},
year = {2026},
issn = {0017-9310},
doi = {https://doi.org/10.1016/j.ijheatmasstransfer.2026.128335},
url = {https://www.sciencedirect.com/science/article/pii/S0017931026000116},
author = {Shan Ding and Yongfu Tian and Lang Qin and Hongxiang Ma and Rui Yang}
}

@article{23Lu2021,
	author = {Lu, Lu and Jin, Pengzhan and Pang, Guofei and Zhang, Zhongqiang and Karniadakis, George Em},
	title = {Learning nonlinear operators via DeepONet based on the universal approximation theorem of operators},
	journal = {Nature Machine Intelligence},
	volume = {3},
	number = {3},
	pages = {218-229},
	ISSN = {2522-5839},
	DOI = {10.1038/s42256-021-00302-5},
	url = {https://doi.org/10.1038/s42256-021-00302-5},
	year = {2021},
	type = {Journal Article}
}

@article{ABUEIDDA2025,
title = {DeepOKAN: Deep operator network based on Kolmogorov Arnold networks for mechanics problems},
journal = {Computer Methods in Applied Mechanics and Engineering},
volume = {436},
pages = {117699},
year = {2025},
issn = {0045-7825},
doi = {https://doi.org/10.1016/j.cma.2024.117699},
url = {https://www.sciencedirect.com/science/article/pii/S0045782524009538},
author = {Diab W. Abueidda and Panos Pantidis and Mostafa E. Mobasher}
}

@misc{li2021fourier,
      title={Fourier Neural Operator for Parametric Partial Differential Equations}, 
      author={Zongyi Li and Nikola Kovachki and Kamyar Azizzadenesheli and Burigede Liu and Kaushik Bhattacharya and Andrew Stuart and Anima Anandkumar},
      year={2021},
      eprint={2010.08895},
      archivePrefix={arXiv},
      primaryClass={cs.LG},
      url={https://arxiv.org/abs/2010.08895}, 
}

@misc{li2023physics,
      title={Physics-Informed Neural Operator for Learning Partial Differential Equations}, 
      author={Zongyi Li and Hongkai Zheng and Nikola Kovachki and David Jin and Haoxuan Chen and Burigede Liu and Kamyar Azizzadenesheli and Anima Anandkumar},
      year={2023},
      eprint={2111.03794},
      archivePrefix={arXiv},
      primaryClass={cs.LG},
      url={https://arxiv.org/abs/2111.03794}, 
}

@misc{Bogdan2023convolutional,
      title={Convolutional Neural Operators for robust and accurate learning of PDEs}, 
      author={Bogdan Raonić and Roberto Molinaro and Tim De Ryck and Tobias Rohner and Francesca Bartolucci and Rima Alaifari and Siddhartha Mishra and Emmanuel de Bézenac},
      year={2023},
      eprint={2302.01178},
      archivePrefix={arXiv},
      primaryClass={cs.LG},
      url={https://arxiv.org/abs/2302.01178}, 
}

@misc{li2020graphkernel,
      title={Neural Operator: Graph Kernel Network for Partial Differential Equations}, 
      author={Zongyi Li and Nikola Kovachki and Kamyar Azizzadenesheli and Burigede Liu and Kaushik Bhattacharya and Andrew Stuart and Anima Anandkumar},
      year={2020},
      eprint={2003.03485},
      archivePrefix={arXiv},
      primaryClass={cs.LG},
      url={https://arxiv.org/abs/2003.03485}, 
}

@misc{zhou2021transformer,
      title={Informer: Beyond Efficient Transformer for Long Sequence Time-Series Forecasting}, 
      author={Haoyi Zhou and Shanghang Zhang and Jieqi Peng and Shuai Zhang and Jianxin Li and Hui Xiong and Wancai Zhang},
      year={2021},
      eprint={2012.07436},
      archivePrefix={arXiv},
      primaryClass={cs.LG},
      url={https://arxiv.org/abs/2012.07436}, 
}

@misc{chung2014,
      title={Empirical Evaluation of Gated Recurrent Neural Networks on Sequence Modeling}, 
      author={Junyoung Chung and Caglar Gulcehre and KyungHyun Cho and Yoshua Bengio},
      year={2014},
      eprint={1412.3555},
      archivePrefix={arXiv},
      primaryClass={cs.NE},
      url={https://arxiv.org/abs/1412.3555}, 
}

@article{1995Universal,
  title={Universal approximation to nonlinear operators by neural networks with arbitrary activation functions and its application to dynamical systems},
  author={ Chen, Tianping  and  Chen, Hong },
  journal={IEEE Trans Neural Netw},
  volume={6},
  number={4},
  pages={911-917},
  year={1995},
}

@article{2021wangpideeponet,
author = {Sifan Wang  and Hanwen Wang  and Paris Perdikaris },
title = {Learning the solution operator of parametric partial differential equations with physics-informed DeepONets},
journal = {Science Advances},
volume = {7},
number = {40},
pages = {eabi8605},
year = {2021},
doi = {10.1126/sciadv.abi8605},
URL = {https://www.science.org/doi/abs/10.1126/sciadv.abi8605},
eprint = {https://www.science.org/doi/pdf/10.1126/sciadv.abi8605}}

@article{Cybenko1989Approximation,
	author = {Cybenko, G.},
	title = {Approximation by superpositions of a sigmoidal function},
	journal = {Mathematics of Control, Signals and Systems},
	volume = {2},
	number = {4},
	pages = {303-314},
	ISSN = {1435-568X},
	DOI = {10.1007/BF02551274},
	url = {https://doi.org/10.1007/BF02551274},
	year = {1989},
	type = {Journal Article}
}

@book{Rudin1987Real,
	title={Real and Complex Analysis},
	author={Rudin, Walter},
	edition={3rd},
	year={1987},
	publisher={McGraw-Hill},
	address={New York}
}
\end{document}